\documentclass[a4paper,12pt]{article}

\usepackage{jheppub} 

\usepackage{inconsolata}
\usepackage{graphicx, color} 
\usepackage{empheq}
\usepackage{tabularx,booktabs,array,makecell}
\usepackage{amsmath,amssymb}
\usepackage{epsf}
\usepackage{color}
\usepackage[italicdiff]{physics}
\usepackage{bm}
\usepackage{mathtools}
\usepackage{enumitem}
\usepackage{tensor}
\usepackage{wrapfig}
\usepackage{footnote}
\usepackage{here, comment}
\usepackage{bbm}
\usepackage[normalem]{ulem}

\usepackage{here, comment}
\newcommand{\red}[1]{\textcolor{red}{#1}}
\newcommand{\green}[1]{\textcolor{green}{#1}}

\usepackage[colorinlistoftodos]{todonotes}
\usepackage{tcolorbox}

\usepackage{mdframed}

\usepackage{amsthm}
\usepackage{thm-restate}

\usepackage{tikz}
\usepackage{xcolor}
\usetikzlibrary{hobby}
\usepackage{pgfplots}

\definecolor{beige}{rgb}{0.95, 0.85, 0.6} 
\definecolor{purple}{HTML}{7600bc}
\definecolor{orange}{HTML}{FF8C00}
\definecolor{darkred}{RGB}{139,0,0}

\definecolor{darkgreen}{rgb}{0.0, 0.5, 0.0}

\newcommand{\Area}[0]{\text{Area}}
\newcommand{\Int}[0]{\operatorname{Int}}

\theoremstyle{definition} 
 \newtheorem{assumption}{Assumption}[section] 
\newtheorem{lemma}{Lemma}[section] 
\newtheorem{corollary}{Corollary}[section] 
\newtheorem{theorem}{Theorem}[section] 
\newtheorem{definition}{Definition}[section]
\usepackage[nameinlink]{cleveref} 
\crefname{assumption}{Assumption}{Assumptions}
\crefname{proposition}{Proposition}{Propositions} 
\newtheorem{remark}{Remark}[section]

\usepackage[T1]{fontenc}
\usepackage[utf8]{inputenc}
\usepackage{xcolor}
\usepackage{inconsolata} 
\usepackage{listings}

\lstdefinelanguage{Mathematica}{
  morekeywords={
    ClearAll,And,Thread,Abs,Resolve,ForAll,Implies,Reals,Reduce,Not
  },
  sensitive=true,
  morecomment=[s]{(*}{*)},   
  morestring=[b]"            
}

\title{\boldmath Zoo of Correlation Inequalities in Holography and Beyond}

\author[a, b]{Kyan Louisia}
\author[b,c]{Takato Mori}
\author[b]{Herbie Warner}

\affiliation[a]{King's College London, Strand, London, WC2R 2LS, United Kingdom}
\affiliation[b]{Perimeter Institute for Theoretical Physics, Waterloo, Ontario N2L 2Y5, Canada}
\affiliation[c]{Department of Physics, Rikkyo University, 3-34-1 Nishi-Ikebukuro, Toshima-ku, Tokyo 171-8501, Japan}

\emailAdd{kyan.louisia@gmail.com}
\emailAdd{takato.mori@yukawa.kyoto-u.ac.jp}
\emailAdd{herbie.warner1@outlook.com}

\abstract{
Information-theoretic inequalities often impose nontrivial constraints on holographic states. In this work, we study measurement-based classical and quantum correlations in holography, focusing on the proposed duals of classical correlation $J_W$, quantum discord $D_W$, and one-shot distillable entanglement $E_D$, defined in terms of the entanglement wedge cross section (EWCS). 
We develop a homological framework tailored to inequalities involving multiple EWCSs and Ryu–Takayanagi surfaces, and use it to prove a family of inequalities, including monotonicity and monogamy/polygamy-type relations, as well as one-way strong superadditivity. For strong superadditivity, we additionally confirm its two-way version using Haar random states. We also examine holography-inspired boundary duals in terms of the reflected entropy and provide proofs and counterexamples for their information-theoretic inequalities. 
Taken together, our results provide further evidence for the duality between the EWCS and its proposed boundary counterparts---measurement-based correlations and one-shot distillable entanglement---while also furnishing a unified, rigorous method for proving multi-EWCS inequalities.

}

\begin{document} 
\begin{flushright}
RUP-25-24
\end{flushright}
\maketitle
\flushbottom

\newpage
\section{Introduction}
Quantifying correlations between subsystems is a central objective in both quantum information theory (QI) and quantum gravity (QG). Over the past two decades, this intersection has flourished, particularly through the AdS/CFT correspondence and the black hole information paradox~\cite{Ryu:2006bv,Maldacena:2013xja,MALDACENA_2000,Pastawski:2015qua,Hayden:2016cfa,Cotler:2017erl,Hayden:2007cs}. A landmark result was the Ryu-Takayanagi (RT) formula~\cite{Ryu:2006bv}, which equates the entanglement entropy of a boundary region to the area of a minimal surface in the bulk
called the holographic entanglement entropy (HEE). This relationship showed that entanglement entropy is a key aspect of holographic theories for the emergence of bulk geometry. 

While entanglement entropy fully characterizes correlations in pure states, its extension to mixed states is more subtle. 
In holography, a natural extension of the RT formula to mixed states is the entanglement wedge cross section (EWCS)~\cite{Takayanagi:2017knl}. Defined as the minimal-area surface that splits the entanglement wedge of a bipartite state $\rho_{AB}$ into regions homologous to $A$ and $B$, the EWCS is expected to capture some mixed-state correlation structure in a geometric way, extending the logic of the RT formula.

The EWCS has several proposed boundary duals; many works have connected it to mixed-state entanglement measures involving optimization, such as the entanglement of purification~\cite{Takayanagi:2017knl} and the entanglement of formation~\cite{Mori:2024gwe}. Later, quantities without optimization such as reflected entropy $S_R$~\cite{Dutta:2019geu}, odd entropy~\cite{Tamaoka:2018neu}, and balanced partial entanglement~\cite{Wen:2021qgx} have also been identified as possible boundary duals. However, information theoretic aspects of these holography-inspired quantities remain unclear (although there has been some analysis in monotonicity of reflected entropy~\cite{Hayden:2023yij} and some indicative results with tripartite entanglement~\cite{Hayden:2021gnu,Zou:2020bly}).

In this work, we extend the program of the duality between entanglement and geometry to more general correlations, following~\cite{Mori:2025gqe}. In particular, we focus on classical and quantum correlations such as quantum discord \cite{PhysRevLett.88.017901} and investigate their behavior in both holographic and non-holographic systems.\\

Although entanglement is often identified with quantum correlations, it is only one aspect of ``quantumness''. 
In fact, quantum correlations can exist even in separable states: for instance, the two-qubit state
\begin{equation}
	\rho_{AB} = \frac{\dyad{00}+\dyad{1+}}{2},\quad \ket{+} = \frac{\ket{0}+\ket{1}}{\sqrt{2}} \label{eq:discordant}
\end{equation}
is separable but exhibits correlations beyond classical ones. By contrast, classically correlated states are separable taking the form
\begin{equation}
	\rho^{\mathrm{CC}}_{AB} = \sum_\alpha p_\alpha \dyad{\varphi_\alpha}_A \otimes \dyad{\chi_\alpha}_B,
\end{equation}
where $\{\ket{\varphi_\alpha}\}, \{\ket{\chi_\alpha}\}$ form orthonormal bases of the respective Hilbert spaces. These correlations are entirely classical and encoded in a single probability distribution $p_\alpha$ and associated orthogonal bases on $A$ and $B$.

This type of quantum correlation beyond entanglement is captured by quantum discord \cite{PhysRevLett.88.017901} as we will review later in this paper. While the primary focus of this paper is information theoretic properties of this quantity in AdS/CFT and its boundary dual, it is worth noting that there is a physical motivation to study discord aside from its entanglement properties. For example in scenarios involving decoherence, including thermalization or the infall of measurement devices into black holes or inflating cosmologies, discord reveals aspects of quantum correlations that entanglement fails to track~\cite{Nambu:2011ae,Matsumura:2020uyg,Mori:2025gqe}.\\

Recently, one of the authors proposed bulk duals for classical and quantum correlations~\cite{Mori:2024gwe}, where classical correlation is defined by locally accessible information and quantum correlation is defined by discord. Both are conjectured to admit bulk duals in terms of the EWCS and RT formula. Moreover the study indicates that there exists both classical and quantum correlation in holography, with the latter always exceeding entanglement. The difference between holographic discord and entanglement turns out to be related to distillable entanglement and tripartite entanglement.
In addition, while the reflected entropy itself is not a correlation measure~\cite{Hayden:2023yij}, it was shown that it composes the classical and quantum correlation-like quantities when it is combined with entropies.
This motivates us to ask: is there any characteristic behavior of holographic correlation quantities? Do these obey some basic information theoretic inequalities like monotonicity and monogamy?

Monotonicity implies that a correlation measure should not increase under local operations such as partial trace. This property is a defining feature for a resource theory of correlations. 
It enforces the second law-like structure and governs the convertibility and irreversibility. 
In holography, monotonicity is also understood geometrically: for example, monotonicity of mutual information and EWCS follows from entanglement wedge nesting~\cite{Headrick:2013zda,Takayanagi:2017knl,Akers:2017ttv}, stemming from bulk locality and reconstructability~\cite{Czech:2012bh,Wall:2012uf,Almheiri:2014lwa,Jafferis:2015del,Dong:2016eik}.

Monogamy, by contrast, limits how quantum correlations can be shared. If $A$ is maximally entangled with $B$, it cannot also be entangled with $C$. This principle, tied to the no-cloning theorem, highly constrains quantum theory, e.g. it underpins security in quantum cryptography~\cite{Tomamichel_2013,Khanian:2025twm} and is also fundamental to the black hole information paradox~\cite{Hayden:2007cs}.

Holographic states display distinctive correlation structures. A well-known example is the monogamy of mutual information~\cite{Hayden:2011ag}. Another characteristic feature is that they carry substantial multipartite entanglement while supporting very little classical correlation, which motivates our investigation of further inequalities constraining classical and quantum correlations.

These constraints have geometric roots. Our holographic arguments rely on structures like the homology condition, extremality of HEE, and entanglement wedge nesting, which make inequalities such as subadditivity and strong subadditivity manifest geometrically~\cite{Headrick:2007km}. Beyond these, monogamy of mutual information has been formalized in the holographic entropy cone program~\cite{Bao:2015bfa}, which attempts a systematic classification of all entropy inequalities consistent with the bulk geometry.

However, inequalities involving multiple EWCSs remain less understood. While the EWCS has been linked to quantities such as the entanglement of formation/purification, the reflected entropy, the Markov gap, and the distillable entanglement~\cite{Rains:1998gp}, only a few inequalities have been confirmed. Some, like the non-negativity of the Markov gap~\cite{Hayden:2021gnu} and the vanishing of distillable entanglement~\cite{Mori:2024gwe}, reinforce the idea that holographic states favor multipartite quantum correlations. But generalizations involving multiple EWCS terms are more involved~\cite{Bao:2021vyq, Jain:2022csf}; it is challenging to prove them rigorously due to lack of systematic mathematical tools. This suggests a more delicate structure than for entropy-based measures.

In this work, we investigate this structure by analyzing a family of inequalities involving classical and quantum correlations in holography and their boundary versions. In Section \ref{Section 2}, the definitions of several correlation measures are given both on the boundary and in the bulk. We also give definitions of correlation inequalities here. In Section \ref{sec:assump-tech}, we list our assumptions, introduce some of the technical tools required for proofs, and discuss some topological properties of the EWCS. In Section \ref{Section 3}, we study monotonicity for both correlation types, and in Section \ref{Section 4}, we establish their monogamy and polygamy properties. In Section \ref{Section 5}, we further examine the strong superadditivity of holographic classical correlation, which is conjectured to be dual to distillable entanglement, using geometric arguments and Haar random analysis. 
In Section \ref{Section 6}, motivated from a holographic duality, we map the bulk quantities to boundary ones again by replacing the EWCS with a half of the reflected entropy, defining so-called reflected measures --- analytic and computationally efficient proxies for classical and quantum correlations. We examine their properties like monotonicity and monogamy with various few-qubit counterexamples. See Tables~\ref{tab:cc_ineq} and \ref{tab:qd_ineq} for a summary of the results. Due to the lengths of proofs, we only provide sketches of each in the main body and relegate the complete proofs to the Appendices~\ref{app:code}-\ref{SSA formal proof}. See Appendix~\ref{s: summary of notation} for a summary of notations used throughout.

\section{Classical and Quantum Correlations}
\label{Section 2}

In this section, we review the correlation measures of interest in this work. These include the classical correlation $J(A|B)$ and quantum discord $D(A|B)$, as well as their conjectured gravity duals $J_W(A|B)$ and $D_W(A|B)$ which are defined in terms of the entanglement wedge cross section (EWCS). We also introduce the boundary duals of these geometric quantities based on the reflected entropy, which provide computable, optimisation-free entanglement measures. We then summarize the correlation inequalities of interest such as monotonicity, monogamy and strong superadditivity. Our ultimate goal is to study whether these properties are universal across quantities on both sides of the bulk/boundary duality. 

\subsection{Classical and Quantum Correlation Measures}

First we introduce the three key correlation measures relevant to our analysis: total correlation or mutual information \( I(A:B) \), classical correlation \( J(A|B) \), and quantum discord \( D(A|B) \). 

\paragraph{Mutual Information.}
The (quantum) mutual information \( I(A:B) \) quantifies the total correlation between two subsystems \( A \) and \( B \) of a bipartite state \( \rho_{AB} \), defined as
\begin{equation}
    I(A:B) \vcentcolon= S_A + S_B - S_{AB}, \label{mutual_info_2}
\end{equation}
where \( S_X \) is the von Neumann entropy of the reduced density matrix on subsystem \( X \). In terms of the quantum conditional entropy \( S(A|B) = S_{AB} - S_B \), the mutual information becomes
\begin{equation}
    I(A:B) = S_A - S(A|B).
\end{equation}
Intuitively, this represents a residual entropy ($\sim$uncertainty) of $A$ after quantum conditioning on $B$.

\paragraph{Classical Correlation.}
A different form of mutual correlation can be defined by measuring a subsystem \( B \). Given a positive operator-valued measure (POVM) \( \Pi = \{E_x\}_x \) on \( B \), a conditional state of \( A \) after observing outcome \( x \) is \( \rho_A^x = \frac{1}{p_x} \Tr_B(E_x \rho) \), with probability \( p_x = \Tr(E_x \rho) \). 
The (classical) mutual information $J_\Pi(A|B)$ can be obtained by replacing the quantum conditional entropy in~\eqref{mutual_info_2} with the (classical) conditional entropy, defined as the average residual entropy of $A$ after measuring $B$ with $\Pi$:
\begin{equation}
    J_\Pi(A|B) = S_A - \sum_x p_x S(\rho_A^x).
    \label{eq:clas-cond-ent}
\end{equation}
Intuitively, this quantifies the reduction in uncertainty about \( A \) due to a measurement on \( B \). In this sense, it is occasionally called as the Groenewold-Ozawa information gain or the locally accessible information \cite{Ozawa1986InformationGain}. The classical correlation is defined by maximizing $J_\Pi$ over all possible POVM measurements:
\begin{equation}
    J(A|B) = \max_\Pi J_\Pi(A|B).
    \label{eq:clas-corr}
\end{equation}
It coincides with \( I(A:B) \) only when \( \rho_{AB} \) is quantum-classical (i.e., classical on \( B \)). On the one hand, \( J(A|B) \) can be nonzero even in the absence of entanglement and on the other, classical correlation in this definition may include contributions from entanglement. For example, for a maximally entangled Bell pair, one finds \( J(A|B) = \log 2 \). For a classically correlated state, whose information can be reduced to a single probability distribution like $\rho_{AB}= \sum_i p_i \dyad{i}_A\otimes\dyad{\varphi_i}_B$, one finds $J(A|B)$ equals the Shannon entropy, representing the entropy of classically correlated bits.

The Koashi-Winter relation~\cite{Koashi:2003pgf} provides an alternative expression for the classical correlation in terms of the entanglement of formation $E_F$ between \( A \) and a purification partner \( C \) of \( \rho_{AB} \). Namely,
\begin{equation}
    J(A|B) = S_A - E_F(A:C),
\end{equation}
where \( \rho_{AB} = \Tr_C \dyad{\Psi}_{ABC} \) for arbitrary purification \( \ket{\Psi}_{ABC} \) and the Entanglement of formation is defined as
\begin{equation}
    E_F(A:C) = \inf \sum_x p_x S_A(\dyad{\psi_x}),
    \label{eq:EoF}
\end{equation}
where the infimum is taken over all purifications of \( \rho_{AC} = \sum_x p_x \dyad{\psi_x}_{AC} \). 

\paragraph{Quantum Discord.}
We define the quantum correlation between \( A \) and \( B \) as the difference between the total and classical correlation:
\begin{equation}
    D(A|B) = I(A:B) - J(A|B),
    \label{eq:discord}
\end{equation}
which is known as the Quantum Discord~\cite{PhysRevLett.88.017901}. It captures nonclassical correlations that survive even in separable states, and like the classical correlation it is in general asymmetric with respect to which party is being measured.

We note that, by construction, both \( J(A|B) \) and \( D(A|B) \) involve an optimization over all POVM measurements. While the optimal POVM is known to be rank-one~\cite{datta2008studies}, it is not necessarily orthogonal, and there are examples where projective measurements are suboptimal~\cite{Chen_2011,Galve_2011}. Computing discord is NP-complete in general~\cite{Huang_2014}, making it infeasible to evaluate for high-dimensional systems. 

In this paper, we address this challenge by exploring bulk and boundary duals of these measures. These duals allow for efficient computation of classical and quantum correlation quantities, even in a high-dimensional Hilbert space.

\subsection{Holographic duals}
In this subsection, we introduce the holographic counterparts of the classical correlation \( J(A|B) \) and quantum discord \( D(A|B) \) at leading order in $G_N$ and in static (time-independent) spacetimes.

\paragraph{Holographic Entanglement Entropy.}
In the static (or more generally time reflection symmetric) cases, the entanglement entropy of a subregion $A$ in a boundary theory is given by the Ryu-Takayanagi formula~\cite{Ryu:2006bv}. The bulk dual is called the holographic entanglement entropy (HEE), which is proportional to the minimal area of a codimension-two\footnote{We will always use codimension with respect to the full spacetime (not just the Cauchy slice).} surface homologous to $A$:
\begin{equation}
    S_A = \min_{\gamma_A} \frac{\mathrm{Area}(\gamma_A)}{4G_N}.
\end{equation}

\paragraph{Entanglement Wedge Cross Section.}
The EWCS between $A$ and $C$ is proportional to the minimal area of a codimension-two surface separating the entanglement wedge of \( AC \) (denoted by \( \mathcal{E}({AC}) \)) into two bulk subregions homologous to \( A \) and \( C \). Namely,
\begin{equation}
	E_W(A:C) = \min_{\Gamma_{A:C}} \frac{\mathrm{Area}(\Gamma_{A:C})}{4G_N},
    \label{eq:EWCS}
\end{equation}
where the minimization is over all codimension-two surfaces \( \Gamma_{A:C} \subset \mathcal{E}_{AC} \) that bisect the entanglement wedge.

The EWCS is conjectured to be dual to several different boundary quantities related to mixed-state entanglement, including the entanglement of formation $E_F$, entanglement of purification $E_P$, and a half of reflected entropy $S_R$~\cite{Dutta:2019geu, Bao:2017nhh, Takayanagi:2017knl}.

\paragraph{Bulk duals for classical and quantum correlations.}
A recent proposal by one of the authors~\cite{Mori:2025gqe} conjectures that the classical and quantum correlations admit bulk duals in holography. Given a purification \( \ket{\Psi}_{ABC} \) of a bipartite mixed state \( \rho_{AB} = \Tr_C \dyad{\Psi} \), the bulk duals of classical correlation $J(A|B)$ and quantum discord $D(A|B)$ are proposed as
\begin{equation}
	J_W(A|B) \equiv S_A - E_W(A:C), \quad D_W(A|B) \equiv S_B - S_C + E_W(A:C). \label{Holographic definitions}
\end{equation}
This is supported either by holographic measurements or the $E_F=E_W$ conjecture, utilizing a disentangled basis along the minimal surface for $B$. We refer the readers to~\cite{Mori:2024gwe} for more details.

\paragraph{Reflected Entropy.}
The reflected entropy is defined as
\begin{equation}
    S_R(A:B) = S_{AA^\ast}(\ket{\rho^{1/2}}),
    \label{eq:reflected entropy}
\end{equation}
where \( \ket{\rho^{1/2}}_{AA^\ast BB^\ast} \) is the \textit{canonical purification} of \( \rho_{AB} \), which is found using the following procedure. Let $\{\lambda_n,\ket{n}\}_n$ be the set of eigenvalues and eigenvectors of the density matrix $\rho_{AB}$. Here we include zero eigenvalues as well.
Then, the canonical purification is given by
\begin{equation}
    \ket{\rho^{1/2}}_{AA^*BB^*} = \sum_n \sqrt{\lambda_n} \ket{n}_{AB} \otimes \ket{n}_{A^*B^*},
    \label{eq:cano-purif}
\end{equation}
where $\mathcal{H}_{AB}\rightarrow\mathcal{H}_{A^*B^*}$ is an endomorphism.

In this work, we also aim to propose an alternative measure to classical and quantum correlations without optimizations. Employing a particular boundary dual of the EWCS, we aim to deduce a boundary quantity out of the holographic measures. Since our aim is to circumvent a difficult optimization, we employ the duality between the EWCS and the reflected entropy,
\begin{equation}\label{eq: ew = 2sr}
    S_R(A:B) = 2E_W(A:B),
\end{equation}
which holds at $\mathcal{O}(1/G_N)$ \cite{Dutta:2019geu}. The subleading correction could be expected depending on the quantum definition of the EWCS.

\paragraph{One-shot distillable entanglement.}
Holographic classical correlation $J_W(A|B)$ (as well as its analog of Haar random states) is proposed to admit another boundary interpretation --- one-shot, one-way distillable entanglement $E_D^{[\text{1WAY}]}(A|B)$~\cite{Mori:2024gwe}.
It is defined as the number of EPR pairs that can be extracted from a bipartite state via one-way local operations and classical communication (LOCC) up to a sufficiently small error (See~\cite{Mori:2024gwe} for a precise definition). Symmetrizing this definition leads to the one-shot distillable entanglement under one-way LOCC in both directions:
\begin{equation}
    E_D^{[\text{1WAY}]}(A:B) = \max(J_W(A|B), J_W(B|A))=:J_W(A:B),
    \label{eq:1shot-ED}
\end{equation}
An important question is whether one-way LOCC outperform two-way LOCC or not. 
In general, two-way LOCC is a larger set of operations than one-way LOCC in both directions as the latter does not feedforward the measurement outcome to the subsequent LOCC in the other direction. While we do not provide a direct answer to this question, we give  supporting evidence by considering the strong superadditivity for both one-way and two-way distillable entanglement.

\subsection{Correlation Inequalities}\label{sec:mono}
We now define the functional properties we examine in this work. 
Let $\sigma(A|B)$ be any (potentially asymmetric) bipartite correlation measure between arbitrary subsystems $A$ and $B$. We denote other subsystems by $C,D$. Note that $\rho_{ABCD}$ is not necessarily pure.

\begin{definition}[Monotonicity]
A measure $\sigma(A|B)$ is \emph{monotone} under partial trace with respect to the measured party if
\begin{align}
    \sigma(A|BC) \ge \sigma(A|B)
\end{align}
and monotone under partial trace with respect to the unmeasured party if
\begin{align}
    \sigma(AC|B) \ge \sigma(A|B).
\end{align}
\end{definition}

\begin{definition}[Monogamy]
A measure $\sigma(A|B)$ is \emph{monogamous} with respect to the measured party if
\begin{equation}
    \sigma(A|B) + \sigma(A|C) \le \sigma(A|BC),
\end{equation}
and monogamous with respect to the unmeasured party if
\begin{equation}
    \sigma(A|B) + \sigma(C|B) \le \sigma(AC|B).
\end{equation}
We call $\sigma$ \emph{polygamous} if the opposite inequality is true. 
\end{definition}

\begin{definition}[Strong superadditivity]
A measure $\sigma(A|B)$ is \emph{strongly superadditive} if
\begin{equation}
    \sigma(AC|BD) \ge \sigma(A|B)+\sigma(C|D).
    \label{eq:strong-sa}
\end{equation}
The strong superadditivity follows from the inclusion relation of operations for optimized quantities. It implies monogamy by taking $A=C$ or $B=D$. 
\end{definition}

While entanglement measures are often monogamous and monotone, the same is not true for classical correlation or quantum discord. We illustrate this with an explicit counterexample. Consider a tripartite state
\begin{equation}
    \rho_{ABC} = \frac{1}{2} \dyad{000} + \frac{1}{2} \dyad{+11},
\end{equation}
where $\ket{+} = (\ket{0} + \ket{1})/\sqrt{2}$. Its marginals are
\begin{equation}
    \rho_{BA} = \rho_{CA} = \frac{1}{2} \dyad{00} + \frac{1}{2} \dyad{1+}, \qquad \rho_{BC} = \frac{1}{2} \dyad{00} + \frac{1}{2} \dyad{11}.
\end{equation}
One finds that $D(C|AB) = 0$ but $D(C|A) > 0$, violating monotonicity with respect to the measured (second) party. Furthermore, $D(C|A)=D(B|A)=D(BC|A)>0$ while $D(C|B) = D(C|AB) = 0$, showing discord is not generally monogamous for either party.

Although quantum discord is known to be monotone with respect to the unmeasured (first) party and classical correlation is monotone in both arguments~\cite{Vedral_2003,Streltsov_2011}, monogamy is not guaranteed~\cite{PhysRevA.85.040102}. This means strong superadditivity is neither guaranteed.
Holographic duals, however, may obey stronger constraints.\footnote{For example, mutual information is known to be monogamous for holographic states while it is not necessarily monogamous for non-holographic states~\cite{Hayden:2011ag}.} (Note the above counterexample is a classically correlated state; such a state is unlikely holographic~\cite{Susskind:2014yaa, Nezami:2016zni, Dong:2021clv,Mori:2024gwe,Mori:2025gqe}.) Testing these inequalities provides not only a consistency check for the proposed dualities but also additional characteristic features of holographic states.

\section{Assumptions and Technical Tools}\label{sec:assump-tech}
In this section, we summarize the assumptions and definitions we use in our proofs. In particular, we define various notations of AdS/CFT topologically in Section~\ref{sec:conv-def} and demonstrate a proof of a holographic inequality in Section~\ref{sss: structure of proofs} with our notations. In Section~\ref{sec:topo-ewcs} we provide a topological characterization of the EWCS by providing some key definitions and properties. See Appendix \ref{s: summary of notation} for a summary of notation and Appendix \ref{s: complete proofs} for additional notation used in the complete proofs. 

\subsection{Assumptions for holographic proofs}
Before turning to the holographic proofs of the aforementioned functional properties, we briefly summarize the key assumptions in this work.

\begin{assumption} [Leading-order analysis]
    In this paper, we always work with the semiclassical limit where Newton's constant $G_N\rightarrow 0$. In this limit, entropic quantities like $I$, $J_W$, and $D_W$ are decomposed into the area term proportional to $G_N^{-1}$ and the bulk matter term which is usually $O(1)$. We do not consider any backreaction from the bulk matter or nonperturbative corrections in this paper.
\end{assumption}

\begin{assumption}[Time-reflection symmetry]\label{assump:sym}
    Generally speaking, all the entropic quantities can be defined in a time-dependent setting using the Hubeny-Rangamani-Takayanagi (HRT) formula~\cite{Hubeny:2007xt} or the maximin prescription~\cite{Wall:2012uf}. However, in this paper, we restrict to the time-reflection symmetric setup, where we can rely on the Ryu-Takayanagi (RT) formula. This is because when one considers a linear combination of multiple entropic quantities (HEE and EWCS), they do not necessarily lie on the same HRT/maximin surface and a mere comparison on a single surface may not be justified.
\end{assumption}

\begin{assumption}[Minimality and Homology Condition]
    In the calculation of HEE and EWCS,
all relevant surfaces are homologous to the boundary subregions in the entanglement wedge of interest and are chosen to minimize the area functional. 
\end{assumption}

\begin{assumption}[Entanglement Wedge Nesting (EWN)]
     For nested boundary regions, their corresponding entanglement wedges also form a nested structure in the bulk. In particular, the monotonicity of the EWCS 
\begin{equation}
    E_W(A:BC)\ge E_W(A:B)
\end{equation}
 follows from EWN~\cite{Czech:2012bh,Akers:2016ugt,Wall:2012uf}. See Figure \ref{Monotonicity Measured party} for an example.
\end{assumption}

\begin{assumption}
[No multiple intersections of geodesics on an AdS Cauchy slice]\label{thm: No multiple intersections of geodesics on an AdS Cauchy slice} This is a standard property of holographic spacetimes under our other assumptions where any two distinct length minimizing geodesics on the slice intersect in the \emph{interior} of the slice at most once. If they intersect more than once, say at points $x$ and $y$, then they coincide between $x$ and $y$. 

This is of course necessary information theoretically. If the extremal surfaces associated to different boundary subregions intersect with each other, the entanglement wedge reconstruction can transmit bulk quantum information inside the intersection of the entanglement wedges to either boundary subregion simultaneously. However, this violates the no-cloning theorem in QI which prohibits cloning unknown quantum information.
\end{assumption}

\begin{assumption}[Bulk boundaries]
    In the definitions and proofs below, we assume that the asymptotic boundaries are the only boundaries of the bulk dual to a pure state. The arguments, however, extend straightforwardly to bulks with additional non-asymptotic boundaries, such as cutoff surfaces or end-of-the-world branes.
In the former case, the boundary Hilbert space is defined on both the asymptotic boundaries and any non-asymptotic boundaries. In the latter case, the brane is dynamically constrained and does not introduce independent degrees of freedom beyond those of the asymptotic boundary~\cite{Takayanagi:2011zk}; accordingly, surfaces that lie along the brane contribute no area term even though they geometrically form part of the bulk boundary.
In addition, when a non-asymptotic boundary admits a naive RT surface that would lie outside the physical bulk, the entanglement wedge must be replaced by the generalized entanglement wedge~\cite{Bousso:2022hlz}.
\end{assumption}

\begin{assumption}[Disjoint boundary subsystems]
    When proving holographic correlation inequalities for some boundary state $ABC\dots$ we will always assume each subsystem is non-overlapping unless otherwise stated. That is $A\cap B$ is empty or codimension two and so on, with this intersect only along the subsystem boundaries on the cut-off surface. 
\end{assumption}

\begin{remark}
Although we use the Poincar\'e disk picture of AdS$_3$/CFT$_2$ to illustrate the proofs, we emphasize that our methodology is general and is valid regardless of dimensionality.
\end{remark}

\subsection{Conventions and Definitions}\label{sec:conv-def}
Since our proofs will be based on topological arguments, here we define some necessary terminology.

\begin{figure}[ht]
    \centering
\scalebox{1}{
\input{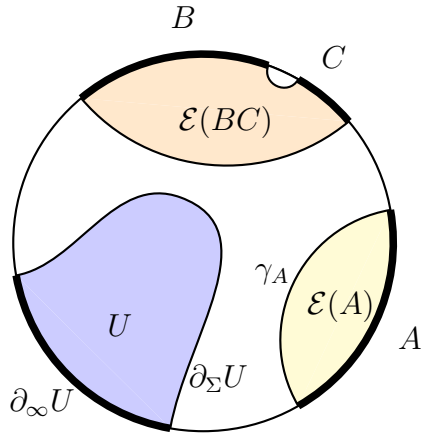}
}
    \caption{Example of definitions given below on the Poincar\'e disk.}
    \label{fig: example of definitions}
\end{figure}

\begin{definition}[Bulk]\label{def:bulk}
The \emph{bulk slice} is a fixed time-reflection symmetric Cauchy slice, denoted by \(\Sigma\).
All geometric objects (subregions, hypersurfaces) considered in this work are taken to lie in \(\Sigma\).
\end{definition}

\begin{definition}[Boundary notation]\label{def:boundary}
For a codimension-one surface $U\subseteq\Sigma$, we write \(\partial_\Sigma U\) for its boundary \emph{within the bulk slice} \(\Sigma\) (i.e.\ the bulk Cauchy slice excluding the bulk asymptotic boundary). We write \(\partial_\infty U\) for the boundary lying on the asymptotic boundary of $\Sigma$. Therefore, all the boundary subregions are a subset of $\partial_\infty\Sigma$. By default, we use a shorthand notation, \(\partial\equiv\partial_\Sigma\), unless otherwise stated.
\end{definition}

\begin{definition}[Subsystem complement]
    Given a subsystem $A$, we write $\bar{A}$ to be its complement such that $A\cup\bar{A}$ is a pure state.
\end{definition}

\begin{definition}[Entanglement wedge]\label{def:ewedge}
For a boundary subsystem \(A\), the \emph{entanglement wedge} \(\mathcal E(A)\subseteq\Sigma\) is the bulk region whose asymptotic and bulk boundaries satisfy
\[
\partial_\infty \mathcal E(A) = A, \qquad
\partial_\Sigma \mathcal E(A) = \gamma_A,
\]
where \(\gamma_A\) is the RT surface of \(A\). For two disjoint boundary subsystems \(A\) and \(B\), the \emph{joint entanglement wedge} \(\mathcal E(AB)\subseteq\Sigma\) is defined analogously using the RT surfaces \(\gamma_{AB}\) again with equality iff the state on $AB$ is pure.
\end{definition}

\begin{definition}[Bridged wedges]\label{def:bridged-wedges}
We can categorize the various types of joint entanglement wedges as follows. For two disjoint boundary subsystems \(A\) and \(B\), with
\[
A = \bigsqcup_i A_i, \qquad B = \bigsqcup_j B_j, \qquad A\cap B = \varnothing, \qquad A_i \cap A_j = A_i \delta_{ij},  \qquad B_i \cap B_j = B_i \delta_{ij}, 
\]
we can decompose \(\mathcal E(AB)\) into its connected (in the standard topological sense) components
\[
\mathcal E(AB) = \bigsqcup_\alpha W_\alpha.
\]
For each component \(W_\alpha\) consider its asymptotic boundary
\(\partial_\infty W_\alpha \subseteq A\cup B\).
We classify components as follows:
\begin{itemize}
\item \(W_\alpha\) is an \emph{\(A\)-only component} iff
      \(\partial_\infty W_\alpha \subseteq A\).
\item \(W_\alpha\) is a \emph{\(B\)-only component} iff
      \(\partial_\infty W_\alpha \subseteq B\).
\item \(W_\alpha\) is an \emph{\(A\)–\(B\) bridged component} iff its asymptotic
      boundary meets both sides,
      \[
      \partial_\infty W_\alpha \cap A \neq \varnothing,
      \qquad
      \partial_\infty W_\alpha \cap B \neq \varnothing.
      \]
\end{itemize}
The union of all \(A\)–\(B\) bridged components,
\[
\mathcal E_{\mathrm{brid}}(A{:}B)
\;\vcentcolon=\;
\bigsqcup_{\alpha:\, W_\alpha\ \text{bridged}} W_\alpha,
\]
is called the \emph{(A{:}B)-bridged entanglement wedge}.
We say that \(A\) and \(B\) have a \emph{bridged entanglement wedge} if
\(\mathcal E_{\mathrm{brid}}(A{:}B)\neq\varnothing\).\footnote{The relation to a connected (entanglement) wedge used in literature is as follows. If the entanglement wedge of two boundary subsystems (composed of multiple boundary subregions) is connected (i.e. non-factorizing), then there is at least one bridged component among the boundary subregions. Thus, a connected entanglement wedge contains a bridged entanglement wedge. However, they are not equivalent as we specifically refer the components connecting a portion of $A$ and a portion of $B$ as \emph{bridged} among various components of the connected entanglement wedge.}

Moreover we say for some $A_i\in A$ that $A_i$ is an \emph{$A$-only component of $AB$} iff there exists an $A$-only component $W_\alpha$ such that $A_i\subseteq\partial_\infty W_\alpha$. We define $B_i$ as a \emph{$B$-only component} similarly and say $A_i$ is a \emph{bridged component of $AB$} iff there exists an $A-B$ bridged component $W_\alpha$ such that $A_i\subseteq\partial_\infty W_\alpha$. See Figure \ref{fig: bridged wedges} for an example.
\end{definition}

\begin{figure}[ht]
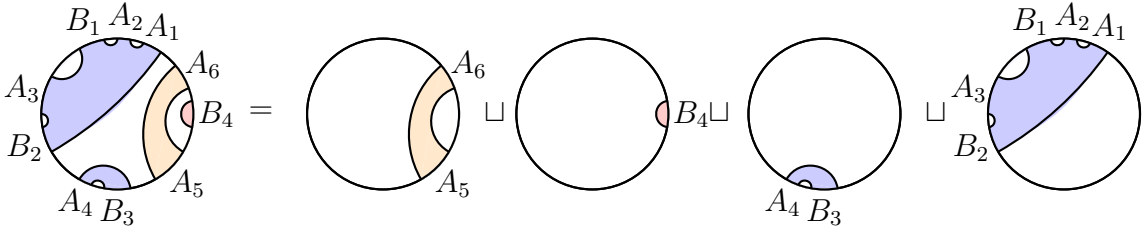

    \centering
\scalebox{1}{
\hspace{-0.5cm}
\input{Sketch_Figures/Bridged_wedges/fig1}
\hspace{-0.5cm}
\begin{tikzpicture}
   \node at (0,-0.5) {$=$};
  \node at (0,-2) { };
\end{tikzpicture}
\hspace{-0.6cm}
\input{Sketch_Figures/Bridged_wedges/fig2}
\hspace{-0.9cm}
\begin{tikzpicture}
   \node at (0,-0.5) {$\sqcup$};
  \node at (0,-2) { };
\end{tikzpicture}
\hspace{-0.9cm}
\input{Sketch_Figures/Bridged_wedges/fig3}
\hspace{-0.7cm}
\begin{tikzpicture}
   \node at (0,-0.5) {$\sqcup$};
  \node at (0,-2) { };
\end{tikzpicture}
\hspace{-0.8cm}
\input{Sketch_Figures/Bridged_wedges/fig6}
\hspace{-0.9cm}
\begin{tikzpicture}
   \node at (0,-0.5) {$\sqcup$};
  \node at (0,-2) { };
\end{tikzpicture}
\hspace{-0.5cm}
\input{Sketch_Figures/Bridged_wedges/fig5}
}
    \caption{Decomposition of $\mathcal{E}(AB)$ into its $A$-only component in orange, $B$-only component in red and a disjoint bridged entanglement wedge $\mathcal{E}_{\text{brid}}(A:B)$ in blue (disjoint as the last two panels are both bridged but disjoint from each other). We have $\{A_5,A_6\}$ are $A$-only components of $AB$, $B_4$ a $B$-only component of $AB$, and $\{A_1,A_2,A_3,A_4,B_1,B_2,B_3\}$ bridged components of $AB$.}
    \label{fig: bridged wedges}
\end{figure}

\begin{remark}[Bridged wedges and correlation measures]\label{remark: bridged wedges}
    Given some (mixed) state on  $AB$, we generally have $\mathcal E(AB) =\sqcup_\alpha W_\alpha,$ with each $W_\alpha$ some connected codimension-one surface. Identify all $A$-only components within this decomposition and label the subset $\mathcal{A}\subseteq A$ such that $\mathcal{A}$ is the full set of $A$-only components in $AB$. Define $\mathcal{B}$ similarly. We can then find the $AB$ bridged components $\mathcal{F}_A = A\setminus \mathcal{A}$ and $\mathcal{F}_B = B\setminus \mathcal{B}$ such that 
 $\mathcal{F}_A$ is the $AB$ bridged components of $A$, and $\mathcal{F}_B$ is the $AB$ bridged components of $B$. By defining $\mathcal{F}=\mathcal{F}_A\cup\mathcal{F}_B$, it then follows that
    \begin{equation}
        \mathcal{E}(AB) = \mathcal{E}(\mathcal{A}) \sqcup \mathcal{E}(\mathcal{B}) \sqcup \mathcal{E}(\mathcal{F}),
    \end{equation}
    and so entropies split additively to
    \begin{equation}
        S_{AB} = S_{\mathcal{A}} + S_\mathcal{B} + S_\mathcal{F}, \quad \text{and} \quad E_W(A:B) = E_W(\mathcal{F}_A:\mathcal{F}_B).
    \end{equation}
    For example in Figure \ref{fig: bridged wedges} we have $S_{AB} = S_{A_5A_6}+S_{B_4}+(S_{A_1A_2A_3B_1B_2}+S_{A_4B_3})$. We have two terms from the bridged component as $\mathcal{E}_{\text{brid}}(AB)$ is composed of the corresponding two disjoint components:
    \begin{equation}
    \mathcal{E}_{\text{brid}}(AB)=\mathcal{E}(\mathcal{F})=\mathcal{E}(A_1A_2A_3B_1B_2)\sqcup\mathcal{E}(A_4B_3).
    \end{equation}

    Due to EWN, $\mathcal{E}(A)\subset\mathcal{E}(AB)$ when $B$ is non-empty, thus $A$-only components in $AB$ disconnect to $AB$-bridged components. 
   
    This leads to
    \begin{equation}
        S_A = S_\mathcal{A} + S_{\mathcal{F}_A}.
    \end{equation}
    That is, when only considering $\mathcal{E}(A)$, $\mathcal{A}$ and $\mathcal{F}_A$ remain unbridged in $A$. The same argument holds for $B$.
\end{remark}

\begin{definition}[Homologous]\label{def:homology}
Let \(A\) be a boundary subsystem and let \(U\subseteq\Sigma\) be a codimension-one bulk \emph{subregion}. 
\begin{enumerate}[label=(\alph*)]
\item We write \(U\sim A\) to mean that $\partial_\infty U=A$ and $\partial_\Sigma U$ is homologous to $A$.
That is, \(U\) is a codimension-one bulk region whose only asymptotic boundary component is \(A\), and it is otherwise bounded by a closed codimension two bulk surface in \(\Sigma\). 

\item We write \(U\sim_V A\), with $V\subseteq\Sigma$ a codimension-one bulk subregion, to mean the above definition  but the replacement of $\Sigma$ to $V$. 

\item For a state on \(AB\), we write \(U\sim_{AB} A\) to mean \(U\sim_{\mathcal E(AB)} A\), i.e.\ \(\partial_\Sigma U\) bounds \(A\) inside the entanglement wedge \(\mathcal E(AB)\).
\end{enumerate}
\label{Homology}
\end{definition}

\begin{definition}[Interior]
    We denote $\operatorname{Int}\mathcal{E}(A)$ as the open bulk interior of the wedge, i.e.\ $\mathcal{E}(A)\setminus\partial\mathcal{E}(A)$.
    \label{Interior}
\end{definition}

\subsection{Example of a holographic proof}\label{sss: structure of proofs}

Let us consider a simple inequality to demonstrate our strategy for proving a holographic inequality using the notations introduced above. We prove the following inequality for a mixed state on $AB$:
\begin{equation}\label{eq: simple inequality to prove}
    S_{AB} + E_W(A:B) \geq \text{max}(S_A,S_B),
\end{equation}
with an example configuration in Figure \ref{Fig: structure of proofs}. 

\begin{figure}[ht]
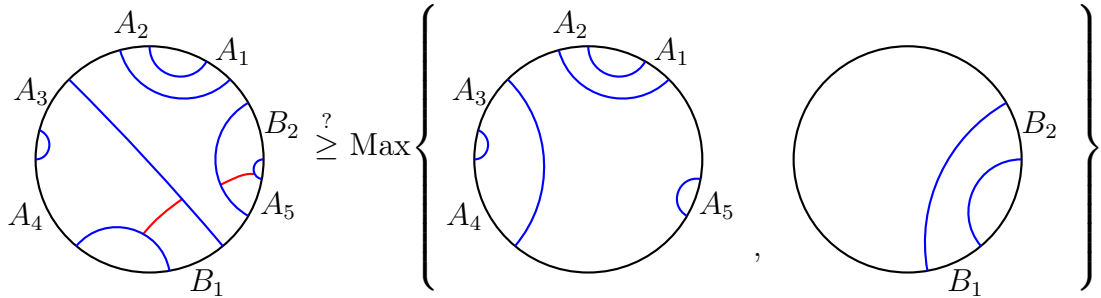

    \centering
\scalebox{1}{
\input{Sketch_Figures/Structure_of_Proofs/LHS}
\hspace{-0.5cm}
\begin{tikzpicture}
   \node at (0,-0.5) {$\stackrel{?}{\geq}$ Max$\left\{\rule{0pt}{5em}\right.$};
   \node at (0,-2.5) { };
\end{tikzpicture}
\hspace{-0.5cm}
\input{Sketch_Figures/Structure_of_Proofs/RHS}
\hspace{-0.5cm}
\begin{tikzpicture}
   \node at (0,-2) {,};
   \node at (0,-2.5) { };
\end{tikzpicture}
\hspace{-0.5cm}
\input{Sketch_Figures/Structure_of_Proofs/RHS2}
\hspace{-0.5cm}
\begin{tikzpicture}
   \node at (0,-0.5) {$\left.\rule{0pt}{5em}\right\}$};
   \node at (0,-2.5) { };
\end{tikzpicture}
}
    \caption{Holographic proof of \eqref{eq: simple inequality to prove}. On the LHS, blue curves are $\gamma_{AB}$ and red curves are $\Gamma^W_{A:B}$. On the RHS, they are $\gamma_A$ and $\gamma_B$ respectively. 
    }
    \label{Fig: structure of proofs}
\end{figure}
\begin{proof}
The proof is in two steps. First, we decompose the LHS and RHS to cancel the disjoint components between two sides. Then, we simply cut and glue various surfaces from the LHS together to serve as upper bounds for the terms on the RHS.

\paragraph{Step 1: Isolating $A$-only components.} We write $A = \cup A_i$ and $B = \cup B_i$, and denote $C$ as the purifier of $AB$. Without loss of generality we can assume $S_A \geq S_B$ as the equation is symmetric under $A\leftrightarrow B$. Generally $A$ can consist of multiple disjoint boundary subsystems and so we decompose $A=\mathcal{A}\cup\mathcal{F}$ with $\mathcal{F}_A$ the $AB$ bridged components of $A$, and $\mathcal{A}$ the $A$-only components in $AB$. For example in Figure \ref{Fig: structure of proofs} we have $\mathcal{F}_A=A_3\cup A_4 \cup A_5$ and $\mathcal{A}=A_1\cup A_2$. Then, entropies split additively to 
\begin{equation}\label{eq: entropy splitting}
S_{AB} = S_{\mathcal{A}} + S_{\mathcal{F}_A\cup B}, \quad S_{A} = S_{\mathcal{F}_A} + S_\mathcal{A},
\end{equation}
with the second equality following from the first by EWN (see Remark \ref{remark: bridged wedges}). For example, using the shorthand notation $S_{ij\dots} = S_{A_iA_j\dots}$ and similarly for the entanglement wedge $\mathcal{E}_{ij\dots}$, in Figure \ref{Fig: structure of proofs}, $S_A = S_{12} + S_{34}+S_5$, as these have disjoint wedges.

Returning to the general case, we can substitute \eqref{eq: entropy splitting} into our target inequality, and using that $E_W(A:B) = E_W(\mathcal{F}_A:B)$ as $\mathcal{A}=A\setminus \mathcal{F}_A$ is an $A$-only component of $AB$, we obtain the inequality
\begin{equation}
S_{\mathcal{F}_A\cup B} + E_W(\mathcal{F}_A:B) \stackrel{?}{\geq} S_{\mathcal{F}_A}.
\label{eq:ex-statement2}
\end{equation}
This is the same as \eqref{eq: simple inequality to prove} but with the constraint that $\mathcal{F}_AB$ has no $\mathcal{F}_A$-only components. Note it may consist of multiple disjoint \emph{bridged} wedges, as in the case of Figure \ref{Fig: structure of proofs} where $\mathcal{E}_{\mathcal{F}_AB}=\mathcal{E}_{{34} B_1}\sqcup\mathcal{E}_{{5}B_2}$. Thus the contribution from $\mathcal{A}$ is exactly canceled between the LHS and RHS. Hence we rewrite $\mathcal{F}_A\to A$ absorbing the original $\mathcal{A}$ into our purifier $C$, and can now assume there are no $A-$only components in $AB$.

\paragraph{Step 2: Proving \eqref{eq:ex-statement2}.}
Let us define the \emph{LHS surface} $\mathbf{L}$ so that
\[
    \frac{\Area(\mathbf{L})}{4G_N} = \text{LHS}
\]
of a given inequality. In the current case, it is given by
\begin{equation}\label{eq: first mathbfl definition}
\mathbf{L} = \gamma_{AB} \cup \Gamma^W_{A:B},
\end{equation}

where $\gamma_{AB}$ is the RT surface of $AB$ and $\Gamma^W_{A:B}$ defines the EWCS of $(A:B)$ as $\gamma_{AB}$ and $\Gamma^W_{A:B}$ have no intersection (except measure zero sets).

Now, consider $\mathcal{E}(AB)\setminus\Gamma^W_{A:B}$. By the definition of EWCS this splits $\mathcal{E}(AB)$ into an $A$-sided region $U_A$ and a $B$-sided region $U_B$. In our notation this is the statement that $\partial_\infty U_A = A$ and $\partial_\infty U_B=B$. As it stands $U_A$ is a codimension-one bulk subregion with a closed boundary along the original arc of $\gamma_{AB}$ from which it was formed but an open arc along $\Gamma^W_{A:B}$ where it was split from $U_B$. Thus we can close $U_A$ in $\Sigma$ by taking
\begin{equation}\label{eq: closure of ua}
    {U}_A^{\mathrm{cl}} = U_A \cup \Gamma^W_{A:B},
\end{equation}
which is just adjoining the boundary $\Gamma^W_{A:B}$ to $U_A$ along its open edge. 

This is illustrated in Figure~\ref{fig: structure of proof 2} based on the example (Figure~\ref{Fig: structure of proofs}), where $U_A=U_{34}\sqcup U_5$.

\begin{figure}[ht]
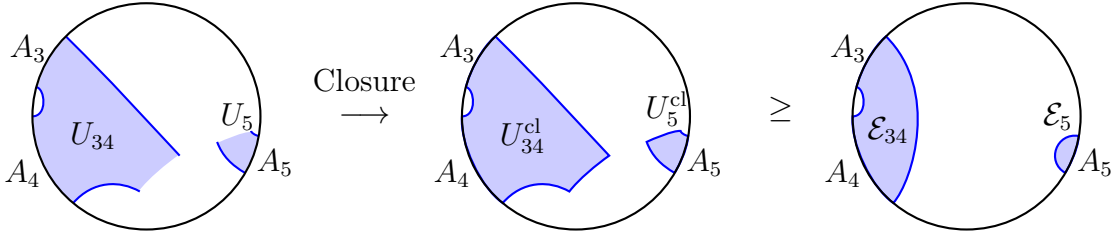

    \centering
\scalebox{1}{
\input{Sketch_Figures/Structure_of_Proofs/LHS_2}
\begin{tikzpicture}
   \node at (-0.5,-0.7) {$\longrightarrow$};
  \node[overlay, anchor=south] at (-0.5,-0.5) {Closure};
  \node at (0,-2.5) { };
\end{tikzpicture}
\input{Sketch_Figures/Structure_of_Proofs/RHS_2}
\begin{tikzpicture}
   \node at (0,-0.6) {$\geq$};
  \node at (0,-2.5) { };
\end{tikzpicture}
\input{Sketch_Figures/Structure_of_Proofs/RHS_2b}
}
    \caption{Decomposition of Figure \ref{Fig: structure of proofs} into $A$-homologous regions. We compute the closure via \eqref{eq: closure of ua} which just glues $\Gamma^W_{A:B}$ to $U_A$.
    }
    \label{fig: structure of proof 2}
\end{figure}

Then $\partial U^\mathrm{cl}_A\subset\mathbf{L}$ as its boundary is just a subset of $\gamma_{AB}$ glued to $\Gamma^W_{A:B}$. Moreover as $\gamma_{AB}$ separates $U_A$ from $C$, and $\Gamma^W_{A:B}$ separates $A$ from $B$ within $\mathcal{E}(AB)$ it must be that $U^\mathrm{cl}_A\sim A$. Thus $\partial U^\mathrm{cl}_{A}$ is homologous to $A$ and so by RT minimality and that $\partial U^\mathrm{cl}_A\subset\mathbf{L}$ we have
\[
S_{AB} + E_W(A:B) = \frac{\Area(\mathbf{L})}{4G_N}\geq\frac{\Area(\partial U^\mathrm{cl}_A)}{4G_N} \geq S_A,
\]
proving the claim.
\end{proof}

\subsection{Properties of the Entanglement Wedge Cross Section}\label{sec:topo-ewcs}

To permit some of our later proofs we must first establish, among the aforementioned qualities (e.g.\ EWN), some  properties of the EWCS. In particular, as evidenced in the proof of equation \eqref{eq: simple inequality to prove}, we need to know when a surgically constructed surface, can be used to upper bound a particular correlation measure; that is whether it satisfies the topological requirements of said measure, but is not necessarily a minimal choice. Hence we define an \emph{RT admissible class} and \emph{EWCS admissible class}.

\begin{definition}[RT admissible class]\label{def: RT admissable}
Let $A$ be boundary subsystem and let $\mathfrak S_{A}$ be the set of properly embedded, piecewise smooth\footnote{We require smoothness only piecewise as we will construct admissible surfaces by gluing various smooth surfaces in a non-smooth way.}, codimension-two hypersurfaces 
$\gamma\subset\Sigma$ such that
\begin{enumerate}[label=\emph{(\roman*)}]
\item \(\Sigma\setminus\gamma = U_A \sqcup U_{\bar{A}}\), with \(\partial_\infty U_A = A\) and $\partial_\infty U_{\bar{A}} = \bar{A}$,
\end{enumerate}
where $\bar{A}$ purifies $A$. This is just the statement that $\gamma$ is homologous to $A$. We say a surface is RT admissible for $A$ iff it satisfies the above condition. Moreover there exists a member $\gamma_A \in \mathfrak{S}_A$ which is the RT surface of $A$ such that $S_A = \Area (\gamma_A)/4G_N$.
\label{RT admissible class}
\end{definition}

\begin{definition}[EWCS admissible class]\label{def: EWCS admissable}
Let $A,B$ be disjoint boundary subsystems.
Let $\mathfrak S_{A:B}$ be the set of properly embedded, piecewise smooth, codimension-two hypersurfaces
$\Gamma\subset\mathcal E(AB)$ such that
\begin{enumerate}[label=\emph{(\roman*)}]
\item \(\partial_\Sigma\Gamma \subset \gamma_{AB}\) (bulk anchoring of $\Gamma$ on $\gamma_{AB}=\partial\mathcal{E}(AB)$), 
\item \(\mathcal{E}(AB)\setminus \Gamma = U_A \sqcup U_B\) with
      \(\partial_\infty U_A= A\) and \(\partial_\infty U_B= B\).
\end{enumerate}
We say a surface is EWCS admissible for $(A:B)$ iff it satisfies the above conditions. Moreover there exists a member $\Gamma_{A:B}^W\in\mathfrak S_{A:B}$ which is is the
\emph{entanglement wedge cross section} (EWCS), with
$E_W(A\!:\!B)=\Area(\Gamma_{A:B}^{W})/(4G_N)$. We will always refer to $\Gamma_{A:B}^{W}$ as the member of the class that minimizes the area function.
\label{EWCS admissible class}
\end{definition}

Proving the admissibility of some constructed codimension-two hypersurface can be lengthy. In our proofs we generally achieve this by constructing paths beginning and ending on certain points at the asymptotic boundary and showing in their evolution they intersect the required surface. As we only include proof sketches in the main text this formalism is outlined in Appendix \ref{s: complete proofs}. We now introduce a fundamental property of the EWCS which will be required for several of our proofs.

\begin{restatable}{theorem}{barriertheorem}\label{thm:barrier}
The EWCS surface, $\Gamma^W_{A:B}$, with $\Area(\Gamma^W_{A:B})/4G_N=E_W(A:B),$ never enters the interior of the individual wedges of its two arguments:
\begin{equation}
\Gamma_{A:B}^{W}\cap \operatorname{Int}\mathcal E(A)=\varnothing,
\qquad
\Gamma_{A:B}^{W}\cap \operatorname{Int}\mathcal E(B)=\varnothing.
\end{equation}
\end{restatable}

From an information theoretic perspective this is of course a necessary condition. In the holographic code picture, $\operatorname{Int}\mathcal E(A)$ consists of bulk degrees of freedom that are reconstructible from $A$ alone, i.e.\ they are ``local'' to one boundary party. Since $E_W(A\!:\!B)$ is intended to quantify correlations that are genuinely \emph{shared} between $A$ and $B$—those accessible only from the joint region $AB$—its dual surface $\Gamma^W_{A:B}$ must not ``count'' any part of the bulk already attributable to a single side. Hence $\Gamma^W_{A:B}$ lives in the joint-only region $\mathcal E(AB)\setminus(\mathcal E(A)\cup\mathcal E(B))$. Equivalently, in the bit-thread formulation, $E_W(A\!:\!B)$ is the maximum flux of threads connecting $A$ to $B$ within $\mathcal E(AB)$, and the RT surfaces $\gamma_A,\gamma_B$ act as capacity bottlenecks for outward flux; any $A\!\leftrightarrow\!B$ thread bundle that detours into $\mathcal E(A)$ or $\mathcal E(B)$ wastes scarce capacity without improving connectivity, so the optimal max flow (and hence the minimal cut $\Gamma^W_{A:B}$) can be pushed to avoid the interiors of the individual wedges.

\begin{figure}[ht]
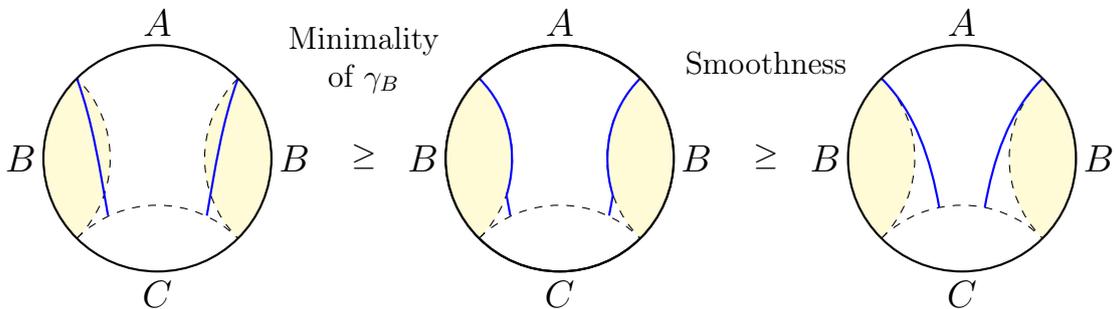

\centering
\scalebox{1}{
\input{Sketch_Figures/Crossing_of_EWCS_into_EB_LHS}
\begin{tikzpicture}
  \node at (0,-0.5) {$\geq$};
  \node[overlay, anchor=south] at (0,0.7) {Minimality};
  \node[overlay, anchor=south] at (0,0.2) {of $\gamma_B$};
  \node at (0,-2.5) { };
\end{tikzpicture}
\input{Sketch_Figures/Crossing_of_EWCS_into_EB_Middle}
\begin{tikzpicture}
  \node at (0,-0.5) {$\geq$};
  \node[overlay, anchor=south] at (0,0.45) {Smoothness};
  \node at (0,-2.5) { };
\end{tikzpicture}
\input{Sketch_Figures/Crossing_of_EWCS_into_EB_RHS}
}
\vspace{-0.5cm}
\caption{$\mathcal{E}(B)$ shaded in yellow. Crossing of $\Gamma^W_{A:B}$ into $\mathcal E(B)$ would require two intersections with $\gamma_B$ (left). By uniqueness, the left segment must coincide with $\gamma_B$ between those points, then non-smoothly join the curve anchoring on $\gamma_C$. By smoothness of minimizers this is forbidden.}
\label{Fig: crossing of EB by EWCS}
\end{figure}

Holographically it is true based on minimality arguments along with EWN. For instance consider in Figure \ref{Fig: crossing of EB by EWCS} where we initially assume $\Gamma^W_{A:B}$ enters $\operatorname{Int}\mathcal{E}(AB)$. However under Assumption \ref{thm: No multiple intersections of geodesics on an AdS Cauchy slice} this is forbidden as we can always push $\Gamma^W_{A:B}$ out of $\mathcal{E}(B)$ to the assumed minimizing geodesic $\gamma_B$. Then by smoothness our original choice of $\Gamma^W_{A:B}$ was incorrect. Of course there is another case with $\mathcal{E}(B)$ a connected wedge as in Figure \ref{fig: example i is not j}. However even here by considering the addition of $S_A$ we find such a configuration erroneously implies $E_W(A:B)>S_A$, and hence the barrier theorem holds. We give a detailed proof of the theorem in Appendix \ref{Barrier proof} in two ways: considering the relation between reflected entropy and the EWCS (equation \eqref{eq: ew = 2sr}) and imposing EWN, and considering a cycle type argument akin to Figure \ref{fig: example i is not j}.

\begin{figure}[ht]
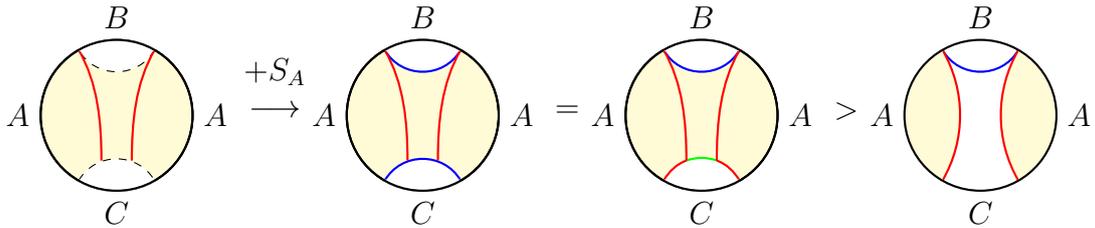

    \centering
\scalebox{1}{
\input{Sketch_Figures/Bridged_wedges/fig7}
\hspace{-0.4cm}
\begin{tikzpicture}
   \node at (0,-0.5) {$\longrightarrow$};
  \node[overlay, anchor=south] at (0,-0.3) {$+S_A$};
  \node at (0,-2) { };
\end{tikzpicture}
\hspace{-0.4cm}
\input{Sketch_Figures/Bridged_wedges/fig8}
\hspace{-0.4cm}
\begin{tikzpicture}
   \node at (0,-0.5) {$=$};
  \node at (0,-2) { };
\end{tikzpicture}
\hspace{-0.4cm}
\input{Sketch_Figures/Bridged_wedges/fig9}
\hspace{-0.4cm}
\begin{tikzpicture}
   \node at (0,-0.5) {$>$};
  \node at (0,-2) { };
\end{tikzpicture}
\hspace{-0.4cm}
\input{Sketch_Figures/Bridged_wedges/fig10}
}
    \caption{Other case in the barrier theorem. We have the dashed curves are $\gamma_A$, the yellow region $\mathcal{E}(A)$ and the red curve $\Gamma^W_{A:B}$ intersects $\operatorname{Int}\mathcal{E}(A)$. This is not a violation however as this configuration implies that $E_W(A:B)>S_B$ which is axiomatically false.}
    \label{fig: example i is not j}
\end{figure}

\section{Holographic Monotonicity}
\label{Section 3}

Having presented the necessary dictionary, as well as various proof tools, we can now proceed to study the monotonicity of the holographic correlation measures $J_W$ and $D_W$. As a summary of our findings, we prove the following monotonicity relations for a mixed state on $ABC$:
\begin{empheq}[box=\fbox]{align}
    D_W(A|BC) &\not\ge D_W(A|B),\\
    J_W(AB|C) &\ge J_W(A|B),\\
    D_W(AC|B) &\ge D_W(A|B)\label{eq: mono unmeasured D},\\
    J_W(AC|B) &\ge J_W(A|B) \label{eq: mono unmeasured J},
\end{empheq}
where $\not\ge$ does not imply the opposite inequality, only that neither is generally true.

\subsection{Measured Party}
In this subsection, we will show that $J_W$ is monotone with respect to the measured party and that $D_W$ is not. Namely,
\begin{empheq}[box=\fbox]{align}
    D_W(A|BC) &\not\ge D_W(A|B),\\
    J_W(A|BC) &\ge J_W(A|B),
\end{empheq}
which are consistent with the known properties for $D$ and $J$~\cite{Henderson:2001wrr}.

The property for $J_W$ is straightforward: expanding out the inequality reduces to
\begin{equation}
    E_W(A:CD) \stackrel{?}{\geq} E_W(A:D),
\end{equation}
where $D$ purifies $ABC$. This is true by EWN and is visualized in Figure~\ref{Monotonicity Measured party}.

\begin{figure}[ht]
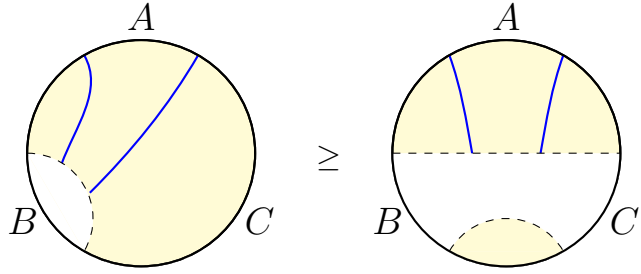

    \centering
\scalebox{1}{
\input{Sketch_Figures/Monotonicity_Measured_party_LHS}
\begin{tikzpicture}
   \node at (0,-0.5) {$\geq$};
   \node at (0,-2.3) { };
\end{tikzpicture}
\input{Sketch_Figures/Monotonicity_Measured_party_RHS}
}
    \caption{Example for $E_W(A:CD)\geq E_W(A:D)$. The blue curves denote the EWCS. Coloured region on the LHS shows $\mathcal{E}(ACD)$ and on the RHS $\mathcal{E}(AD)$.
    }
    \label{Monotonicity Measured party}
\end{figure}

For $D_W$, there are counterexamples to the inequality in both directions.
Consider a pure $ABC$, such that the monotonicity of $D_W$ reduces to
\begin{equation}\label{eq: counter mono DW}
    I(A:C) \stackrel{?}{\geq} E_W(A:C),
\end{equation}
which is not true in general. For example, when $A$ and $C$ are barely connected, $I(A:C)\approx 0$ while $E_W(A:C)=O(1/G_N)$. Thus, holographic discord is not monotone with respect to the measured party. 

The opposite inequality can also be violated, that is, monotonicity sometimes holds. For example, if one takes $A$ and $C$ adjacent to each other, counting the UV divergence satisfies equation~\eqref{eq: counter mono DW}.

\subsection{Unmeasured Party}

Here we show that both $D_W$ and $J_W$ are monotone with respect to the unmeasured party. This is consistent with the known monotonicity relations of the original boundary correlation measures $D(AC|B)\ge D(A|B)$ and $J(AC|B)\ge J(A|B)$~\cite{Henderson:2001wrr,Modi_2012}.
\begin{lemma}[Monotonicity in the unmeasured party of either $J_W$ or $D_W$ implies the other]\label{lemma: monotonicity implies the other} For a mixed boundary state $ABC$ we have
\begin{equation}
D_W(AC|B) \geq D_W(A|B) \iff J_W(AC|B) \geq J_W(A|B).
\end{equation}
\end{lemma}
\begin{proof}
    Expanding out equations \eqref{eq: mono unmeasured D} and \eqref{eq: mono unmeasured J} we see
    \begin{align*}
        D_W(AC|B) \ge D_W(A|B)&\Longrightarrow S_{CD} + E_W(AC:D) \geq S_D + E_W(A:CD)\\
        J_W(AC|B) \ge J_W(A|B)&\Longrightarrow S_{AC} + E_W(A:CD) \geq S_A + E_W(AC:D)
    \end{align*}
    where $D$ purifies $ABC$. Taking $A\leftrightarrow D$ in the second and using the symmetry of EWCS shows they are identical inequalities.
\end{proof}
We have shown one implies the other and so we only need to prove a single measure is monotone.
\begin{restatable}{theorem}{dwmonotone}\label{thm:dwmonotone}
$D_W$ is monotone in the unmeasured party. That is $D_W$ satisfies
\begin{equation}
    D_W(AC|B) \;\ge\; D_W(A|B),
    \label{eq:dw-monotone}
\end{equation}
for any state on $ABC$ (not necessarily pure).
\end{restatable}

We give a detailed proof of this theorem in Appendix~\ref{Holographic unmeasured monotonicity} and here we just illustrate the proof with a particular configuration. 
After relabeling the subsystems, \eqref{eq:dw-monotone} is equivalent to
\[
S_{BC} + E_W(AB:C)\geq
 S_C + E_W(A:BC),
\]
again with $ABC$ not necessarily pure. Figure~\ref{fig: example decomp monotinicity} shows one particular configuration of systems in the leftmost diagram. The middle diagram corresponds to the same curves, but grouped differently such that the blue curve is EWCS admissible for $(A:BC)$ and so upper bounds $E_W(A:BC)$ and the red curve is RT admissible for $C$ and so upper bounds $S_C$. This implies the inequality showing the theorem is true in this configuration. In the figure, the green piece is a leftover and not needed to upper bound the LHS.

\begin{figure}[ht]
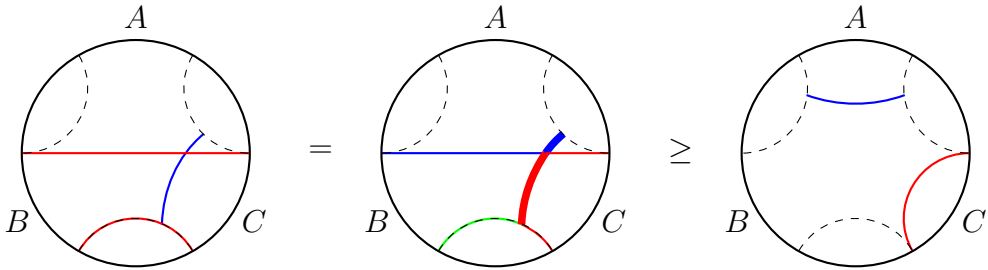

    \centering
\scalebox{1}{
\input{Sketch_Figures/Monotonicity/fig1}
\begin{tikzpicture}
   \node at (0,-0.5) {$=$};
  \node at (0,-1.9) { };
\end{tikzpicture}
\input{Sketch_Figures/Monotonicity/fig2}
\begin{tikzpicture}
   \node at (0,-0.5) {$\geq$};
  \node at (0,-1.9) { };
\end{tikzpicture}
\input{Sketch_Figures/Monotonicity/fig3}
}
    \caption{Example configuration showing $D_W$ in monotone in the unmeasured party.
    On the left, red denotes $\gamma_{BC}$ and blue denotes $\Gamma^W_{AB:C}$; the dashed curve is $\gamma_{ABC}$.
    In the middle, we reassign pieces of the combined red and blue surfaces into a new decomposition serving as an RT candidate for $C$ in red, and EWCS candidate for $(A:BC)$ in blue. The right-hand panel shows $\gamma_C$ and $\Gamma^W_{A:BC}$; by smoothness/minimality, our candidate surfaces in the middle panel have strictly larger area than these extremizers on the right.}
    \label{fig: example decomp monotinicity}
\end{figure}

\section{Holographic Monogamy/Polygamy}
\label{Section 4}

In this section, we study the monogamy/polygamy properties of our holographic correlation measures, $D_W$ and $J_W$. As an initial summary, we find that
\begin{empheq}[box=\fbox]{align}
    D_W(A|B)+D_W(A|C) &\not\le D_W(A|BC), \label{monogamy_discord_measured}\\
    D_W(A|B)+D_W(C|B) &\not\le D_W(AC|B), \label{monogamy_discord_unmeasured}\\
    J_W(A|B)+J_W(C|B) &\le J_W(AC|B), \label{monogamy_classical_unmeasured}\\
    J_W(A|B)+J_W(A|C) &\le J_W(A|BC),\label{monogamy_classical_measured}
\end{empheq}
where as before $\not\leq$ does not imply the opposite just that neither way is true. However for pure ABC we find discord to be polygamous with respect to the unmeasured party
\begin{equation}
     D_W(A|B)+D_W(C|B) \ge D_W(AC|B), \quad \text{pure $ABC$}.
\end{equation}

\subsection{Polygamy of discord}
We first begin by showing that holographic discord is polygamous for the unmeasured party when $ABC$ are pure:
\begin{equation}
     D_W(A|B)+D_W(C|B) \ge D_W(AC|B).
\end{equation}
For pure $ABC$, this is equivalent to
\begin{equation}
    E_W(A:C) \ge \frac{I(A:C)}{2},
\end{equation}
which is a known lower bound for the EWCS~\cite{Freedman:2016zud,Takayanagi:2017knl}. Moreover, when dealing with non-holographic states, we can replace $E_W$ with $E_P$ or $S_R/2$, and the inequality remains true. It is not true however for the original quantum discord, in which $E_W \rightarrow E_F$. 
For example, the polygamy of the original quantum discord is violated for a maximally mixed state, whose purification is the GHZ state on $ABC$. Given the conjectured duality between $D$ and $D_W$, such a polygamy-violating state should be prohibited in holography --- indicating no GHZ entanglement in holography.

For mixed states however polygamy does not hold. For example consider three contiguous intervals $ACB$, namely, $C$ is placed between $A$ and $B$. The rest of the boundary is denoted by $D$. If $C$ is small enough, we have $S_{AB}=S_{CD}=S_C+S_{ABC}$ and $E_W(A:CD)= E_W(AC:D)+E_W(C:AD)<S_A$, leading to
\begin{align}
    & \phantom{=} D_W(A|B)+D_W(C|B)-D_W(AC|B) \nonumber \\
    &= S_B -S_{AB} -S_{BC} +S_{ABC} +E_W(A:CD) + E_W(C:AD) - E_W(AC:D) \nonumber \\
    &= S_B - (S_C+S_{ABC}) - S_{BC} + S_{ABC} +E_W(AC:D)+E_W(C:AD)  \nonumber \\
    &\phantom{=}+ E_W(C:AD) - E_W(AC:D) \nonumber\\
    &= S_B +2E_W(C:AD) -S_C -S_{BC}.
    \label{eq:poly-example}
\end{align}
This is illustrated in Fig.~\ref{fig: poincare patch example for polygamy of discord}. Due to the minimality condition, \eqref{eq:poly-example} is non-negative, indicating polygamy of discord in this case.

\begin{figure}[ht]
\begin{center}
    \input{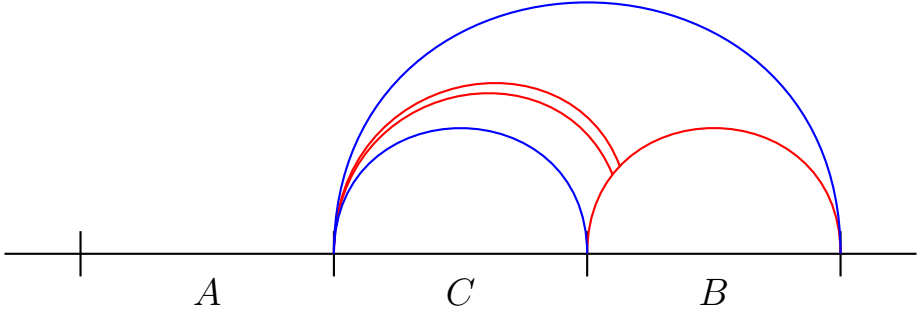}
    \caption{Polygamy of discord for mixed state $ABC$ on the Poincar\'e patch: $S_B +2E_W(C:AD) -S_C -S_{BC}\ge 0$. Positive curves colored in red and negative in blue.}
    \label{fig: poincare patch example for polygamy of discord}
\end{center}
\end{figure}

Moreover, holographic mixed states allow a range of monogamous examples as well. For example, consider three disjoint intervals on the boundary circle:
\begin{equation}
    A = \left(0,\frac{\pi}{3}+\theta\right), \quad B = \left(\frac{2\pi}{3},\pi+\theta\right), \quad C = \left(\frac{4\pi}{3},\frac{5\pi}{3}+\theta\right).
\end{equation}
When $\theta = 0$, they are equally spaced and when $\theta = \pi/3$, the state on $ABC$ is pure and the subsystems are mutually adjacent. Then, the remaining purifying subsystem $D$ consists of three segments
\begin{equation}
    D_j = \left(\frac{(2j-1)\pi}{3}+\theta,\frac{2j\pi}{3}\right),
\end{equation}
where $j$ runs from $1-3$. When $\theta>0$ is sufficiently small, all regions are disconnected. Moreover, the EWCS are just given by the entropies of the smallest argument, and so the inequality becomes
\begin{equation}
    \sin\left(\frac{\frac{\pi}{3}-\theta}{2}\right) \stackrel{?}{\geq}  \sin\left(\frac{\frac{\pi}{3}+\theta}{2}\right).
\end{equation}
which is violated for any sufficiently small $\theta>0$. 

Thus, this concludes that holographic discord is neither monogamous or polygamous for mixed states. 

~

Is there a monogamy relation for holographic discord with respect to the measured party? Namely, we ask whether
\begin{equation}
    D_W(A|B)+D_W(A|C) \stackrel{?}{\le} D_W(A|BC).
\end{equation}
Expanding the inequality for pure $ABCD$ yields
\begin{equation}
    S_D + I(B:C) + E_W(A:CD)+E_W(A:BD)  \stackrel{?}{\leq} S_{CD} + S_{BD} + E_W(A:D).
\end{equation}
An easy counter example is taking $D = \emptyset$ in which this reduces to
\begin{equation}
    E_W(A:C) + E_W(A:B) \stackrel{?}{\leq} S_A,
\end{equation}
which has obvious violating configurations. Similarly we can violate the polygamy by taking $C = \emptyset$ in which the inequality reduces to
\begin{equation}
    E_W(A:D) \stackrel{?}{\geq} S_{A},
\end{equation}
which is axiomatically false. Thus $D_W$ is neither polygamous nor monogamous for the measured party.

Lastly, we briefly remark on strong superadditivity. Since holographic discord $D_W$ is not monogamous, it does not satisfy strong superadditivity \eqref{eq:strong-sa} as well: strong superadditivity implies monogamy so taking its contraposition denies strong superadditivity of $D_W$ for either party.

\subsection{Monogamy of classical correlations} \label{sec:holo-tradeoff}
We will now prove that $J_W$ is monogamous in both arguments beginning with the measured party.

\subsubsection{Measured party}

\begin{restatable}{theorem}{jwmonogamousmeasured}\label{thm:mono-measured}
$J_W$ is monogamous in the measured party. That is $J_W$ satisfies
\begin{equation}
    J_W(A|B)+J_W(A|C) \;\le\; J_W(A|BC). \label{eq:jw-monogamy}
\end{equation}
for any state $ABC$ (not necessarily pure).
\end{restatable}

\begin{figure}[ht]
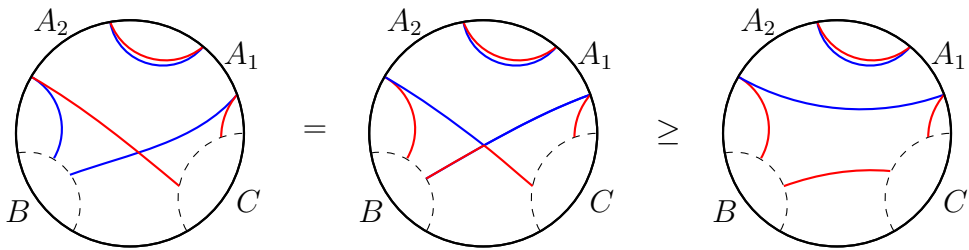

    \centering
\scalebox{1}{
\input{Sketch_Figures/monogamy_measured/fig1}
\begin{tikzpicture}
   \node at (0,0) {$=$};
   \node at (0,-1.4) { };
\end{tikzpicture}
\input{Sketch_Figures/monogamy_measured/fig2}
\begin{tikzpicture}
   \node at (0,0) {$\geq$};
   \node at (0,-1.4) { };
\end{tikzpicture}
\input{Sketch_Figures/monogamy_measured/fig3}
}
    \caption{Example for monogamy of $J_W$ for the measured party. On the LHS we have $\Gamma^W_{A:CD}$ in blue, and $\Gamma^W_{A:BD}$ in red. The dashed lines are the RT surfaces of $\gamma_B$ and $\gamma_C$ and we slightly displace to top red and blue curves to show they are both present. We decompose this in the middle figure into a blue curve $\partial r_A$, an RT candidate for $A$, and a red curve $\Gamma$, which is an EWCS candidate for $(A:D)$. This upper bounds the RHS figure via smoothness.}
    \label{Monogamy classical unmeasured party}
\end{figure}

For pure $ABC$, it immediately follows from the polygamy of $D_W$. By combining that 1.) quantum mutual information is monogamous in holography~\cite{Hayden:2011ag} and 2.) $I(A:B) = D_W(A|B) + J_W(A|B)$ by definition, $J_W$ must be monogamous to compensate for the polygamy of $D_W$ for pure states on $ABC$. To summarize,
\begin{equation}
\left\{
\begin{aligned}
    & D_W(A|B)+D_W(C|B) \ge D_W(AC|B) \\
    & I_W(A:B)+I_W(C:B) \le I_W(AC:B) \\
    & I_W(A:B) = D_W(A|B) + J_W(A|B)
\end{aligned}
\right\}
 \Rightarrow 
J_W(A|B)+J_W(C|B) \le J_W(AC|B).
\label{eq:trade-off}
\end{equation}

Beyond pure states, the full proof is presented in Appendix~\ref{monogamy measured party}. For an example case consider Figure \ref{Monogamy classical unmeasured party} which plots the expansion of Theorem \ref{thm:mono-measured}:
\begin{equation}
E_W(A\!:\!BD)\;+\;E_W(A\!:\!CD)\ \ge\ S_A\;+\;E_W(A\!:\!D).
\end{equation}
The proof in this configuration then proceeds similarly to that of Theorem \ref{thm:dwmonotone}. We simply consider the LHS surfaces and relabel particular components to serve as admissible surfaces for the quantities on the RHS. One has to be careful about formalizing such a decomposition however due to potential degeneracy of surfaces as seen in Figure \ref{Monogamy classical unmeasured party} where we have two identical curves. This is just due to representing two surfaces on the same copy of the bulk. Hence we need to keep track of the multiplicities involved (see Definition \ref{def:multisurface_sheeted_simple}).

\subsubsection{Unmeasured party}

\begin{restatable}{theorem}{jwmonogamousunmeasured}\label{thm:mono-unmeasured}
$J_W$ is monogamous in the unmeasured party. That is $J_W$ satisfies
\[
     J_W(A|B)+J_W(C|B) \le J_W(AC|B).
\]
for any states on $ABC$ (not necessarily pure).
\end{restatable}

\begin{figure}[ht]
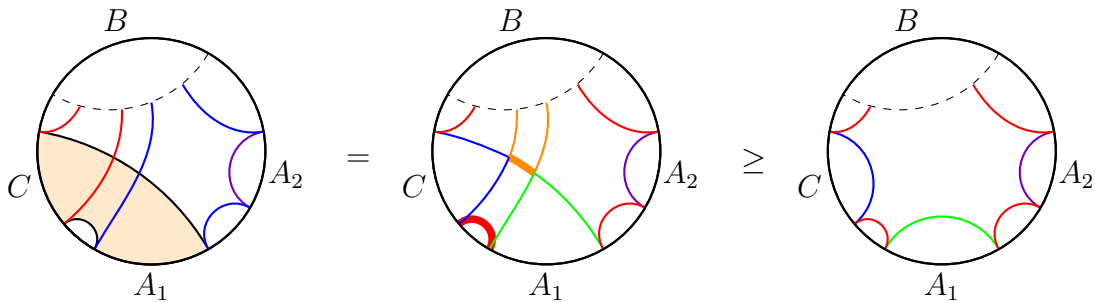

    \centering
\scalebox{1}{
\input{Sketch_Figures/monogamy_unmeasured/fig1}
\begin{tikzpicture}
   \node at (0,0) {$=$};
   \node at (0,-1.85) { };
\end{tikzpicture}
\input{Sketch_Figures/monogamy_unmeasured/fig2}
\begin{tikzpicture}
   \node at (0,0) {$\geq$};
   \node at (0,-1.85) { };
\end{tikzpicture}
\input{Sketch_Figures/monogamy_unmeasured/fig3}
}
    \caption{Example for monogamy of $J_W$ for the unmeasured party, for $\mathcal{F}=A_1\cup C$. On the LHS we have $\gamma_{\mathcal{F}}$ in black, $\mathcal{E}(\mathcal{F})$ shaded in orange, $\Gamma^W_{C:AD}$ in red, $\Gamma^W_{A:CD}$ in blue, and $\gamma_{A_2}$ in purple. We decompose this in the middle figure into a blue curve $\partial U^\mathrm{cl}_{\mathcal{F}_C}$, a green curve $\partial  U^\mathrm{cl}_{\mathcal{F}_A}$, a red curve $\Gamma$, and an extra orange curve. This upper bounds the RHS figure via minimality. $S_{A_2}$ is identical in each panel and hence cancels identically.}
    \label{Monogamy classical measured party}
\end{figure}

Here we will illustrate this with an explicit example as shown in Figure~\ref{Monogamy classical measured party}.
The full proof is given in Appendix~\ref{monogamy unmeasured party proof}. Expanding out \eqref{eq:jw-monogamy} in terms of HEE and EWCS is equivalent to
\begin{equation}
S_{AC}\;+\;E_W(A:CD)\;+\;E_W(C:AD)\ \ge\ S_A\;+\;S_C\;+\;E_W(AC:D),
\end{equation}
for a pure state $ABCD$. In our example $A$ is disjoint with two intervals $A_1$ and $A_2$. First, notice that the unbridged components of $A$ and $C$ when computed in $AC$ will have their entropic contributions to $S_{AC}$ cancel with the corresponding unbridged contributions to $S_A$ and $S_C$ by Remark \ref{remark: bridged wedges}. This is seen in Figure \ref{Monogamy classical measured party} where the purple curve, defining $S_{A_2}$, is present in all panels. We then proceed similarly to the proof of Theorem~\ref{thm:dwmonotone}. Going from the the leftmost to the middle diagram of Figure~\ref{Monogamy classical measured party}, we have only relabeled surface so that each color upper bounds the corresponding minimal surface on the right-hand side, except the orange curve which is an extra piece. Formalizing this process for a general setting proves Theorem~\ref{thm:mono-unmeasured}.

\section{Holographic Strong Superadditivity}
\label{Section 5}

In this section, we turn to a different type of inequality from those discussed earlier---namely, the strong superadditivity of $J_W$. This is motivated from the fact that distillable entanglement satisfies strong superadditivity and a previous observation that $J_W$ equals the distillable entanglement $E_D$, which will be reviewed briefly here.

\subsection{Distillable entanglement and holographic conjecture}\label{sec:ED-review}
The distillable entanglement between $A$ and $B$, denoted by $E_D(A:B)$, quantifies the \emph{maximum} number of EPR pairs that can be extracted via local operations and classical communication (LOCC) between $A$ and $B$.\footnote{More precisely, it is defined in either the asymptotic iid limit or one-shot scenario with some error. For more details, see~\cite{PhysRevA.60.179,Mori:2024gwe,Li:2025nxv} for instance.} In most cases, this procedure requires one party to perform a measurement and send the outcome to the other, thereby necessitating classical communication.

Since $E_D$ is operationally defined, a restricted set of LOCC will lead to smaller $E_D$. For instance let us consider four parties $A,B,C,D$ and focus on the distillable entanglement between $AB$ and $CD$. The set of LOCC operations allowed between the joint systems $AB$ and $CD$ is evidently much larger than the set of operations achievable when performing LOCC separately between $A-C$ and $B-D$. This leads to the following inequality called strong superadditivity:
\begin{equation}
    E_D(AB:CD) \ge E_D(A:C) + E_D(B:D).
\end{equation}

In~\cite{Mori:2024gwe}, it is proposed that a one-way, one-shot distillable entanglement admits a gravity dual as
\begin{equation}
    E_D(A\leftarrow B) = J_W(A|B)
    \label{eq:ED=JW}
\end{equation}
at leading order in $G_N$ for holographic states and in the large-dimension limit for random states. $A\leftarrow B$ means the LOCC considered here is one-way from $B$ to $A$. Specifically equation \eqref{eq:ED=JW} follows from the following three inequalities:
\begin{align}
    E_D(A\leftarrow B) &\lesssim J(A|B) \\
    E_D(A\leftarrow B) &\ge J_W(A|B) \\
    J(A|B) &= J_W(A|B).
\end{align}

The first inequality is a general relation that holds whenever high-fidelity EPR pairs are distilled; it is not limited to holographic or random states. It follows from the fact that, once an EPR pair has been successfully distilled, measuring one of its subsystems necessarily reduces the entropy by a corresponding amount.

The second inequality arises from the quantum error–correcting property of holographic or random states. One can explicitly construct a one-way LOCC distillation protocol that makes use of the expansion of the entanglement wedge of $A$ after a partial measurement on $B$. The resulting distillable entanglement coincides with $J_W(A|B)$. Although this protocol may not be optimal, it constitutes a valid distillation procedure and thus provides a lower bound on $E_D$.

The third inequality can be shown independently in two settings: holographic or Haar random states. In the former case, the derivation relies on the assumption of holographic measurements and geometric optimization in fixed-area states of holography. In the latter case, the inequality is rigorously argued by the measure concentration of Haar random states.

Given this background, \cite{Mori:2024gwe} conjectures the classical correlation $J_W(A|B)$ is dual to one-way distillable entanglement $E_D(A\leftarrow B)$ and so we expect a one-way version of strong superadditivity:
\begin{equation}
    J_W(AB|CD) \stackrel{?}{\ge} J_W(A|C) + J_W(B|D).
    \label{eq:Jw-SSA-1way}
\end{equation}
It is worth noting that this one-way strong superadditivity implies the monogamy relation for both unmeasured and measured parties if one takes $C=D$ or $A=B$. However, for holographic states, since our topological proof assumes non-overlapping subsystems, monogamy and strong superadditivity must be proven independently.

Restricting to holographic/random basis measurements (without feedback), one-way LOCC does not outperform two-way LOCC. If it is generically true for holographic/random states, the two-way, one-shot distillable entanglement should be given by
\begin{equation}
    E_D(A:B)\stackrel{?}{=}\max(E_D(A\leftarrow B), E_D(B\leftarrow A)).
    \label{eq:2wayED=1way?}
\end{equation}
This motivates us to think of a symmetrized classical correlation~\cite{Mori:2024gwe,Mori:2025gqe}
\begin{equation}\label{eq: jw 2 way def}
    J_W(A:B) \equiv \max(J_W(A|B),J_W(B|A)).
\end{equation}

While answering question~\eqref{eq:2wayED=1way?} is hard, here we ask if its consequences hold or not. If the two-way distillable entanglement is dual to the symmetrized correlation, that is, $E_D(A:B)=J_W(A:B)$, then there must be a two-way strong superadditivity
\begin{equation}
    J_W(AB:CD) \stackrel{?}{\ge} J_W(A:C) + J_W(B:D).
    \label{eq:Jw-SSA-2way}
\end{equation}
This is strictly stronger than~\eqref{eq:Jw-SSA-1way} so it does not follow from the one-way strong superadditivity. Proving this holographically and/or for random states provides partial supporting evidence that one-way LOCC is sufficient in holography and/or random states.

\subsection{Holographic proof of one-way strong superadditivity}

In this section we show that $J_W(A|B)$ respects strong superadditivity in the one-way direction.

\begin{restatable}{theorem}{jwssa}\label{thm:jw-ssa}
$J_W$ is one-way strongly superadditive. That is $J_W$ satisfies
\begin{equation}
     J_W(AB|CD) \ge J_W(A|C) + J_W(B|D).\label{eq: one way classical ssa}
\end{equation}
for any states on $ABCD$ (not necessarily pure).
\end{restatable}

The full proof is given in Appendix~\ref{SSA formal proof}. The one-way strong superadditivity~\eqref{eq: one way classical ssa} is equivalently written as
\begin{equation}\label{eq: one way SSA}
S_{AB} + E_W(A:BDE) + E_W(B:ACE) \stackrel{?}{\geq} S_A + S_B + E_W(AB:E).
\end{equation}
The proof proceeds with the same strategy as all our other holographic proofs: consider the LHS surfaces, relabel particular components, show these components are admissible, and hence upper bound, quantities on the RHS. This is seen in Figure~\ref{Strong superadditivtg figure}. 

\begin{figure}[ht]
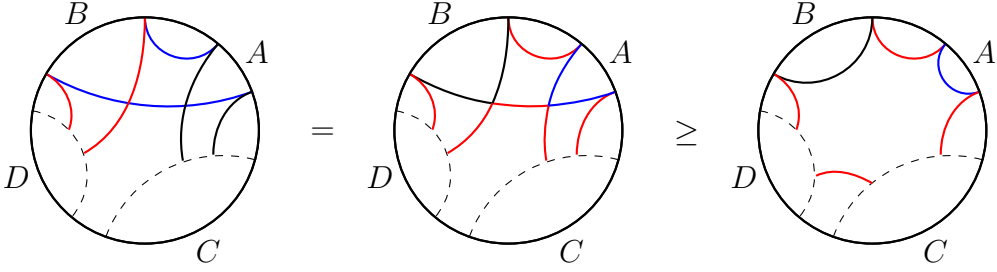

    \centering
\scalebox{1}{
\input{Sketch_Figures/ssa/fig1}
\begin{tikzpicture}
   \node at (0,0) {$=$};
   \node at (0,-1.7) { };
\end{tikzpicture}
\input{Sketch_Figures/ssa/fig2}
\begin{tikzpicture}
   \node at (0,0) {$\geq$};
   \node at (0,-1.7) { };
\end{tikzpicture}
\input{Sketch_Figures/ssa/fig3}
}
    \caption{Example for one way classical strong super additivity. On the LHS we have $\Gamma^W_{A:BDE}$ in black, $\Gamma^W_{B:ACE}$ in red, and $\gamma_{AB}$ in blue. The dashed lines are the RT surfaces of $\gamma_C$ and $\gamma_D$. We decompose this in the middle figure into a black curve $\partial r_B$, a blue curve $\partial r_A$, and a red curve $\Gamma$. These upper bound $S_B$, $S_A$ and $E_W(AB:E)$ respectively by minimality.}
    \label{Strong superadditivtg figure}
\end{figure}

\subsection{Toward holographic proof of two-way strong superadditivity}

Although we are not able to present a complete proof here it is worth analyzing why the method employed for all other proofs fails. For convenience, let us introduce the following shorthand notation:
\begin{align*}
    \alpha_1 &\equiv J_W(AB|CD) = S_{AB} - E_W(AB:E), \\
     \alpha_2 &\equiv J_W(CD|AB) = S_{CD} - E_W(CD:E), \\
      \beta_1 &\equiv J_W(A|C) = S_{A} - E_W(A:BDE), \\
       \beta_2 &\equiv J_W(C|A) = S_{C} - E_W(C:BDE), \\
        \gamma_1 &\equiv J_W(B|D) = S_{B} - E_W(B:ACE), \\
         \gamma_2 &\equiv J_W(D|B) = S_{D} - E_W(D:ACE).
\end{align*}
Then, the one-way strong superadditivity~\eqref{eq: one way classical ssa} is written as
\begin{equation}
	\alpha_1 \stackrel{?}{\ge} \beta_1 + \gamma_1,
	\label{onewaySSA}
\end{equation}
while the two-way version corresponds to
\begin{equation}
\max(\alpha_1,\alpha_2)  \stackrel{?}{\ge} \max(\beta_1,\beta_2) + \max(\gamma_1,\gamma_2).
\label{twowaySSA}
\end{equation}
As we show \eqref{onewaySSA} the only remaining case is
\begin{equation}\label{eq: ssa other way}
	\alpha_1  \stackrel{?}{\ge} \beta_1 + \gamma_2,
\end{equation}
subject to $\alpha_1 \geq \alpha_2$, $\beta_1 \geq \beta_2$ and $\gamma_2 \geq \gamma_1$. The various other cases are implied by equations \eqref{onewaySSA} and \eqref{eq: ssa other way} and the assumed max statements within.

\begin{figure}[ht]
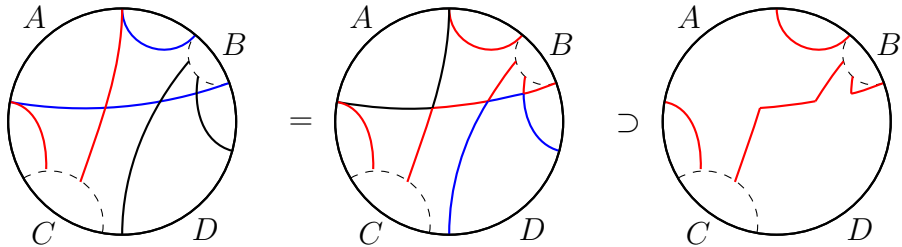

    \centering
\scalebox{1}{
\input{Sketch_Figures/ssa_failure/fig1}
\begin{tikzpicture}
   \node at (0,0) {$=$};
   \node at (0,-1.6) { };
\end{tikzpicture}
\input{Sketch_Figures/ssa_failure/fig2}
\begin{tikzpicture}
   \node at (0,0) {$\supset$};
   \node at (0,-1.6) { };
\end{tikzpicture}
\input{Sketch_Figures/ssa_failure/fig3}
}
    \caption{Failure of standard proof methodology for two-way strong superadditivity. On the LHS we have $\Gamma^W_{A:BDE}$ in red, $\Gamma^W_{D:ACE}$ in black, and $\gamma_{AB}$ in blue. The dashed lines are the RT surfaces of $\gamma_C$ and $\gamma_B$. We attempt the usual decomposition into a black curve upper bounding $S_A$, a blue curve upper bounding $S_D$ and a remaining red curve to upper bound $E_W(AB:E)$. Evidently the remaining curve is not EWCS admissible for $E_W(AB:E)$ and the standard proof method fails.}
    \label{Strong superadditivtg failure}
\end{figure}

For the one-way direction proof we did not need to invoke the additional information provided by the assumed max statements: the one-way inequality in equation \eqref{eq: one way SSA} is true for arbitrary $ABCD$. However for the remaining two-way condition,
\begin{equation}
	J_W(AB|CD)  \stackrel{?}{\ge} J_W(A|C) + J_W(D|B),
\end{equation}
we find the standard approach insufficient (see Figure \ref{Strong superadditivtg failure}). The proof thus requires a step beyond considering the basic topology of these surfaces and invoking more exotic properties along with the assumed inequalities between $\alpha_1$ and $\alpha_2$ and so on. Nevertheless, we find no counterexample both analytically or numerically. Since most holographic inequalities are proven on topological grounds, it is interesting if there exists some approach that captures more fine-grained details of holographic theories.

\subsection{Haar Random analysis}
In this subsection, we consider $n$-qudit Haar random pure states composed of five subsystems $A,B,C,D,E$. A subsystem $X$ contains $n_X$ qudits and the total number of qudits is denoted by $n$. 
In the large-$n$ limit with fixed $n_X/n$, the Haar random state is the simplest toy model of holography, in a sense that the EE is calculated by the minimal bond cut surface as the RT formula does. Since it is volume-law entangled, it is identified with the infinite-temperature pure black hole state or the coarse-grained random tensor network corresponding to a fixed-area state.

Since the one-way strong superadditivity is generically proven based on topological arguments, which also apply for the large-$n$ asymptotics of Haar random states, the only remaining case we should examine is equation \eqref{eq: ssa other way}. In the large-$n$ limit, it simplifies to (up to exponentially suppressed corrections)
\begin{equation}
    \begin{split}
    & \min(n_A+n_B,n_C+n_D+n_E) + \min(n_A,n_B+n_D+n_E) + \min(n_D,n_A+n_C+n_E) \\
    \stackrel{?}{\ge} & \min(n_A,n_B+n_C+n_D+n_E) + \min(n_D,n_A+n_B+n_C+n_E) + \min(n_A+n_B,n_E)
    \end{split}
    \label{eq:haar-ssa}
\end{equation}
under the assumptions $\alpha_1\ge \alpha_2$, $\beta_1\ge \beta_2$, and $\gamma_2 \ge \gamma_1$.
This simplification is because there is a measure concentration, which leads to
\begin{equation}
    S_X = \min(n_X,n_{\bar{X}}),\quad E_W(X:Y) = \min(n_X,n_Y)
\end{equation}
where $\bar{X}$ represents the complementary subsystem to $X$. We confirm strong superadditivity~\eqref{eq:haar-ssa} analytically by case analysis (see Appendix~\ref{app:code}).

\section{Boundary dual to holographic correlations from reflected entropy}
\label{Section 6}

In this section, inspired by holographic correlation measures, we map our correlation measures back to boundary quantities by making use of another duality between EWCS and a half of reflected entropy, defined in \ref{eq:reflected entropy}. We define reflected correlation measures by making the following replacement:
\begin{equation}
    E_W(A:B) \rightarrow \frac{S_R(A:B)}{2}.
\end{equation}
The reflected correlation measures not only match with the original correlation measures for holographic/random states, but also provide well-defined, computable quantities in any quantum systems away from the holographic regime. Notably, they are free of complex optimizations that the original measures have.

In Section \ref{sec:def-DR}, we define these reflected correlation measures and conjecture a bound on the original measures. We demonstrate this bound with two-qubit states. In Section \ref{sec:mono-DR} and \ref{sec:monog-DR}, we investigate monotonicity and monogamy/polygamy for the reflected measures. While reflected entropy is not a correlation measure~\cite{Hayden:2023yij}, it turns out $J_R$ and $D_R$ obey the same monotonicity and monogamy relations as the original measures except for monotonicity in the measured party.

\subsection{Definition and relation to original measures}\label{sec:def-DR}

Let us define the reflected classical correlation $J_R$ and reflected discord $D_R$ by
\begin{align}
    J_R(A|B) &\equiv S_A - \frac{S_R(A:C)}{2},\\
    D_R(A|B) &\equiv I(A:B) - J_R(A:B) = S_B -S_{AB} + \frac{S_R(A:C)}{2},
\end{align}
where $C$ purifies $AB$. Unlike the original measures, they no longer have a complex optimization over all possible subsystem POVMs, offering an alternative, computable measure for classical and quantum correlations. Furthermore, they do not depend on the choice of the purification partner $C$ as they are entirely entropic.

It is immediate to see that
\begin{equation}
    D_R\ge D \Leftrightarrow J_R \le J \Leftrightarrow S_R\ge 2 E_F,
    \label{eq:conj-DR}
\end{equation}
if $E_f$ is additive over a canonical purification. This is as additivity implies that
\begin{equation}
    S_R(A:B) = E_F(AA^\ast: BB^\ast) \ge E_F(A:B)+E_F(A^\ast:B^\ast) = 2E_F(A:B),
\end{equation}
showing equation \eqref{eq:conj-DR}. While it is known that the additivity of the entanglement of formation can be violated~\cite{Hastings:2009ybd}, it is not clear whether it is also violated for canonically purified states. Furthermore, whether there exists an extensive\footnote{By extensive, we mean that it scales as $O(\log d)$ as we increase the dimension $d$.} (large) violation remains an open problem~\cite{Hayden:2020vyo,Youn_2022,kalantar2025examples}.

It is also interesting to note that there was a conjecture about the relation between the reflected entropy and entanglement of purification $E_P(A:B)$: a quantity closely related to entanglement of formation. Namely, it is generically true that~\cite{Terhal:2002riz}
\begin{equation}
    E_P(A:B) \ge E_F(A:B).
\end{equation}
It was conjectured by~\cite{Akers:2023obn} that $2E_P(A:B)\stackrel{?}{\ge} S_R(A:B)$ but later fine-tuned counterexamples were found numerically~\cite{Couch:2023pav}. For these counterexamples, our conjecture \eqref{eq:conj-DR} is automatically true due to $S_R\ge 2E_P\ge 2E_F$. However, for most states, including two-qubits with no known counterexamples, \eqref{eq:conj-DR} remains nontrivial.

Let us give an explicit, neither holographic nor random example: a family of two-qubit states, called the Bell-diagonal states~\cite{Horodecki:2009zz}, where we can confirm the conjecture explicitly. These Bell-diagonal take the form
\begin{equation}
    \rho_{AB} = p \dyad{\Phi^+} + (1-p) \dyad{\Phi^-},
    \label{eq:bell-diag}
\end{equation}
where the Bell pairs are defined as $\ket{\Phi^\pm} = \frac{1}{\sqrt{2}}(\ket{00}\pm \ket{11})$.

Since it is a two-qubit state, one can calculate $E_f$ analytically using~\cite{Wootters:1997id}:
\begin{equation}
    E_F(A:B) = h\qty(\frac{1+\sqrt{1-C^2}}{2}),\quad h(q) = -q\log q -(1-q)\log(1-q),
\end{equation}
with the concurrence $C$ given by ~\cite{Lang:2010xtw}
\begin{equation}
    C=\max(0,2\max(p,1-p)-1)=\abs{2p-1}.
\end{equation}
With this, one finds that
\begin{equation}
    E_F(A:B) = h\qty(\frac{1+2\sqrt{p(1-p)}}{2}).
\end{equation}

Meanwhile, the reflected entropy can be calculated as follows. The canonical purification is
\begin{equation}
    \ket{\sqrt{\rho}}_{ABA'B'} = \sqrt{p}\ket{\Psi^+}_{AB}\otimes \ket{\Psi^+}_{A'B'} + \sqrt{1-p}\ket{\Psi^-}_{AB}\otimes \ket{\Psi^-}_{A'B'}.
\end{equation}
After partial trace over $BB'$, one finds that
\begin{equation}
    \rho_{AA'} = \frac{I_{AA'}}{4} + \frac{\sqrt{p(1-p)}}{2} Z_A\otimes Z_{A'}, \quad Z=\mqty(1&0\\0&-1),
\end{equation}
leading to the reflected entropy
\begin{equation}
    S_R(A:B)=S_{AA'} = 2(-\lambda_+\log\lambda_+ - \lambda_- \log\lambda_-), \quad \lambda_\pm=\frac{1}{4}\pm\frac{\sqrt{p(1-p)}}{2}.
\end{equation}

Thus, the question $S_R(A:B)\stackrel{?}{\ge}2E_F(A:B)$ is equivalent to
\begin{equation}
    -\lambda_+\log\lambda_+ - \lambda_- \log\lambda_- \stackrel{?}{\ge} h\qty(\frac{1+2\sqrt{p(1-p)}}{2}),
\end{equation}
which is satisfied for any valid $p$ ( $0\le p\le 1$).

\subsection{Monotonicity of reflected measures}\label{sec:mono-DR}
Let us turn to examine the monotonicity properties of reflected correlation measures. In summary, we find monotonicity in the measured party for both measures, and no monotonicity in general for the unmeasured party:
\begin{empheq}[box=\fbox]{align}
        D_R(AC|B)&\ge D_R(A|B),\label{mono_discord_reflect}\\
        J_R(AC|B)&\ge J_R(A|B),\label{mono_classical_reflect}\\
        D_R(A|BC)&\not\ge D_R(A|B), \label{eq:mono-dr-q}\\
        J_R(A|BC)&\not\ge J_R(A|B). \label{eq:mono-jr-q}
\end{empheq}

\subsubsection{Measured party}\label{sss: Monotonicity for the measured party}
Let us start with asking if monotonicity holds for the measured party in $J_R$~\eqref{eq:mono-jr-q}. The inequality in question is equivalent to asking 
\begin{equation}
    S_R(A:CD) \stackrel{?}{\ge} S_R(A:D),
\end{equation}
with $ABCD$ pure. However, numerous counterexamples have been found for the monotonicity of $S_R$ \cite{Hayden:2023yij} and so $J_R$ is not monotone in the measured party. 

For the reflected discord, the monotonicity~\eqref{eq:mono-dr-q} is equivalent to
\begin{equation}
    2S_{AD} - 2S_D + S_R(A:D) \stackrel{?}{\ge} 2S_B - 2S_{CD} + S_R(A:CD).
\end{equation}
Taking $D=\emptyset$, the inequality becomes
\begin{equation}
    I(A:C) \stackrel{?}{\ge} \frac{S_R(A:C)}{2}.
    \label{eq:mi-vs-sr}
\end{equation}
One such counter example is given by the Werner state \cite{Werner:1989zz}:
\begin{equation}
    \rho_{AC} = p \dyad{\Phi^+}_{AC} + (1-p)\frac{I_{AC}}{4},
\end{equation}
which violates \eqref{eq:mi-vs-sr} for $0< p \lesssim 0.42$ (see Appendix \ref{app: Counter example to monotonicity of reflected discord in the measured party}).

\subsubsection{Unmeasured party}
\label{sss: monotonicity for the unmeasured party}

\begin{theorem}[Monotonicity of classical and quantum correlations for the unmeasured party]\label{thm:mono-unmeasured-reflected}
For a mixed state $ABC$,
\begin{align}
    J_R(AC|B)&\ge J_R(A|B)\\
    D_R(AC|B)&\ge D_R(A|B)
\end{align}
\end{theorem}

\begin{proof}

As per the holographic case (Lemma \ref{lemma: monotonicity implies the other}) the inequality for both $J_R$ and $D_R$ is identical after appropriate renaming. The monotonicity of these reflected measures for the unmeasured party is equivalent to asking whether
\begin{equation}
2S_{AB}+S_R(A:BC) \stackrel{?}{\ge} 2S_A + S_R(AB:C).
\label{eq:wedge-disc-sr}
\end{equation}

We first note that for any subsystems $X$ and $Y$, it follows that
\begin{equation}
    S_R(X:Y) \equiv S_{XX'}(\lvert \psi_{XYX'Y'} \rangle) = S_X + S_{X'}-I(X:X')=2S_X - I(X:X'),
\end{equation}
where in the last equality, we used that canonical purification implies $S_X=S_{X'}$. Applying this identity to the reflected entropies appearing in \eqref{eq:wedge-disc-sr}, we have 
\begin{equation}
    S_R(A:BC) = 2S_A - I(A:A'), \quad S_R(AB:C) = 2S_{AB} - I(AB:A'B').
\end{equation}
We note that in both $S_R(A:BC)$ and $S_R(AB:C)$, the canonically purified state is identical: $\ket{ \psi_{ABCA'B'C'}}$. 
Hence, we can rewrite the desired inequality \eqref{eq:wedge-disc-sr} as
\begin{equation}
I(AB:A'B') - I(A:A') \stackrel{?}{\ge} 0,
\end{equation}
which is indeed positive due to the monotonicity of the mutual information. Thus, $D_R$ and $J_R$ are monotone with respect to the unmeasured party.
\end{proof}

\subsection{Monogamy/Polygamy of reflected measures}\label{sec:monog-DR}
Here we investigate the monogamy/polygamy properties of $J_R$ and $D_R$. We find these quantities are neither monogamous nor polygamous for either party. Namely,
\begin{empheq}[box=\fbox]{align}
    J_R(A|B) + J_R(A|C) &\not\le J_R(A|BC),\\
    D_R(A|B) + D_R(A|C) &\not\le D_R(A|BC),\\
    J_R(A|B) + J_R(C|B) &\not\le J_R(AC|B),\\
    D_R(A|B) + D_R(C|B) &\not\le D_R(AC|B).
\end{empheq}

\subsubsection{Measured party}

Just as in \eqref{eq:trade-off}, where the holographic classical correlations were monogamous for pure states and there was a form of trade-off with the polygamy of holographic discord, here we see that states violating the monogamy of one of these reflected correlation measures often violate the polygamy of the other.  We present two instances where there is this complementary violation of monogamy/polygamy for $J_R$ and $D_R$.

\paragraph{W states: }
Let us consider a counterexample to the polygamy of $J_R$ and the monogamy of $D_R$ with respect to the measured party. The monogamy of $D_R$ for the measured party means, for pure $ABCD$,
\begin{align}
    &D_R(A|B)+D_R(A|C)\le D_R(A|BC), \\
    \Leftrightarrow\, &J_R(A|B) + J_R(A|C) - J_R(A|BC) \ge I(A:B:C), \\
    \Leftrightarrow\, & S_A -\frac{1}{2}(S_R(A:CD)+S_R(A:BD)-S_R(A:D))\ge I(A:B:C),
\end{align}
where $I(A:B:C)=S_A+S_B+S_C-S_{AB}-S_{BC}-S_{CA}+S_{ABC}$ is the tripartite information~\cite{Casini:2008wt}. Since this quantity vanishes for pure $ABC$, showing a violation to the measured-party monogamy of $D_R$ is equivalent to showing a violation to the measured-party polygamy of $J_R$.

Consider the three-qubit W state (with empty $D$)
\begin{equation}
    \ket{\Psi}_{ABC} = \frac{1}{\sqrt{3}}\left(\ket{001} + \ket{010} + \ket{100}\right).
\end{equation}
Since $S_R(A:B) = S_R(A:C)$ and  $S_R(A:C)=\log 6 - \frac{1}{\sqrt{3}}\log(2+\sqrt{3}) > S_A=\log 3 -\frac{2}{3}\log 2$, the W state is neither polygamous for $J_R$ nor monogamous for $D_R$. Thus, $J_R$ is not necessarily polygamous and $D_R$ is not necessarily monogamous for the measured party.

\paragraph{Four-qubit state violations}
Let us now consider a four-qubit state that violates the measured-party monogamy of $J_R$. 
Consider a four-qubit pure state
\begin{equation}
    \ket{\Psi}_{ABCD} = \frac{1}{\sqrt{2}}\left(\ket{0100} + \ket{1011}\right)_{ABCD}.
\end{equation}
The measured-party monogamy of $J_R$ is equivalent to
\begin{equation}
    2S_A + S_R(A:D) \stackrel{?}{\le} S_R(A:BD) + S_R(A:CD).
\end{equation}
By symmetry and straightforward computation,
\begin{equation}
    S_A = S_R(A:D) = S_R(A:BD) = S_R(A:CD) = \log 2,
\end{equation}
so that
the monogamy inequality is violated by $\log 2$.

To show that $D_R$ is neither polygamous, consider
\begin{equation}
    \ket{\Psi}_{ABCD} = \frac{1}{2}\left(\ket{0000} + \ket{0011} + \ket{0110} + \ket{1101}\right).
\end{equation}
Since
\begin{equation}
    S_B = S_C = S_D = \log 2, \qquad S_{BC} = S_{BD} = S_{CD} = \frac{3}{2}\log 2,
\end{equation}
\begin{align}
    S_R(A:CD) &= S_R(A:BD) =3\log 2 - \frac{5}{8}\log 5,\\
    S_R(A:D) &=\frac{3}{2}\log 2 + \frac{1}{2\sqrt{2}}\log\frac{2-\sqrt{2}}{2+\sqrt{2}}.
\end{align}
The discord polygamy inequality is
\begin{equation}
    D_R(A|B) + D_R(A|C) \stackrel{?}{\le} D_R(A|BC),
\end{equation}
but substituting the above values violates the inequality. Thus, this state is neither monogamous for reflected classical correlations nor polygamous for reflected discord; this is the same as the original quantities.

\subsubsection{Unmeasured Party}
Just as in the measured-party case, we find that monogamy and polygamy of boundary correlation measures \(J_R\) and \(D_R\) tend to fail in complementary pairs. In particular, we again observe that states violating the monogamy of one quantity often violate the polygamy of the other. We present two examples that exhibit this complementary pattern for the unmeasured party.

\paragraph{Tripartite pure state violations}

We first examine the discord monogamy inequality:
\begin{equation}
    S_R(C:AD) + S_R(A:CD) + 2S_B + 2S_D \stackrel{?}{\le} 2S_{AD} + 2S_{CD} + S_R(AC:D).
\end{equation}
Taking \(D = \emptyset\), this reduces to:
\begin{equation}
    S_R(A:C) \stackrel{?}{\le} I(A:C),
    \label{reflected violation unmeasured}
\end{equation}
which is the negation of the known upper bound to the reflected entropy \cite{Dutta:2019geu}, and therefore not generally satisfied. This provides a direct counterexample to monogamy of discord.

To show that classical correlations are not polygamous either, we again take \(D = \emptyset\), reducing the inequality to \eqref{reflected violation unmeasured} which is again the negation of the known upper bound and thus not true. Hence, classical correlation is neither monogamous nor polygamous with respect to the unmeasured party. However we note that for pure \(ABC\), classical correlation is monogamous.

\paragraph{Four-qubit state violations}

For classical correlations, the monogamy inequality is given by:
\begin{equation}\label{eq: mono discord?}
    2I(A:C) + S_R(AC:D) \stackrel{?}{\le} S_R(A:CD) + S_R(C:AD).
\end{equation}
Consider the four-qubit state:
\begin{equation}
    \ket{\Psi}_{ABCD} = \frac{1}{\sqrt{2}}\left(\ket{0001}+\ket{1110}\right),
\end{equation}
with the ordering of subsystems given by \(ABCD\). Then we have, in base-2 logarithms,
\begin{equation}
    I(A:C) = S_R(AC:D) = S_R(A:CD) = S_R(C:AD) = 1,
\end{equation}
thus violating the inequality. The values of the reflected entropy follow from tracing out \(B\), yielding:
\begin{equation}
    \rho_{ACD} = \frac{1}{2}\left(\dyad{001} + \dyad{110}\right),
\end{equation}
which is symmetric under all system permutations with appropriate qubit flipping. To show that discord is also not polygamous, consider the state:
\begin{equation}
    \ket{\Psi}_{ABCD} = \frac{1}{\sqrt{3}}\left(\ket{0101}+\ket{0110}+\ket{1011}\right).
\end{equation}
By symmetry and standard reflected entropy properties, we find:
\begin{align}
    S_C = S_A = S_R(A:CD) = S_R(AC:D) = S_R(C:AD) &= \log3 - \frac{2}{3}\log2\\
    S_{AC} &= \log3
\end{align}
which substituting these into \eqref{eq: mono discord?} shows the LHS exceeds the RHS, violating polygamy of discord as well.

\section{Summary and future outlook}
In this work, we performed a systematic analysis of the classical and quantum correlation measures on both sides of the holographic duality. In particular, we examined whether the bulk measures $J_W,D_W$ and the boundary measures $J_R,D_R$, originally proposed in~\cite{Mori:2025gqe}, obey various properties of correlation measures such as monotonicity, monogamy, and strong superadditivity. Our study tests the information theoretic consistency of the holographic proposal through these zoo of correlation inequalities.

\begin{table}[htbp]
\centering
\small
\renewcommand{\arraystretch}{1.25}
\setlength{\tabcolsep}{6pt}

\newcolumntype{Y}{>{\centering\arraybackslash}X}

\begin{tabularx}{\linewidth}{|c|Y|Y|Y|Y|Y|}
\hline
 & \multicolumn{2}{c|}{\textbf{Monotonicity}}
 & \multicolumn{2}{c|}{\textbf{Monogamy/Polygamy}}
 & \makecell{\textbf{S.}\\\textbf{Superaddit.}} \\
\hline
\textbf{Party} & \textbf{Unmeasured} & \textbf{Measured}
              & \textbf{Unmeasured} & \textbf{Measured} &  \\
\hline
$\bm{J}$   & Known  & Known  & No      & No      & No \\
\hline
$\bm{J_W}$ & Proved & Proved
           & Proved & Proved & \makecell{1WAY\\ No Counter} \\
\hline
$\bm{J_R}$ & Proved & Counterex.~\cite{Hayden:2023yij}
           & Counterex. & Counterex. & \makecell{Proved for\\ Haar random} \\
\hline
\end{tabularx}

\caption{Summary of inequalities for classical correlation measures: $J$ (original), $J_W$ (holographic), and $J_R$ (reflected). S. Superaddit. stands for strong superadditivity.}
\label{tab:cc_ineq}
\end{table}

\begin{table}[htbp]
\centering
\small
\renewcommand{\arraystretch}{1.25}
\setlength{\tabcolsep}{6pt}

\newcolumntype{Y}{>{\centering\arraybackslash}X}

\begin{tabularx}{\linewidth}{|c|Y|Y|Y|Y|Y|}
\hline
 & \multicolumn{2}{c|}{\textbf{Monotonicity}}
 & \multicolumn{2}{c|}{\textbf{Monogamy/Polygamy}}
 & \makecell{\textbf{S.}\\\textbf{Superaddit.}} \\
\hline
\textbf{Party} & \textbf{Unmeasured} & \textbf{Measured}
              & \textbf{Unmeasured} & \textbf{Measured} &  \\
\hline
$\bm{D}$   & Known  & Counterex. & Neither & Neither & No \\
\hline
$\bm{D_W}$ & Proved & Counterex. & \makecell{Polygamous\\(pure)} & Counterex. & No \\
\hline
$\bm{D_R}$ & Proved & Counterex. & \makecell{Polygamous\\(pure)} & Counterex. & No \\
\hline
\end{tabularx}

\caption{Summary of inequalities for quantum discord measures: $D$ (original), $D_W$ (holographic), and $D_R$ (reflected).}
\label{tab:qd_ineq}
\end{table}

Based on topological arguments, our results (summarized in Tables~\ref{tab:cc_ineq} and \ref{tab:qd_ineq}) support the holographic proposal $J=J_W$ and $D=D_W$ by proving that $J_W$ is monotone for both parties and $D_W$ is monotone for the unmeasured party in arbitrary dimensions. We also found that $D_W$ is neither monotonic nor monogamous with respect to the measured party. All of these properties are satisfied by the original correlation measures $J$ and $D$, showing consistency with the proposed duality. 

We also found some special features of holographic correlations. We proved the barrier theorem (Theorem \eqref{thm:barrier}) which has interesting information theoretic consequences. Moreover, we proved that $J_W$ is monogamous for both unmeasured and measured parties and that $D_W$ is polygamous for the unmeasured party when written over a pure state. These features are special to holographic states; the original correlation measures do not satisfy such properties. Similar to monogamy of holographic mutual information, they provide additional necessary conditions for holographic states. For example, the polygamous nature of pure-state holographic discord indicates the absence of GHZ entanglement in these states. It also implies a trade-off relation among the monogamy of mutual (total) correlation, the monogamy of classical correlation, and the polygamy of quantum discord in holography.

Another perspective for the monogamy of $J_W$ stems from the one-way strong superadditivity of one-way distillable entanglement. In~\cite{Mori:2024gwe}, it has been conjectured that $J_W$ is further dual to the one-way distillable entanglement $E_D^{\rm[1WAY]}$, which quantifies the number of EPR pairs that can be distilled via one-way LOCC. From operational reasons, if the proposed duality holds, $J_W$ should obey the one-way strong superadditivity, which further implies monogamy. Indeed, we proved the one-way strong superadditivity of $J_W$, providing a piece of evidence for the conjectured duality $J_W=E_D^{\rm[1WAY]}$. 

On the conjecture between $J_W$ and distillable entanglement, we further ask if $J_W$ could potentially be dual to $E_D\equiv E_D^{\rm[2WAY]}$, the distillable entanglement under two-way LOCC. While the earlier study indicates a sequence of one-way LOCC with randomized measurements would not improve distillability, whether it can be improved under generic two-way LOCC remains open. One supporting evidence comes from the \emph{two-way} strong superadditivity, which $E_D^{\rm[1WAY]}$ does \emph{not} necessarily have but $E_D^{\rm[2WAY]}$ does. 
We found that \emph{two-way} strong superadditivity holds for Haar random states, which are the simplest toy model of holographic states. Also, in pure AdS$_3$, numerically we find no counterexample. Thus, it is tempting to prove the two-way strong superadditivity geometrically in holography. Nevertheless, we also find that while topological proofs succeed in some cases, one remaining case requires a proof beyond a mere topological argument. We leave this conjecture of (two-way) strong superadditivity of $J_W$ as a future problem.

On the technical side, the topological method we employed in this paper works in any dimension and is powerful enough to deal with holographic inequalities involving multiple EWCSs. While there have been numerous such geometric approaches in the past, to our knowledge, only a few papers have developed proofs with such a rigour, especially when multiple EWCSs associated to different subregions are included. 
Compared to entropies as in the holographic entropic cone program, inequalities involving multiple EWCS are less explored, and the homological tools developed and formalized in this paper would be useful in future systematic exploration.

~

Going back to the boundary side, we also propose the computable `reflected' measures $J_R$ and $D_R$ based on the duality between the reflected entropy and EWCS. While the duality itself only holds for holographic states, these reflected measures themselves now become computable in any quantum systems without optimizations. We find the reflected measures are similar to the original measures. Namely, both $J_R$ and $J$ are not monogamous with respect to either party; both $D_R$ and $D$ are monotone for the unmeasured party while not for the measured party and not measured-party monogamous. However, regarding the unmeasured-party polygamy for pure states, $D_R$ is similar to the holographic one. Finally, $J_R$ shows a distinct feature regarding the measured-party monotonicity. It exhibits a counterexample, contrary to the proven monotonicity for $J$ and $J_W$.

~

As a future outlook, a natural direction is the systematic exploration of the ``EWCS cone''. While the Holographic Entropy Cone program characterizes the constraints on von Neumann entropies in holographic states, our results suggest there may be even more, unexplored inequalities involving both multiple EWCSs and HEEs. 
A systematic search for such inequalities, perhaps using the techniques developed here, could reveal the broader space of mixed-state correlations in holography. This would be particularly relevant for defining and constraining multipartite generalizations of discord, where the interplay of multiple EWCSs becomes important.

Furthermore, extending our results to time-dependent spacetimes remains a critical task. Our current proofs are based on the RT minimization on a common time slice. Investigating whether properties like monotonicity, monogamy, and strong superadditivity hold in covariant settings using the HRT formula is essential, particularly in the context of black hole evaporation or quenches.

\acknowledgments
Research at Perimeter Institute is supported in part by the Government of Canada through the Department of Innovation, Science and Economic Development and by the Province of Ontario through the Ministry of Colleges, Universities, Research Excellence and Security.
KL's and HW's work was also supported by Michael Serbinis and Laura Adams through the PSI Start research internship.
This work was supported by JSPS KAKENHI Grant Number 23KJ1154, 24K17047.

\appendix

\section{Summary of notation}\label{s: summary of notation}

We collect the entropic quantities and topological notations used throughout.
Unless stated otherwise, $S_X$ denotes the von-Neumann entropy of $\rho_X$:
$S_X \equiv S(\rho_X):=-\Tr(\rho_X\log\rho_X)$, where $\rho_X=\Tr_{\bar X}\rho$.

\paragraph{Information-theoretic quantities}
\begin{itemize}
\setlength\itemsep{2pt}
     \item Classical conditional entropy~\eqref{eq:clas-cond-ent}: $J_\Pi(A|B) = S_A - \sum_x p_x S(\rho_A^x)$
    \item Classical correlation~\eqref{eq:clas-corr}: $J(A|B)=\max_\Pi J_\Pi (A|B)$ 
    \item Mutual information: $I(A:B)=S_A+S_B-S_{AB}=S_A-S(A|B)$
      \item Entanglement of formation~\eqref{eq:EoF}: $E_F(A:C) = \inf \sum_x p_x S_A(\dyad{\psi_x})$
    \item Quantum discord~\eqref{eq:discord}: $D(A|B)=I(A:B)-J(A|B)$
    \item Entanglement wedge cross section (EWCS)~\eqref{eq:EWCS}: $E_W(A:C) = \min_{\Gamma_{A:C}} \frac{\mathrm{Area}(\Gamma_{A:C})}{4G_N}$
    \item Holographic classical correlation~\eqref{Holographic definitions}: $J_W(A|B) \equiv S_A - E_W(A:C)$
    \item Holographic discord~\eqref{Holographic definitions}: $D_W(A|B) \equiv S_B - S_C + E_W(A:C)$
    \item Reflected entropy~\eqref{eq:reflected entropy}: $S_R(A:C) = S_{AA^\ast}(\ket{\rho^{1/2}})$
    \item Canonical purification of $\rho$~\eqref{eq:cano-purif}: $\ket{\rho^{1/2}}_{AA^\ast CC^\ast}$
    \item Reflected classical correlation~\eqref{Holographic definitions}: $J_R
    (A|B) \equiv S_A - \frac12S_R(A:C)$
    \item Reflected discord~\eqref{Holographic definitions}: $D_R(A|B) \equiv S_B - S_C + \frac12S_R(A:C)$
    \item One-shot distillable entanglement: (See Section~\ref{sec:ED-review} and~\cite{Mori:2024gwe} for the definition)
    
    $E_D^{\rm [1WAY]}(A:B)=\max(E_D^{\rm [1WAY]}(A\leftarrow B),E_D^{\rm [1WAY]}(B\leftarrow A))$
    \item Symmetrized classical correlation~\eqref{eq:1shot-ED}: $J_W(A:B)\equiv\max(J_W(A|B), J_W(B|A))$
    \item Holographic one-way distillable entanglement conjecture~\eqref{eq:ED=JW}: $E_D^{\rm [1WAY]}(A\leftarrow B)=J(A|B)$
\end{itemize}

\paragraph{Topological quantities.} We follow the convention that codimension is counted in the full spacetime;
since all objects lie on the Cauchy slice $\Sigma$, a spacetime codimension-two surface
appears as a codimension-one hypersurface in $\Sigma$.
\begin{itemize}
\setlength\itemsep{2pt}
    \item $\Sigma$ (Definition~\ref{def:bulk}): Cauchy slice (codimension-one spacelike slice) 
     \item For a codimension-one bulk region $U\subseteq\Sigma$:
    \begin{itemize}
        \setlength\itemsep{2pt}
        \item $\partial_\Sigma U\equiv \partial U$ (Definition~\ref{def:boundary}): boundary of $U$ within $\Sigma$, excluding asymptotic boundary segments.
        \item $\partial_\infty U$ (Definition~\ref{def:boundary}): boundary of $U$ on the asymptotic boundary of $\Sigma$.
    \end{itemize}
   
    \item $\gamma_A$: the Ryu-Takayanagi (RT) surface for a boundary subsystem $A$.
    \item $\mathcal{E}(A)$ (Definition~\ref{def:ewedge}): entanglement wedge of $A$ (a codimension-one region in $\Sigma$), with
    $\partial_\infty \mathcal{E}(A)=A$ and $\partial_\Sigma \mathcal{E}(A)=\gamma_A$.
    \item $\cup,\cap$: union or intersection of subsets. 
    \item $\sqcup$: disjoint union of subsets. This means each subset is individually put in the set so that the multiplicity is taken into account. We only use disjoint union topologically when the arguments have zero intersect when considered in $\Sigma$.
        \item Wedge decomposition (Definition~\ref{def:bridged-wedges}, see also Fig.~\ref{fig: bridged wedges}). Decompose as 
         \[
            \mathcal{E}(AB)=\bigsqcup_\alpha W_\alpha,
        \]
        where each $W_\alpha$ is a connected component.
    \begin{itemize}
        \setlength\itemsep{2pt}
        \item Component types:
        $W_\alpha$ is \emph{$A$-only} (resp.\ \emph{$B$-only}) if its asymptotic boundary lies entirely in $A$ (resp.\ $B$);
        it is \emph{$A$--$B$ bridged} if its asymptotic boundary meets both $A$ and $B$. 
        \item Boundary sets induced by the decomposition:
        \begin{itemize}
            \setlength\itemsep{2pt}
            \item $\mathcal{A}\subseteq A$ (resp.\ $\mathcal{B}\subseteq B$): the union of boundary components that are $A$-only (resp.\ $B$-only) components in $\mathcal{E}(AB)$.
            \item $\mathcal{F}_A := A\setminus \mathcal{A}$ and $\mathcal{F}_B := B\setminus \mathcal{B}$: the bridged boundary portions on each side; $\mathcal{F}:=\mathcal{F}_A\cup \mathcal{F}_B$.
        \end{itemize}
        \item $\mathcal{E}_{\rm brid}(A:B)$: $(A:B)$-bridged entanglement wedge, defined as the union of all $A$--$B$ bridged components.
    \end{itemize}
    \item Homology notation:
    $U\sim A$ means $U$ is homologous to $A$ (see \ref{Homology});
    $U\sim_V A$ means homologous to $A$ within a bulk subregion $V\subseteq \Sigma$.
    \item $\operatorname{Int}\mathcal{E}(A)$: open bulk interior of $\mathcal{E}(A)$ (see \ref{Interior}).
    \item $\ell$ is a simple path (Definition \ref{simple path}).
     \item RT admissible class (Definition~\ref{def: RT admissable}):
    $\mathfrak{S}_{A}$ is the class of codimension-2 surfaces homologous to $A$; its minimizer is $\gamma_A$, whose area computes $S_A$.
    \item EWCS admissible class (Definition~\ref{EWCS admissible class}):
    $\mathfrak{S}_{A:B}$ is the class of codimension-2 surfaces in $\mathcal{E}(AB)$ separating $A$ and $B$; its minimizer is $\Gamma^W_{A:B}$, whose area computes $E_W(A:B)$.
   
\end{itemize}

\section{Mathematica code for strong superadditivity of Haar random states}\label{app:code}

\begin{lstlisting}
ClearAll[m];
m[x_, y_] := (x + y - Abs[x - y])/2;

ClearAll[a, b, c, d, e];
vars = {a, b, c, d, e};

(*Constraints*)
cons = And @@ Thread[vars > 0] && a + b + c + d + e == 1;

ab = a + b;
cd = c + d;
de = d + e;
ce = c + e;
be = b + e;
ae = a + e;

abe = a + b + e;
ace = a + c + e;
bde = b + d + e;
abce = a + b + c + e;
bcde = b + c + d + e;
acde = a + c + d + e;
abde = a + b + d + e;
cde = c + d + e;

(*J_W assumptions*)

A1 = m[ab, cde] - m[ab, e] >= m[cd, abe] - m[cd, e];
A2 = m[a, bcde] - m[a, bde] >= m[c, abde] - m[c, bde];
A3 = m[d, abce] - m[d, ace] >= m[b, acde] - m[b, ace];

assumptionsAll = cons && A1 && A2 && A3;

(*Inequality*)

goal = m[ab, cde] + m[a, bde] + m[d, ace] >= 
   m[a, bcde] + m[d, abce] + m[ab, e];

(*Formal proof*)
proof1 = Resolve[ForAll[vars, Implies[assumptionsAll, goal]], Reals];

(*Show NO counterexample: if False empty set of counter examples*)
noCounterEx = Reduce[assumptionsAll && Not[goal], vars, Reals]
\end{lstlisting}

\section{Counterexample to monotonicity of reflected discord in measured party}\label{app: Counter example to monotonicity of reflected discord in the measured party}
We consider the Werner state
\begin{equation}
    \rho_{AC} = p \dyad{\Phi^+}_{AC} + (1-p)\frac{I_{AC}}{4},
\end{equation}
and show
\begin{equation}\label{eq: app to vio}
    I(A:C) \stackrel{?}{\ge} \frac{S_R(A:C)}{2}.
\end{equation}
is not true for all $0\le p\le 1$, thus showing the reflected discord is not monotone with respect to its measured party (see Section \ref{sss: Monotonicity for the measured party}). The mutual information is given by
\begin{equation}
I(A:C) = \frac{1}{4}\left(-3(-1 + p)\log(1 - p) + (1 + 3p)\log(1 + 3p)\right).
\end{equation}
To find $S_R(A:C)$ we first canonically purify to
\begin{equation}
    \ket{\rho_{ACA^*C^*}} = \sqrt{\frac{1-p}{4}}\Big(\ket{\Phi^-\Phi^-} + \ket{0101} + \ket{1010}\Big) + \sqrt{\frac{1+3p}{4}}\ket{\Phi^+\Phi^+},
\end{equation}
with qubit order $ACA^*C^*$ and $\ket{\Phi^\pm\Phi^\pm}=\ket{\Phi^\pm}\otimes\ket{\Phi^\pm}$. Tracing out $CC^*$ we find
\begin{align}
\rho_{AA^*}
&= \frac{3 - p + \sqrt{(1-p)(1+3p)}}{8}\,\big( \dyad{00} + \dyad{11} \big)
\\&+ \frac{1 + p - \sqrt{(1-p)(1+3p)}}{8}\,\big( \dyad{01} + \dyad{10} \big)
\\&+ \frac{1 - p + \sqrt{(1-p)(1+3p)}}{4}\,\big( \ketbra{00}{11} + \ketbra{11}{00} \big).
\end{align}
It is block-diagonal with respect to the basis $\{\ket{00},\ket{11}\}\oplus\{\ket{01},\ket{10}\}$ and so we find
\[
S_R(A:C)
= - \lambda_{1} \log \lambda_{1}
  - 3\,\lambda_{2} \log \lambda_{2},
\]
with
\[
\lambda_{1} =
\frac{5 - 3p + 3\sqrt{(1-p)(1+3p)}}{8},
\qquad
\lambda_{2} =
\frac{1 + p - \sqrt{(1-p)(1+3p)}}{8}.
\]
The original inequality \eqref{eq: app to vio} is false for $0\leq p\leq p_\text{max}$ with $p_\text{max}\approx 0.41598$.

\section{Additional tools for proofs}
\label{s: complete proofs}

For checking EWCS admissibility of a surface, it is more constructive to consider the behavior of paths between the arguments of the EWCS and whether they must cross said surface. We formalize this as follows.

\begin{definition}[Simple path]\label{def:C1-curves} 
We define a simple path on $\Sigma$ as an injective continuous map
\[
\ell\colon [0,1] \to \Sigma.
\]
This means there are no self-intersections but the path can generally be jagged, have non-smooth corners and the like.
\label{simple path}
\end{definition}

\begin{definition}[Path adjancy to a subsystem]
Let $\Sigma$ be the bulk Cauchy slice with regulated asymptotic boundary
$\partial_\infty\Sigma$. Let $A$ be a boundary subsystem ($A\subset\partial_\infty\Sigma$) and set
$A^\circ:=\mathrm{int}_{\partial_\infty\Sigma}(A):=A\setminus\partial A$.
Fix a collar embedding $c:\partial_\infty\Sigma\times[0,\varepsilon_0)\hookrightarrow\Sigma$
with $c(p,0)=p \in \Sigma$ where $\varepsilon_0 >0$ is some small number. For $0<\varepsilon<\varepsilon_0$ define
\[
\mathrm{Adj}_\varepsilon(A):=c\!\left(A^\circ\times(0,\varepsilon)\right)\subset\Sigma.
\]
We say $x\in\Sigma$ is adjacent to $A$ if $x\in\mathrm{Adj}_\varepsilon(A)$.
A simple path $\ell:[0,1]\to\Sigma$ starts (ends) adjacent to $A$ if
$\ell(0)\in\mathrm{Adj}_\varepsilon(A)$ (resp.\ $\ell(1)\in\mathrm{Adj}_\varepsilon(A)$).
\end{definition}

\begin{remark}[Codimension on $\Sigma$.]
We count codimension with respect to the full spacetime. 
Since all surfaces considered here lie on the Cauchy slice $\Sigma$ (see Assumption~\ref{assump:sym}), 
any spacetime codimension-two surface $\gamma\subset\Sigma$ is a codimension-one hypersurface when considered in $\Sigma$ 
and may therefore separate $\Sigma$ and intersect $\Sigma$-paths.
\end{remark}

\begin{remark}[Path--separation formulation]\label{rem:path-separation}
Condition \emph{(ii)} in Definition~\ref{def: EWCS admissable} can be equivalently phrased
as a path--separation property in terms of simple paths. Namely, let $\Gamma\in\mathfrak S_{A:B}$ be a properly embedded, piecewise smooth codimension-$2$ hypersurface
with $\partial_\Sigma\Gamma\subset\gamma_{AB}$.
Then condition \emph{(ii)} is equivalent to the requirement that
every simple path
\[
\ell \subset \mathcal E(AB)
\]
whose endpoints lie on opposite sides of $\Gamma$
(for instance, one endpoint in a neighborhood of $A$ and the other in a neighborhood of $B$ on the boundary)
must intersect $\Gamma$ at least once.\footnote{These paths are similar to bit threads, discussed in a slightly different context~\cite{Freedman:2016zud}. The path-separating surface there is interpreted as a RT surface after optimization.}

Conversely, if a smooth, embedded hypersurface
$\Gamma\subset\mathcal E(AB)$ with bulk anchoring
$\partial_\Sigma\Gamma\subset\gamma_{AB}$ has the property that every such simple path from the $A$-side to the $B$-side intersects $\Gamma$, then the complement
$\mathcal E(AB)\setminus\Gamma$ decomposes as
\begin{equation}
\mathcal E(AB)\setminus\Gamma = U_A \sqcup U_B
\end{equation}
with $\partial_\infty U_A=A$ and $U_B\partial_\infty =B$. Then $\Gamma$ is EWCS admissible and in particular
\begin{equation}
    \frac{\Area(\Gamma)}{4G_N} \geq E_W(A:B),
\end{equation}
by minimality. Thus checking the intersection properties of specific simple paths are sufficient to prove EWCS admissibility. 
\end{remark}

\begin{figure}[ht]
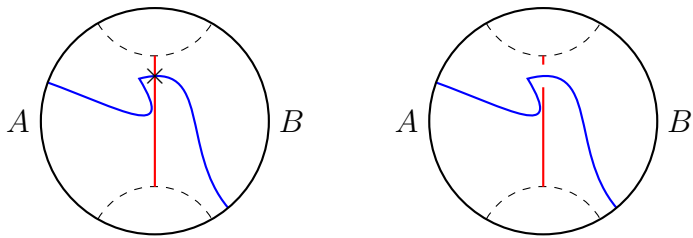

    \centering
\scalebox{1}{
\input{Sketch_Figures/connectedness/LHS}
\hspace{0.5cm}
\input{Sketch_Figures/connectedness/RHS}
}
    \caption{Dashed lines are $\gamma_{AB}$, the red line $\Gamma$ is some surface we are checking for $(A:B)$ EWCS admissibility, and the blue some simple path $\ell$. On the left, $\Gamma$ is EWCS admissible as all choices of $\ell$ intersect $\Gamma$ when traveling from $A$ to $B$ within $\mathcal{E}(AB)$. On the right, it is not however as $\ell$ can start adjacent to $A$ and reach $B$ without crossing $\Gamma$.} 
    \label{fig: crossing of C1 curves}
\end{figure}

Our proofs below rely on comparing various surfaces on the cauchy slice. In particular the proof of equation \eqref{eq: simple inequality to prove}, worked well as given the constraint of a bridged wedge when we defined $\mathbf{L}$ in \eqref{eq: first mathbfl definition} we had that $\gamma_{AB}\cap\Gamma^W_{A:B}=0$ (except points) hence
\[
\frac{\Area(\mathbf{L})}{4G_N} = \frac{\Area(\gamma_{AB}\cup\Gamma^W_{A:B})}{4G_N} = \frac{\Area(\gamma_{AB})}{4G_N} + \frac{\Area(\Gamma^W_{A:B})}{4G_N} = S_{AB} + E_W(A:B).
\]
But in doing this we have naturally embedded both $\gamma_{AB}$ defining $S_{AB}$, and $\Gamma^W_{A:B}$ defining $E_W(A:B)$ onto the same Cauchy slice and computed areas from there. Of course if they have non-zero intersect then formally union removes multiplicity. As our proofs are strict numerical inequalities we therefore need to keep track of the multiplicity of surfaces when translating from the numerical measures to their topological description. This motivates the following definition.

\begin{definition}[Sheeted slice, multisurfaces, and multiplicity bookkeeping]
\label{def:multisurface_sheeted_simple}

Fix the bulk slice $\Sigma$. For each $N\in\mathbb{Z}^+$ let
\[
\widehat{\Sigma}_N := \bigsqcup_{n=1}^N \Sigma^{(n)}
\]
be the disjoint union of $N$ isometric copies (“sheets”) of $\Sigma$.
Write $\iota_n:\Sigma\hookrightarrow \widehat{\Sigma}_N$ for the canonical inclusion into sheet $n$ and
$\pi:\widehat{\Sigma}_N\to\Sigma$ for the projection forgetting the sheet label.

\begin{figure}[ht]
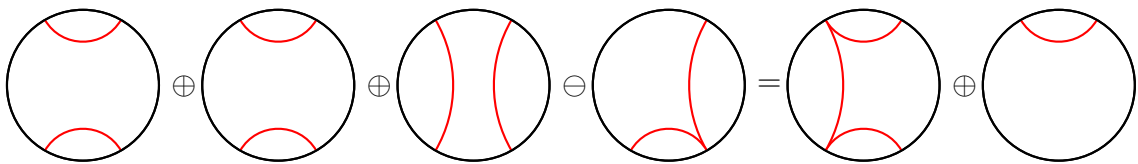

    \centering
\scalebox{1}{
\input{Sketch_Figures/multisurface/fig1}
\hspace{-0.4cm}
\begin{tikzpicture}
   \node at (0,0.87) {$\oplus$};
   \node at (0,0) {};
\end{tikzpicture}
\hspace{-0.35cm}
\input{Sketch_Figures/multisurface/fig2}
\hspace{-0.4cm}
\begin{tikzpicture}
   \node at (0,0.87) {$\oplus$};
     \node at (0,0) {};
\end{tikzpicture}
\hspace{-0.35cm}
\input{Sketch_Figures/multisurface/fig3}
\hspace{-0.4cm}
\begin{tikzpicture}
   \node at (0,0.87) {$\ominus$};
     \node at (0,0) {};
\end{tikzpicture}
\hspace{-0.35cm}
\input{Sketch_Figures/multisurface/fig4}
\hspace{-0.4cm}
\begin{tikzpicture}
   \node at (0,0.87) {$=$};
   \node at (0,0) {};
\end{tikzpicture}
\hspace{-0.35cm}
\input{Sketch_Figures/multisurface/fig5}
\hspace{-0.4cm}
\begin{tikzpicture}
   \node at (0,0.87) {$\oplus$};
   \node at (0,0) {};
\end{tikzpicture}
\hspace{-0.35cm}
\input{Sketch_Figures/multisurface/fig6}
}
    \caption{Example of multisurface addition and subtractions for some a representatives of four multisurfaces.}
    \label{fig: multisurface definition}
\end{figure}

\noindent
Throughout, “surface” means a \emph{connected} embedded codimension-two hypersurface\footnote{Recall we use codimension in terms of the full spacetime, not the Cauchy slice.} in $\Sigma$ of finite area
(we identify surfaces that differ only on sets of codimension-two measure zero).
A disconnected surface is always understood as the finite disjoint union of its connected components.

\smallskip
\noindent\textbf{Multisurface (representative).}
A \emph{representative} of a multisurface is a finite embedded union
\[
X \;=\; \bigsqcup_{j=1}^k \iota_{s(j)}(\gamma_j)\ \subset\ \widehat{\Sigma}_N,
\]
where each $\gamma_j\subset\Sigma$ is a connected surface and $s(j)\in\{1,\dots,N\}$.
(Thus multiple copies of the same $\gamma$ are represented by placing them on different sheets, possibly with other
surfaces on the same sheet.)

\smallskip
\noindent\textbf{Equivalence.}
Two representatives $X\subset\widehat{\Sigma}_N$ and $X'\subset\widehat{\Sigma}_{N'}$ are \emph{equivalent} if,
after possibly adding empty sheets to either side, one can relabel sheets to match them.
Concretely: there exists $L\ge\max\{N,N'\}$ and a permutation $\sigma\in S_L$ such that, viewing both as subsets of
$\widehat{\Sigma}_L$ (by adding empty sheets),
\[
X' \;=\; \widehat{\sigma}(X),
\]
where $\widehat{\sigma}$ is the induced isometry of $\widehat{\Sigma}_L$ permuting sheets by $\sigma$.
A \emph{multisurface} $\mathbf{X}$ is an equivalence class of such representatives.

\smallskip
\noindent\textbf{Support.}
For any representative $X\subset\widehat{\Sigma}_N$ of $\mathbf{X}$ define the support in the physical slice by
\[
|\mathbf{X}| \;:=\; \pi(X)\subset\Sigma.
\]
This is independent of the representative (sheet relabeling does not change $\pi$).

\smallskip
\noindent\textbf{Multi-union (stacking).}
Given multisurfaces $\mathbf{X},\mathbf{Y}$ with representatives
$X\subset\widehat{\Sigma}_N$ and $Y\subset\widehat{\Sigma}_M$, define $\mathbf{X}\oplus\mathbf{Y}$ to be the class of
\[
X \sqcup Y \;\subset\; \widehat{\Sigma}_{N+M},
\]
where $X$ is placed on the first $N$ sheets and $Y$ on the last $M$ sheets. This is well-defined up to equivalence. If $\gamma$ is some surface on $\Sigma$ we define $\mathbf{X}\oplus\gamma$ as above taking $\gamma$ to be a multisurface of multiplicity one.

\smallskip
\noindent\textbf{Difference (removing copies).}
We say $\mathbf{Y}$ is \emph{removable from} $\mathbf{X}$, and write $\mathbf{Y}\preceq\mathbf{X}$, if there exist
representatives $X,Y\subset\widehat{\Sigma}_L$ on a common sheet number $L$ such that $Y\subset X$ as subsets.
In that case define
\[
\mathbf{X}\ominus\mathbf{Y}
\]
to be the multisurface represented by the set-theoretic difference $X\setminus Y\subset\widehat{\Sigma}_L$.
This is well-defined: different choices of representatives with $Y\subset X$ differ only by adding empty sheets and
permuting sheets, which preserves the equivalence class of $X\setminus Y$.

\smallskip
\noindent\textbf{Area with multiplicity.}
For a representative $X\subset\widehat{\Sigma}_N$ define
\[
\Area^\oplus(\mathbf{X}) \;:=\; \Area(X) \;=\; \sum_{j=1}^k \Area(\gamma_j),
\]
i.e.\ the ordinary area computed in the disjoint union $\widehat{\Sigma}_N$.
Then $\Area^\oplus(\mathbf{X}\oplus\mathbf{Y})=\Area^\oplus(\mathbf{X})+\Area^\oplus(\mathbf{Y})$, and if
$\mathbf{Y}\preceq\mathbf{X}$ then $\Area^\oplus(\mathbf{X}\ominus\mathbf{Y})
=\Area^\oplus(\mathbf{X})-\Area^\oplus(\mathbf{Y})$.
\end{definition}

\section{Proof: Barrier theorem}
\label{Barrier proof}
\barriertheorem*
Here we show the barrier theorem via two methods: reflected entropy (in general dimensions), and via cycles on a two-dimensional slice.

\subsection{Proof by reflected entropy}

\begin{proof}[Proof of Theorem~\ref{thm:barrier} (via reflected entropy)] The simplest proof is seen by EWN and the relation in \eqref{eq: ew = 2sr}. We will show that $\Gamma^W_{A:B}\cap\operatorname{Int}\mathcal{E}(A)=\varnothing$ and then the $B$ statement follows by symmetry.

\medskip\noindent
\textbf{Step 1: Reduction to a single bridged component.} If $A$ and $B$ have no $(A{:}B)$-bridged entanglement wedge, then $E_W(A{:}B)=0$ and we may take $\Gamma^W_{A:B}=\varnothing$, so the claim is trivial. Hence assume $A$ and $B$ have at least one bridged component.
Distinct bridged components of $\mathcal{E}(AB)$ are disjoint and the EWCS minimizer decomposes additively across them, so it suffices to work in a single connected bridged wedge, which we continue to denote by $\mathcal{E}(AB)$.

\medskip\noindent
\textbf{Step 2: Canonical purification and the doubled wedge geometry.}
Let $\rho_{AB}$ be the boundary state on $AB$. Consider its canonical purification
$\ket{\sqrt{\rho_{AB}}}\in\mathcal{H}_{AB}\otimes\mathcal{H}_{A'B'}$,
where $A',B'$ are isomorphic purifying systems. In the holographic, time-reflection-symmetric setup, we model the bulk
Cauchy slice dual to $\ket{\sqrt{\rho_{AB}}}$ by the \emph{double} of $\mathcal{E}(AB)$ along its RT surface:
\begin{equation}\label{eq:doubled-slice}
\widetilde{\Sigma}
\;:=\;
\mathcal{E}(AB)\ \cup_{\ \gamma_{AB}\ }\ \mathcal{E}(AB)'\,,
\end{equation}
where $\mathcal{E}(AB)'$ is an isometric copy of $\mathcal{E}(AB)$, and the two copies are glued along their common
codimension-$1$ boundary
\[
\gamma_{AB}=\partial_\Sigma\mathcal{E}(AB)
\]
(i.e.\ the RT surface of $AB$ on $\Sigma$). The resulting $\widetilde{\Sigma}$ carries a $\mathbb{Z}_2$ involutive
isometry $\mathcal{R}$ exchanging the two copies (primed $\leftrightarrow$ unprimed) and fixing $\gamma_{AB}$ pointwise. We regard the unprimed copy as an isometric
embedded submanifold of $\widetilde{\Sigma}$, so any subset of $\mathcal{E}(AB)$ is canonically identified with its image in
$\widetilde{\Sigma}$.

Now let $\Gamma\in\mathfrak{S}_{A:B}$ be any EWCS-admissible surface for $(A{:}B)$ inside the unprimed copy. Define its $\mathbb{Z}_2$-doubling by
\begin{equation}\label{eq:Gamma-tilde-def}
\widetilde{\Gamma}\;:=\;\Gamma\ \cup\ \mathcal{R}(\Gamma)\ \subset\ \widetilde{\Sigma}.
\end{equation}
Since $\partial_\Sigma\Gamma\subset \gamma_{AB}$, the doubled surface $\widetilde{\Gamma}$ is a closed embedded
codimension-$2$ hypersurface in $\widetilde{\Sigma}$ (any junction along $\gamma_{AB}$ is a set of codimension-$2$ measure
zero and does not affect area or admissibility). Moreover, since $\Gamma$ separates $\mathcal{E}(AB)$ into an $A$-side and
a $B$-side, the doubled surface $\widetilde{\Gamma}$ separates $\widetilde{\Sigma}$ into an $AA'$-side and a $BB'$-side.
Equivalently, $\widetilde{\Gamma}$ is RT-admissible (i.e.\ a valid homology competitor) for both boundary regions
$AA'$ and $BB'$ in the doubled geometry.

\medskip\noindent
\textbf{Step 3: Specialize to the EWCS minimizer and identify an RT surface.}
Now take $\Gamma=\Gamma^W_{A:B}$, the area-minimizing member of $\mathfrak{S}_{A:B}$. In the canonical purification,
the reflected entropy is defined as
\[
S_R(A{:}B)\;=\;S_{\sqrt{\rho_{AB}}}(AA') \;=\; S_{\sqrt{\rho_{AB}}}(BB'),
\]
which in the static holographic setting gives
\[
S_{\sqrt{\rho_{AB}}}(AA')=\frac{\Area(\gamma_{AA'})}{4G_N},
\]
where $\gamma_{AA'}\subset \widetilde{\Sigma}$ is the RT surface for $AA'$.
By the leading-order holographic relation in \eqref{eq: ew = 2sr} we have that
\begin{equation}\label{eq:area-identity-AAp}
\Area(\gamma_{AA'})
\;=\;
2\,\Area(\Gamma^W_{A:B})
\;=\;
\Area\!\bigl(\widetilde{\Gamma^W_{A:B}}\bigr),
\end{equation}
where $\widetilde{\Gamma^W_{A:B}}:=\Gamma^W_{A:B}\cup \mathcal{R}(\Gamma^W_{A:B})$.
But Step~2 showed $\widetilde{\Gamma^W_{A:B}}$ is RT-admissible for $AA'$, hence by RT minimality,
$\widetilde{\Gamma^W_{A:B}}$ is itself an RT minimizer for $AA'$ (possibly among several, if the RT surface is non-unique), and so we can identify
\begin{equation}\label{eq:choose-RT}
\gamma_{AA'}=\widetilde{\Gamma^W_{A:B}}.
\end{equation}
This is seen in the central panel of Figure \ref{Fig: barrier by SR}.

\begin{figure}[ht]
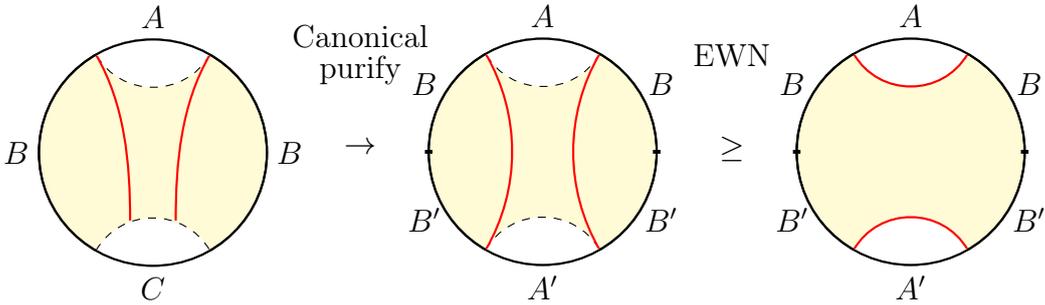

\centering
\scalebox{1}{
\input{Sketch_Figures/barrierNew/fig1}
\begin{tikzpicture}
  \node at (0,-0.5) {$\rightarrow$};
  \node[overlay, anchor=south] at (0,0.7) {Canonical};
  \node[overlay, anchor=south] at (0,0.2) {purify};
  \node at (0,-2.5) { };
\end{tikzpicture}
\input{Sketch_Figures/barrierNew/fig2}
\begin{tikzpicture}
  \node at (0,-0.5) {$\geq$};
  \node[overlay, anchor=south] at (0,0.45) {EWN};
  \node at (0,-2.5) { };
\end{tikzpicture}
\input{Sketch_Figures/barrierNew/fig3}
}
\caption{Demonstration of barrier theorem by reflected entropy and EWN. In the first figure we have $\Gamma^W_{A:B}$ naively traverses the yellow region $\mathcal{E}(B)$. Canonically purifying to $\rho_{ABA'B'}$ this becomes the reflected slice in central figure. This new red curve has length $2E_W(A:B) = S_R(A:B) = S_{BB'}$, by equation \eqref{eq: ew = 2sr}. However this violates EWN on this slice and so the original choice of $\Gamma^W_{A:B}$ was non-extremal.}
\label{Fig: barrier by SR}
\end{figure}

\medskip\noindent
\textbf{Step 4: Apply EWN in the doubled geometry.}
Since $\gamma_{AA'}$ is an RT surface for $AA'$ on $\widetilde{\Sigma}$, it lies on the boundary of the entanglement wedge
$\mathcal{E}_{\widetilde{\Sigma}}(AA')$ and therefore does not intersect its interior:
\begin{equation}\label{eq:RT-outside-interior}
\gamma_{AA'}\cap \Int\mathcal{E}_{\widetilde{\Sigma}}(AA')=\varnothing.
\end{equation}
However by EWN and our embedding
\begin{equation}
    \mathcal{E}_\Sigma(A) = \mathcal{E}_{\widetilde{\Sigma}}(A) \subseteq \mathcal{E}_{\widetilde{\Sigma}}(AA'),
\end{equation}
with equality iff the state on $A$ is pure, and so
\begin{equation}
    \widetilde{\Gamma^W_{A:B}} \cap \operatorname{Int}\mathcal{E}_{\widetilde{\Sigma}}(A) = \varnothing, \quad\Leftrightarrow \quad\Gamma^W_{A:B} \cap \operatorname{Int}\mathcal{E}_\Sigma(A) = \varnothing,
\end{equation}
proving the claim. The $B$ statement holds by symmetry and this concludes the proof.
\end{proof}

\subsection{Proof by cycles}\label{app:proof}

Throughout this subsection we work on a two-dimensional Cauchy slice $\Sigma$ and assume (as in the main text)
that $\mathcal{E}(AB)$ is a single connected bridged wedge. Namely, we decompose the EWCS minimizer into a disjoint union of its connected components
\begin{equation}
\Gamma^W_{A:B}=\bigsqcup_{k} \sigma_k ,
\end{equation}
where each $\sigma_k$ is a smooth embedded geodesic segment in $\Sigma$. Moreover if $A$ is disjoint on the boundary then its RT surface decomposes as 
\begin{equation}
    \gamma_A = \bigsqcup_i \gamma^A_i,
\end{equation}
with each $\gamma_i^A$ a geodesic.

Moreover, we define an \emph{excursion} into $A$ as follows. Fix a component $\sigma$ and parameterize it as an injective continuous map $\sigma:[0,1]\to\Sigma$
with $\sigma((0,1))$ in the bulk. Set
\[
T:=\{t\in(0,1):\sigma(t)\in \Int\mathcal{E}(A)\},
\]
such that if $T\neq\varnothing$, pick a connected component $(\tau_1,\tau_2)\subset T$ and set
\[
\alpha:=\sigma([\tau_1,\tau_2])\subset \mathcal{E}(A), \qquad
\alpha^\circ:=\sigma((\tau_1,\tau_2))\subset \Int\mathcal{E}(A).
\]
By maximality of $(\tau_1,\tau_2)$, the endpoints lie on $\partial\mathcal{E}(A)\cap\Sigma=\gamma_A$:
\begin{equation}\label{eq:excursion-endpoints-clean}
x:=\sigma(\tau_1)\in\gamma_i^A,\qquad y:=\sigma(\tau_2)\in\gamma_j^A,
\end{equation}
for some indices $i,j$ (possibly $i=j$). We term $\alpha^\circ$ the excursion into $A$.

We will prove that no component of $\Gamma^W_{A:B}$ enters $\Int \mathcal{E}(A)$.
The only nontrivial case is when a component enters $\Int\mathcal{E}(A)$ and exits through \emph{different}
components of $\gamma_A$; this will be ruled out by a cycle-and-surgery argument which is summarized in Figure \ref{fig: example i is not j}. By symmetry of $A\leftrightarrow B$ we hence show Theorem \ref{thm:barrier}.

Some important notations defined in this subsection are summarized and visualized in Figure~\ref{Fig: barrier by cycle}.

\begin{lemma}[Even-intersection lemma]\label{lem:even-intersection-app}
Fix a connected component $\gamma_i^A$ of $\gamma_A$. Then the number of bulk intersection points
$\gamma_i^A\cap \Gamma^W_{A:B}$ is even.
\end{lemma}

\begin{proof}
Because $\Gamma^W_{A:B}$ is EWCS-admissible for $(A{:}B)$, it separates $\mathcal{E}(AB)$ into
two regions $U_A$ and $U_B$ with $\partial_\infty U_A=A$ and $\partial_\infty U_B=B$.
Along $\gamma_i^A$, define a ``side label'' by declaring a nearby point to lie on the $A$-side or $B$-side
according to whether it lies in $U_A$ or $U_B$. This label can change only when $\gamma_i^A$ crosses
$\Gamma^W_{A:B}$. At both UV ends of $\gamma_i^A$ the label is the same (namely the $A$-side),
so the label must flip an even number of times. Hence $\#(\gamma_i^A\cap \Gamma^W_{A:B})$ is even.
\end{proof}

\begin{figure}[ht]
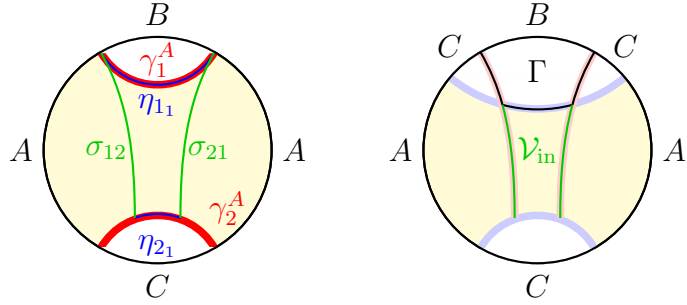

\centering
\scalebox{1}{
\input{Sketch_Figures/barrierCycle/fig1}
\hspace{0.4cm}
\input{Sketch_Figures/barrierCycle/fig2}
}
\caption{Notation used in proof of the barrier theorem by cycles. In both panels we have $\mathcal{E}(A)$ shaded in yellow. In the first panel the green curves are the conjectured $\Gamma^W_{A:B}$ and the blue curves are the resulting EWCS $\Gamma$ by \eqref{eq:Gammahat-def-clean}. By Lemma \ref{lem:strict-area-fixed} the area of the blue curves is strictly less than the area of the green curves. In the second panel, unlike the first, we have $\mathcal{V}_{\text{in}}\neq \Gamma^W_{A:B}$. The resulting EWCS by \eqref{eq:Gammahat-def-clean} is in black and has strictly less area than the original EWCS surface in salmon.}
\label{Fig: barrier by cycle}
\end{figure}

\begin{lemma}[Cycle extraction from a cross-component excursion]\label{lem:cycle-extraction-clean}
Assume there exists an excursion $\alpha^\circ\subset\Int\mathcal{E}(A)$ with endpoints
$x\in\gamma_i^A$ and $y\in\gamma_j^A$ for $i\neq j$. Then there exists a finite cyclic chain of EWCS components
\begin{equation}\label{eq:cycle-V-clean}
\mathcal{V}:=\{\sigma_{i_1 i_2},\,\sigma_{i_2 i_3},\,\dots,\,\sigma_{i_n i_1}\}\subset \Gamma^W_{A:B},
\qquad n\ge 2,
\end{equation}
such that each $\sigma_{i_k i_{k+1}}$ has a nontrivial excursion into $\Int\mathcal{E}(A)$ with endpoints on
$\gamma_{i_k}^A$ and $\gamma_{i_{k+1}}^A$ (and $i_k\neq i_{k+1}$ for all $k$).
Let $\mathcal{V}_{\mathrm{in}} := \mathcal{V}\cap\Int\mathcal{E}(A)$ denote the portions of these components lying in $\Int\mathcal{E}(A)$; then
$\mathcal{V}_{\mathrm{in}}\neq\varnothing$.
\end{lemma}

\begin{proof}
Build a finite multigraph $G$ whose vertices are the indices $i$ labeling components $\gamma_i^A$,
and whose edges are the EWCS components $\sigma_k$ that admit an excursion into $\Int\mathcal{E}(A)$:
an edge connects $i$ to $j$ if some excursion on $\sigma_k$ has endpoints on $\gamma_i^A$ and $\gamma_j^A$.

By assumption $G$ contains at least one edge between distinct vertices $i\neq j$.
Moreover, by Lemma~\ref{lem:even-intersection-app}, every vertex in $G$ has even degree
(indeed, excursions contribute bulk intersections on $\gamma_i^A$, and the parity constraint forces them to pair).
Hence, starting from the edge $i\!-\!j$ and walking along adjacent edges, one can never get stuck at a vertex,
and since $G$ is finite the walk must eventually repeat a vertex, producing a cycle of edges.
Translating back to EWCS components gives \eqref{eq:cycle-V-clean}. At least one edge of the cycle exists by assumption,
so $\mathcal{V}_{\mathrm{in}}\neq\varnothing$.
\end{proof}

\begin{lemma}[Closing the cycle on $\gamma_A$]\label{lem:C-simple-closed}
Let $\mathcal{V}$ be a cycle as in \eqref{eq:cycle-V-clean}, and let $\mathcal{V}_{\mathrm{in}}$ be its portions in
$\Int\mathcal{E}(A)$. For each vertex $i_k$ in the cycle, let
\[
p_{k,-}:=\gamma_{i_k}^A\cap \sigma_{i_{k-1}i_k},\qquad
p_{k,+}:=\gamma_{i_k}^A\cap \sigma_{i_ki_{k+1}}
\]
be the two bulk intersection points contributed by the cycle.
Let $\eta_{i_k}\subset\gamma_{i_k}^A$ be the (unique) sub-arc connecting $p_{k,-}$ to $p_{k,+}$, and set
\begin{equation}\label{eq:P-def-clean}
\mathcal{P}:=\bigcup_{k=1}^n \eta_{i_k}.
\end{equation}
Then
\begin{equation}\label{eq:C-def-clean}
\mathcal{C}:=\mathcal{V}_{\mathrm{in}}\ \cup\ \mathcal{P}
\end{equation}
is a simple closed curve contained in $\mathcal{E}(A)$, and hence bounds an open region
$U\subset\Int\mathcal{E}(A)$ with $\partial_\infty U=\varnothing$ (except points) and $\partial_\Sigma U=\mathcal{C}$.
\end{lemma}

\begin{proof}
By construction, $\mathcal{C}$ is closed: each geodesic segment in $\mathcal{V}_{\mathrm{in}}$ has endpoints on
$\gamma_A$, and the arcs $\eta_{i_k}$ join precisely the endpoints of adjacent segments in the cyclic order.

To see $\mathcal{C}$ is simple, note that distinct EWCS components are disjoint in the bulk
(otherwise one can cut-and-reglue to reduce total length, contradicting minimality of $\Gamma^W_{A:B}$),
and the arcs $\eta_{i_k}$ lie on distinct components of $\gamma_A$. Hence the only intersections occur at the shared
endpoints $p_{k,\pm}$.

Since $\mathcal{C}\subset \mathcal{E}(A)$ and does not reach $\partial_\infty\Sigma$, the Jordan curve theorem on the
regulated slice implies $\mathcal{C}$ bounds an open region $U$ with $\partial_\infty U=\varnothing$.
Moreover $U\subset \Int\mathcal{E}(A)$ because $\mathcal{C}$ meets $\partial\mathcal{E}(A)$ only along $\gamma_A$.
\end{proof}

\begin{lemma}[Surgery and EWCS-admissibility]\label{lem:Gammahat-admissible}
Let $\mathcal{C}$ and $U$ be as in Lemma~\ref{lem:C-simple-closed}. Define the surgically modified surface
\begin{equation}\label{eq:Gammahat-def-clean}
\Gamma
\;:=\;
\bigl(\Gamma^W_{A:B}\setminus \mathcal{V}_{\mathrm{in}}\bigr)\ \cup\ \mathcal{P}.
\end{equation}
Then $\Gamma$ is EWCS-admissible for $(A{:}B)$, i.e.\ $\widehat{\Gamma}\in\mathfrak{S}_{A:B}$.
\end{lemma}

\begin{proof}
Let $U_A,U_B$ be the complementary regions for $\Gamma^W_{A:B}$:
\[
\mathcal{E}(AB)\setminus \Gamma^W_{A:B}=U_A\sqcup U_B,\qquad \partial_\infty U_A=A,\ \partial_\infty U_B=B.
\]
By Lemma~\ref{lem:C-simple-closed}, the loop $\mathcal{C}$ bounds a bubble $U\subset\Int\mathcal{E}(A)$ with
$\partial_\infty U=\varnothing$. Replacing $\mathcal{V}_{\mathrm{in}}$ by $\mathcal{P}\subset\partial U$
shifts the separating surface across the bubble. Define
\[
\widehat{U}_A:=U_A\cup U,\qquad \widehat{U}_B:=U_B\setminus U.
\]
Since $\partial_\infty U=\varnothing$, we still have $\partial_\infty \widehat{U}_A=A$ and
$\partial_\infty \widehat{U}_B=B$. By construction, the common bulk boundary between $\widehat{U}_A$ and $\widehat{U}_B$
inside $\mathcal{E}(AB)$ is exactly $\Gamma$, hence $\Gamma$ separates $\mathcal{E}(AB)$ into an
$A$-side and a $B$-side with the correct asymptotic boundaries. Therefore $\Gamma$ is EWCS-admissible.
\end{proof}

\begin{lemma}[Strict area improvement]\label{lem:strict-area-fixed}
With $\Gamma$ as in \eqref{eq:Gammahat-def-clean}, one has
\[
\Area(\Gamma)\ <\ \Area(\Gamma^W_{A:B}).
\]
\end{lemma}

\begin{proof}
Define the surfaces
\[
X:=\Gamma^W_{A:B}\cup \mathcal{P},
\qquad
Y:=\Gamma\cup \mathcal{V}_{\mathrm{in}}.
\]
By construction, $X$ and $Y$ have the same support as sets (we have swapped $\mathcal{V}_{\mathrm{in}}$ and $\mathcal{P}$),
so
\[
X=Y.
\]
Taking areas and using that overlaps occur only at finitely many junction points (measure zero), we obtain
\begin{equation}\label{eq:area-swap-identity-clean}
\Area(\Gamma^W_{A:B}) + \Area(\mathcal{P})
\;=\;
\Area(\Gamma) + \Area(\mathcal{V}_{\mathrm{in}}).
\end{equation}
Thus, it suffices to show
\begin{equation}\label{eq:P<Vin-clean}
\Area(\mathcal{P})\ <\ \Area(\mathcal{V}_{\mathrm{in}}).
\end{equation}

To prove \eqref{eq:P<Vin-clean}, use RT minimality of $\gamma_A$ via a cut-and-reglue argument.
The loop $\mathcal{C}=\mathcal{V}_{\mathrm{in}}\cup\mathcal{P}$ bounds a bubble $U\subset\Int\mathcal{E}(A)$ with
$\partial_\infty U=\varnothing$ (Lemma~\ref{lem:C-simple-closed}). Excising $U$ from $\mathcal{E}(A)$ produces a new region
$\mathcal{E}(A)'=\mathcal{E}(A)\setminus U$ with the same asymptotic boundary $A$ but with bulk boundary
\[
\partial_\Sigma \mathcal{E}(A)'
=
(\gamma_A \setminus \mathcal{P})\ \cup\ \mathcal{V}_{\mathrm{in}}.
\]
Hence $(\gamma_A\setminus \mathcal{P})\cup\mathcal{V}_{\mathrm{in}}$ is a valid RT competitor for $A$, so
\begin{align}
\Area(\gamma_A)
&\le
\Area\!\Bigl((\gamma_A\setminus \mathcal{P})\ \cup\ \mathcal{V}_{\mathrm{in}}\Bigr)\notag\\
&=
\Area(\gamma_A) - \Area(\mathcal{P}) + \Area(\mathcal{V}_{\mathrm{in}}),
\end{align}
which implies $\Area(\mathcal{P})\le \Area(\mathcal{V}_{\mathrm{in}})$.

The inequality is strict: if $\Area(\mathcal{P})=\Area(\mathcal{V}_{\mathrm{in}})$ then the competitor would also be
RT-minimizing for $A$. Under the assumed uniqueness/no-multiple-intersections of minimizing geodesics on the slice
(see Assumption~\ref{thm: No multiple intersections of geodesics on an AdS Cauchy slice}),
this forces $\mathcal{V}_{\mathrm{in}}\subset\gamma_A$, contradicting $\mathcal{V}_{\mathrm{in}}\subset\Int\mathcal{E}(A)$.
Thus, \eqref{eq:P<Vin-clean} holds. Substituting \eqref{eq:P<Vin-clean} into \eqref{eq:area-swap-identity-clean} gives
$\Area(\Gamma)<\Area(\Gamma^W_{A:B})$ as claimed.
\end{proof}

\begin{corollary}[No cross-component excursions]\label{cor:no-ij-neq-clean}
There is no excursion $\alpha^\circ\subset\Int\mathcal{E}(A)$ with endpoints on two distinct components
$\gamma_i^A$ and $\gamma_j^A$ with $i\neq j$.
\end{corollary}

\begin{proof}
Assume such an excursion exists. Then Lemma~\ref{lem:cycle-extraction-clean} produces a cycle $\mathcal{V}$,
Lemma~\ref{lem:C-simple-closed} produces a bubble $U$ bounded by $\mathcal{C}$, Lemma~\ref{lem:Gammahat-admissible}
produces an EWCS-admissible $\Gamma$ by surgery, and Lemma~\ref{lem:strict-area-fixed} gives
$\Area(\Gamma)<\Area(\Gamma^W_{A:B})$, contradicting the minimality of $\Gamma^W_{A:B}$.
\end{proof}

\begin{proof}[Proof of Theorem \ref{thm:barrier} by cycles (2D)]
We show $\Gamma^W_{A:B}\cap \Int\mathcal{E}(A)=\varnothing$.
Suppose not. Then some component $\sigma\subset\Gamma^W_{A:B}$ admits a nontrivial excursion
$\alpha^\circ\subset\Int\mathcal{E}(A)$ with endpoints as in \eqref{eq:excursion-endpoints-clean}.

\smallskip
\noindent\textbf{Case 1: $i=j$.}
Then $\alpha$ and $\gamma_i^A$ are length-minimizing geodesics intersecting at two distinct bulk points.
By Assumption~\ref{thm: No multiple intersections of geodesics on an AdS Cauchy slice}, they must coincide between those points.
But $\gamma_i^A\subset\partial\mathcal{E}(A)$ contains no point of $\Int\mathcal{E}(A)$, contradicting
$\alpha^\circ\subset\Int\mathcal{E}(A)$.

\smallskip
\noindent\textbf{Case 2: $i\neq j$.}
This is excluded by Corollary~\ref{cor:no-ij-neq-clean}.

\smallskip
Hence no component of $\Gamma^W_{A:B}$ enters $\Int\mathcal{E}(A)$.
By symmetry under $A\leftrightarrow B$, the same holds with $A$ and $B$ exchanged, and the barrier theorem follows.
\end{proof}

\section{Proof: Monotonicity in unmeasured party}
\label{Holographic unmeasured monotonicity}
\dwmonotone*

For a pure state on $ABCD$ after purification, by expanding the target inequality in Theorem \ref{thm:dwmonotone} in terms of HEE and EWCS
(and relabeling subsystems), it is equivalent to
\begin{equation}\label{eq:unmeasured-mono-target}
S_{BC} + E_W(AB{:}C)\ \geq\ S_C + E_W(A{:}BC).
\end{equation}

\begin{lemma}[Discarding $A$-only components]\label{lem:discard-Aonly}
Decompose $A$ (with respect to the state on $ABC$) into $A$-only components $\mathcal{F}$ and $A$-bridged
components $\mathcal{A}$, so that $A=\mathcal{F}\cup\mathcal{A}$ and $\mathcal{F}\cap\mathcal{A}=\varnothing$.
Then, as in Remark~\ref{remark: bridged wedges},
\begin{equation}\label{eq:calA-reduction}
E_W(AB{:}C)=E_W(\mathcal{A}B{:}C),
\qquad
E_W(A{:}BC)=E_W(\mathcal{A}{:}BC).
\end{equation}
In particular, if $\mathcal{A}=\varnothing$ then \eqref{eq:unmeasured-mono-target} reduces to
\eqref{eq: simple inequality to prove}. Otherwise, one may replace $A$ by $\mathcal{A}$ and absorb
$\mathcal{F}$ into the purifier $D$ without changing \eqref{eq:unmeasured-mono-target}.
\end{lemma}

\begin{proof}
If $\mathcal{A}=\varnothing$, both EWCS terms in \eqref{eq:unmeasured-mono-target} are unchanged upon removing $A$,
and the inequality reduces to \eqref{eq: simple inequality to prove}.
If $\mathcal{A}\neq\varnothing$, then \eqref{eq:calA-reduction} shows that the EWCS terms depend only on $\mathcal{A}$;
thus $\mathcal{F}$ decouples and can be absorbed into the purifier $D$.
\end{proof}

Hereafter we implicitly use Lemma \ref{lem:discard-Aonly} assuming $A$ has no $A-$only components in $ABC$.

\begin{lemma}[Existence of an RT admissible surface for $C$ from $\gamma_{BC}$ and $\Gamma^W_{AB:C}$]\label{lem:UC-C-competitor}
Let $\Gamma^W_{AB:C}$ be the EWCS minimizer for $(AB{:}C)$, and let $\gamma_{BC}$ be the RT surface for $BC$.
Define the restriction of the EWCS to the $BC$ wedge by
\begin{equation}\label{eq:Gamma1Gamma2}
\Gamma_1 := \Gamma^W_{AB:C}\cap \mathcal{E}(BC),
\qquad
\Gamma_2 := \Gamma^W_{AB:C}\setminus \Gamma_1.
\end{equation}
Then $\Gamma_1$ separates $\mathcal{E}(BC)$ into two regions
\[
\mathcal{E}(BC)\setminus \Gamma_1 \;=\; U_B \sqcup U_C,
\qquad \partial_\infty U_B=B,\ \partial_\infty U_C=C.
\]
Let $U^\mathrm{cl}_C := U_C\cup \Gamma_1$ be the closure obtained by adjoining the open face along $\Gamma_1$.
Then $U^\mathrm{cl}_C$ is $C$-homologous in $\mathcal{E}(ABC)$, and its boundary
\begin{equation}\label{eq:E-def}
\mathbf{E}:=\partial U^\mathrm{cl}_C
\end{equation}
is a valid RT competitor for $C$. In particular,
\begin{equation}\label{eq:E-bounds-SC}
\frac{\Area(\mathbf{E})}{4G_N}\ \ge\ S_C.
\end{equation}
\end{lemma}

\begin{proof}
Since $\Gamma^W_{AB:C}$ is EWCS-admissible for $(AB{:}C)$, its restriction $\Gamma_1$ to $\mathcal{E}(BC)$
separates $\mathcal{E}(BC)$ into a $B$-side and a $C$-side region, giving $U_B,U_C$ as stated.
The set $U^\mathrm{cl}_C=U_C\cup\Gamma_1$ is $C$-homologous because it is bounded away from $AD$ by the inherited
portion of $\gamma_{BC}$ and bounded away from $B$ inside $\mathcal{E}(BC)$ by $\Gamma_1$.
Thus $\mathbf{E}=\partial U^\mathrm{cl}_C$ is a closed codimension-2 surface in $\Sigma$ homologous to $C$,
so is RT admissible.
\end{proof}

In the middle of Figure \ref{fig: example decomp monotinicity} we have $U^\mathrm{cl}_C$ is the region bounded by the red curves. 

\begin{lemma}[Bookkeeping decomposition of the LHS multisurface]\label{lem:LHS-bookkeeping}
Let $\gamma_{BC}$ be the RT surface for $BC$ and $\Gamma^W_{AB:C}$ the EWCS minimizer for $(AB{:}C)$.
Define the LHS multisurface
\begin{equation}\label{eq:L-def}
\mathbf{L} := \gamma_{BC}\ \oplus\ \Gamma^W_{AB:C},
\end{equation}
so that
\[
\frac{\Area^\oplus(\mathbf{L})}{4G_N} = S_{BC}+E_W(AB{:}C).
\]
Let $\mathbf{E}=\partial U^\mathrm{cl}_C$ be as in Lemma~\ref{lem:UC-C-competitor} and set
\begin{equation}\label{eq:Gamma-remnant}
\Gamma := \abs{\mathbf{L}\ominus \mathbf{E}}.
\end{equation}
Then $\mathbf{E}$ is supported on $\mathbf{L}$ and the areas split as
\begin{equation}\label{eq:area-split}
\Area^\oplus(\mathbf{L}) \;=\; \Area(\Gamma) + \Area(\mathbf{E}).
\end{equation}
Moreover, if $\gamma_{ABC}$ denotes the RT surface for $ABC$, the surface
\begin{equation}\label{eq:tildeGamma-def}
\tilde{\Gamma}
\;:=\;
\mathrm{cl}\left\{\Gamma \cap \Int\mathcal{E}(ABC)\right\},
\end{equation}
where $\mathrm{cl}$ is the standard topological closure,
is a codimension-2 surface in $\mathcal{E}(ABC)$ with $\partial \tilde{\Gamma}\subset \gamma_{ABC}$ and
\begin{equation}\label{eq:Gamma-to-tildeGamma}
\Area(\Gamma)\ \ge\ \Area(\tilde{\Gamma}).
\end{equation}
\end{lemma}

\begin{proof}
By construction, $\mathbf{E}$ consists of a sub-arc of $\gamma_{BC}$ together with $\Gamma_1\subset \Gamma^W_{AB:C}$,
hence $\mathbf{E}$ lies in the support of $\mathbf{L}=\gamma_{BC}\oplus\Gamma^W_{AB:C}$ and the multiset difference
$\mathbf{L}\ominus \mathbf{E}$ is well-defined.
Since overlaps occur only along codimension-2 junction sets (measure zero), areas satisfy \eqref{eq:area-split}.
Finally, $\tilde{\Gamma}$ is obtained by restricting $\Gamma$ to $\Int\mathcal{E}(ABC)$ and taking the closure in
$\mathcal{E}(ABC)$, which is just reattaching the codimension three boundaries of $\Gamma$ which we remove by taking the intersect with $\Int\mathcal{E}(ABC)$.  This intersect can only decrease area, giving \eqref{eq:Gamma-to-tildeGamma}.
\end{proof}

In the central panel of Figure \ref{fig: example decomp monotinicity} we had $\Gamma$ as the union of the green and blue curves; $\Tilde{\Gamma}$ is now just the blue curve as the green lies along $\partial\mathcal{E}(ABC)=\gamma_{ABC}$. The purpose of taking the closure here is to enforce $\partial\Tilde{\Gamma}\subset\gamma_{ABC}$ instead of $\varnothing$.

\begin{lemma}[$\tilde{\Gamma}$ is EWCS admissible for $(A:BC)$]\label{lem:path-sep-tildeGamma}
Assume $A$ has no $A$-only components in $ABC$.
Let $\tilde{\Gamma}$ be defined by \eqref{eq:tildeGamma-def}. Then, every simple path in $\mathcal{E}(ABC)$
joining a point adjacent to $A$ to a point adjacent to $BC$ intersects $\tilde{\Gamma}$.
Consequently, by Remark~\ref{rem:path-separation}, $\tilde{\Gamma}$ is EWCS-admissible for $(A{:}BC)$.
\end{lemma}

\begin{figure}[ht]
    \centering
\scalebox{1}{
\input{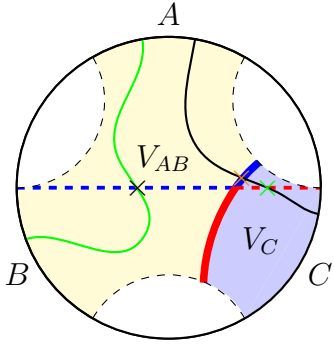}
}
    \caption{Continuation of case in Figure \ref{fig: example decomp monotinicity}. Taking $\mathcal{E}(ABC)\setminus\Gamma^W_{AB:C}$ yields two regions: $V_{AB}$ in yellow, and $V_C$ in blue. We have $U_C\subset V_C$ as the region bounded by the red curves. \emph{Case 1} shown by green path, and \emph{Case 2} by the black curve with their intersections point $x$ with $\gamma_{BC}$ marked by black and green $\times$. Intersection of black path and $\Gamma^W_{AB:C}$ as the orange $\times$. In either case the path crosses the blue curve $\Tilde{\Gamma}$.}
    \label{fig: example decomp monotonicity contd}
\end{figure}

\begin{proof}
Decompose $\tilde{\Gamma}$ into its portion lying on $\gamma_{BC}$ and its complement:
\[
\Gamma_{\mathrm{in}} := \tilde{\Gamma}\cap \gamma_{BC},
\qquad
\Gamma_{\mathrm{out}} := \tilde{\Gamma}\setminus \Gamma_{\mathrm{in}}.
\]
In the middle panel of Figure \ref{fig: example decomp monotinicity} we have $\Gamma_{\mathrm{in}}$ the thin blue curve and $\Gamma_{\mathrm{out}}$ the thick blue curve. By construction of $\mathbf{E}=\partial U^\mathrm{cl}_C$, the excised surface $\mathbf{E}$ consists of
(i) an arc $\eta_C\subset\gamma_{BC}$ and (ii) the segment $\Gamma_1=\Gamma^W_{AB:C}\cap \mathcal{E}(BC)$.
Hence the remaining part of $\Gamma^W_{AB:C} = \Gamma_1 \cup \Gamma_2$ inside $\mathcal{E}(ABC)$ but outside $\mathcal{E}(BC)$ is precisely
\begin{equation}\label{eq:Gammaout=Gamma2}
\Gamma_2 = \Gamma^W_{AB:C}\setminus \Gamma_1 = \Gamma_{\text{out}}.
\end{equation}
This is clearly seen for the case in Figure \ref{fig: example decomp monotinicity} with $\Gamma_{\mathrm{out}}=\Gamma_2$ the thick blue curve.

Now let $\ell$ be any simple path in $\mathcal{E}(ABC)$ from a point adjacent to $A$ to a point adjacent to $BC$.
Because $A$ has no $A$-only components in $ABC$, $\ell$ must cross $\gamma_{BC}$ in order to enter the $BC$-side.
Let $x$ be the \emph{first} intersection of $\ell$ with $\gamma_{BC}$ along the parameterization of $\ell$.

\smallskip
\noindent\emph{Case 1: $x\in \Gamma_{\mathrm{in}}$.}
Then $x\in\tilde{\Gamma}$ and this concludes the proof.

\smallskip
\noindent\emph{Case 2: $x\notin \Gamma_{\mathrm{in}}$.}
Then $x$ lies on the arc $\eta_C\subset\gamma_{BC}$ that was excised as part of $\mathbf{E}$, i.e.\ on the
$\gamma_{BC}$-portion of $\partial U^\mathrm{cl}_C$. Therefore immediately after crossing $x$, the path $\ell$
enters the region $U_C$ inside $\mathcal{E}(BC)$.

On the other hand, since $\Gamma^W_{AB:C}$ is EWCS-admissible for $(AB{:}C)$, it partitions $\mathcal{E}(ABC)$ into
two regions $V_{AB}$ and $V_C$ with $\partial_\infty V_{AB}=AB$ and $\partial_\infty V_C=C$.
By construction $U_C\subset V_C$ (it is the part of $V_C$ lying inside $\mathcal{E}(BC)$), while a neighborhood of $A$
lies in $V_{AB}$. Thus, as $\ell$ runs from $A$ towards $BC$, it must leave $V_{AB}$ and enter $V_C$ at some first
parameter value. Let $t_*$ be the first such parameter and set $y=\ell(t_*)$. Then $y\in \Gamma^W_{AB:C}$.

We claim $y\notin \Gamma_1$. Indeed, if $y\in\Gamma_1\subset\mathcal{E}(BC)$, then $y$ lies in $\mathcal{E}(BC)$.
But we already noted that immediately after the \emph{earlier} point $x$ (the first intersection with
$\gamma_{BC}$), the path lies in $U_C\subset V_C$. Hence $\ell$ must have entered $V_C$ no later than the parameter
of $x$, contradicting the choice of $t_*$ as the \emph{first} entry into $V_C$.
Therefore $y\in \Gamma^W_{AB:C}\setminus\Gamma_1=\Gamma_2$. Using \eqref{eq:Gammaout=Gamma2}, we conclude $y\in\Gamma_{\mathrm{out}}\subset \tilde{\Gamma}$, so $\ell$ intersects
$\tilde{\Gamma}$ also in Case 2. In all cases $\ell$ intersects $\tilde{\Gamma}$. The admissibility statement then follows from
Remark~\ref{rem:path-separation}.
\end{proof}

\begin{proof}[Proof of equation \eqref{eq:unmeasured-mono-target}]

By Lemma~\ref{lem:discard-Aonly} we may assume $A$ has no $A$-only components in $ABC$.
Let $\mathbf{L}$ be the LHS multisurface \eqref{eq:L-def}. Then
\[
\frac{\Area^\oplus(\mathbf{L})}{4G_N} = S_{BC}+E_W(AB{:}C).
\]

Let $\mathbf{E}=\partial U^\mathrm{cl}_C$ be the $C$-competitor from Lemma~\ref{lem:UC-C-competitor}, and define
$\Gamma$ and $\tilde{\Gamma}$ as in Lemma~\ref{lem:LHS-bookkeeping}. Using the bookkeeping identities
\eqref{eq:area-split} and \eqref{eq:Gamma-to-tildeGamma}, we obtain
\begin{equation}\label{eq:LHS-lower}
S_{BC}+E_W(AB{:}C)
=\frac{\Area^\oplus(\mathbf{L})}{4G_N}
=\frac{\Area(\Gamma)}{4G_N}+\frac{\Area(\mathbf{E})}{4G_N}
\ge \frac{\Area(\tilde{\Gamma})}{4G_N}+\frac{\Area(\mathbf{E})}{4G_N}.
\end{equation}

By Lemma~\ref{lem:UC-C-competitor}, $\Area(\mathbf{E})/(4G_N)\ge S_C$.
By Lemma~\ref{lem:path-sep-tildeGamma}, $\tilde{\Gamma}$ is EWCS-admissible for $(A{:}BC)$, hence by EWCS minimality
\begin{equation}\label{eq:tildeGamma-bounds-EW}
\frac{\Area(\tilde{\Gamma})}{4G_N}\ \ge\ E_W(A{:}BC).
\end{equation}
Substituting these bounds into \eqref{eq:LHS-lower} yields
\[
S_{BC}+E_W(AB{:}C)\ \ge\ E_W(A{:}BC)+S_C,
\]
which is exactly \eqref{eq:unmeasured-mono-target} hence proving Theorem \ref{thm:dwmonotone}.
\end{proof}

\section{Proof: Monogamy of classical correlations: Measured party}
\label{monogamy measured party}
\jwmonogamousmeasured*

\begin{proof}
Using the holographic expressions for $J_W$ in terms of HEE and EWCS and purifying $ABC$ by $D$,
the claim is equivalent to asking
\begin{equation}
E_W(A\!:\!BD)\;+\;E_W(A\!:\!CD)\ \ge\ S_A\;+\;E_W(A\!:\!D).
\end{equation}
Let $\Gamma^W_{A:BD}$ and
$\Gamma^W_{A:CD}$ denote EWCS minimizers for $(A:BD)$ and $(A:CD)$. Then introduce the \emph{LHS multisurface}
\[
\mathbf L\;:=\;\Gamma^W_{A:BD}\ \oplus\ \Gamma^W_{A:CD},
\qquad
\frac{\Area^\oplus(\mathbf L)}{4G_N}
\;=\;E_W(A\!:\!BD)+E_W(A\!:\!CD).
\]
\paragraph{Step 1: Carving an $\mathbf{A}$-homologous region $\mathbf{r_A}$.} By EWCS admissibility, each cross section splits its wedge into an $A$-side and a complementary side:
\begin{align*}
    \operatorname{Int}\bigl(\mathcal{E}(ABD)\bigr) \setminus \Gamma^W_{A:BD} 
    &= U_A^{(BD)}\;\sqcup\;U_{BD}^{(A)},
    &\partial_\infty U_A^{(BD)}&=A,\quad \partial_\infty U_{BD}^{(A)}=BD,\\[0.3em]
    \operatorname{Int}\bigl(\mathcal{E}(ACD)\bigr) \setminus \Gamma^W_{A:CD} 
    &= U_A^{(CD)}\;\sqcup\;U_{CD}^{(A)},
    &\partial_\infty U_A^{(CD)}&=A,\quad \partial_\infty U_{CD}^{(A)}=CD.
\end{align*}
We work with the open interiors so that $U_A^{(BD)}$ and $U_A^{(CD)}$ carry no boundary arcs along $\gamma_{ABD}$ and $\gamma_{ACD}$:
\begin{equation}
\partial_\Sigma U_A^{(BD)}=\partial_\Sigma U_A^{(CD)}=\varnothing.
\end{equation}
We then \emph{close} both $A$-sides in their respective wedges by adjoining the corresponding EWCS surfaces:
\begin{equation}
    \widehat{U}_A^{(BD)} := U_A^{(BD)} \cup \Gamma^W_{A:BD}, 
    \qquad
    \widehat{U}_A^{(CD)} := U_A^{(CD)} \cup \Gamma^W_{A:CD}.
\end{equation}
They are of course not generally closed in $\Sigma$ but by construction satisfy $\widehat{U}_A^{(BD)}\sim_{ABD}A$ and $\widehat{U}_A^{(CD)}\sim_{ACD}A$.

Now define
\begin{equation}
    r_A \;:=\; \widehat{U}_A^{(BD)} \cap \widehat{U}_A^{(CD)},
    \qquad
    U_A \;:=\; \widehat{U}_A^{(BD)} \cup \widehat{U}_A^{(CD)}.
\end{equation}
Thus $r_A\subseteq U_A$, and by the barrier property of the EWCS (Theorem~\ref{thm:barrier}) we have
\begin{equation}\label{eq: barrier in monogamy}
    \Gamma^W_{A:BD}\cap \operatorname{Int}\mathcal E(A)=\varnothing,
    \qquad
    \Gamma^W_{A:CD}\cap \operatorname{Int}\mathcal E(A)=\varnothing,
\end{equation}
so that the $A$-wedge lies on the $A$-side of both cross sections. In particular,
\begin{equation}\label{eq: ea in ra}
    \mathcal E(A)\;\subseteq\; r_A\;\subseteq\; U_A,
\end{equation}
with equality in the first inclusion only if one of the EWCS surfaces coincides with $\gamma_A$, and equality in the second one only if both do.

\begin{figure}[ht]
    \centering
\scalebox{1}{
\input{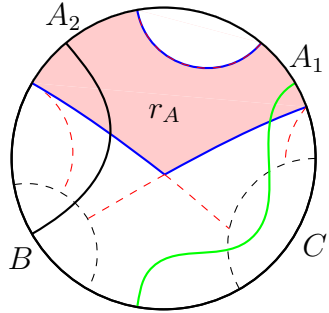}
}
    \caption{Continuation of case in Figure \ref{Monogamy classical unmeasured party} to show $r_A\sim A$, with $\partial r_A$ in blue. Green path shows $\ell$ intersects $\partial r_A$ when traveling from $A\to D$, and black path shows $\ell$ intersects $\partial r_A$ when traveling to $B$.}
    \label{fig: monogam unmeasured cont}
\end{figure}

\paragraph{Step 2: Showing $\mathbf{r_A\sim A}$.}

As we know $\partial_\infty r_A = A$, all that remains to show is that $\partial r_A$ is homologous to $A$. Equivalently any simple path
\[
    \ell\colon [0,1]\to\Sigma
\]
with $\ell(0)\in\operatorname{Int}\mathcal{E}(A)$ and $\ell(1)$ in 
$\operatorname{Int}\mathcal{E}(B)\cup\operatorname{Int}\mathcal{E}(C)\cup\operatorname{Int}\mathcal{E}(D)$
must intersect $\partial r_A$.

Define
\begin{equation}\label{eq: mathcal L definition}
\mathcal{L} \;:=\; \abs{\mathbf{L}}=\Gamma^W_{A:BD}\cup \Gamma^W_{A:CD}.
\end{equation}
Let us first assume for contradiction that there exists a simple path $\ell$ with 
\begin{equation}\label{eq: simple path conditions}
\ell(0)\in\operatorname{Int}\mathcal{E}(A)\subset r_A,\qquad
\ell(1)\in\operatorname{Int}\mathcal{E}(X),
\end{equation}
for some $X\in\{B,C,D\}$, such that $\ell$ does \emph{not} intersect $\mathcal{L}$. We will show this is a contradiction and so all such simple paths satisfying \eqref{eq: simple path conditions}  must intersect $\mathcal{L}$; we will then demonstrate that all first intersection points of $\ell$ and $\mathcal{L}$ lie on $\partial r_A$ hence showing $r_A \sim A$.

\medskip
\noindent\textbf{Paths from $\mathbf{A}$ to $\mathbf{D}$.}
First consider $\ell(1)\in\operatorname{Int}\mathcal{E}(D)$. Assume for contradiction $\ell$ misses $\Gamma^W_{A:BD}$. Then $\ell$ stays in $U_A^{(BD)}$ and can only exit this region along its open boundary on $\gamma_C$ into $\operatorname{Int}\mathcal{E}(C)$. This is the green path in Figure \ref{fig: monogam unmeasured cont}. If it does not have an open boundary on $\gamma_C$ then $\Gamma^W_{A:BD}$ coincides with $\gamma_A$, and we are done; thus we can assume it is open along this arc. As $\operatorname{Int}\mathcal{E}(C)\cap U_A^{(CD)}=\varnothing$ via the barrier theorem, in evolving from $\ell(0)\in \operatorname{Int}r_A$, it remains in $U_A^{(BD)}$ but must exit $U_A^{(CD)}$ such that it may reach $\operatorname{Int}\mathcal{E}(C)$. However, to exit $U_A^{(CD)}$ without crossing $\Gamma^W_{A:CD}$, $\ell$ can only achieve this by exiting through $U_A^{(CD)}$'s open arc along $\gamma_B$ into $\operatorname{Int}\mathcal{E}(B)$. 
Likewise if this region has no open arc along $\gamma_B$ then $\Gamma^W_{A:CD}$ coincides with $\gamma_A$, and we are also done; we thus continue assuming it has this open surface. However the barrier theorem tells us that $\operatorname{Int}\mathcal{E}(B)\cap U_A^{(BD)}=\varnothing$ so $\ell$ cannot remain in $U_A^{(BD)}$ whilst exiting $U_A^{(CD)}$ and so this is a contradiction. We argue identically for if $\ell$ misses $\Gamma^W_{A:CD}$ thus $\ell\cap \mathcal{L}\neq\varnothing$.

\medskip
\noindent\textbf{Paths from $\mathbf{A}$ to $\mathbf{B}$.}
The argument for $\ell(1)\in\operatorname{Int}\mathcal{E}(B)$ is analogous and is seen by the black path in Figure \ref{fig: monogam unmeasured cont}. Assume $\ell$ misses $\Gamma^W_{A:CD}$, and so remains in $U_A^{(CD)}$ until it crosses its open boundary on $\gamma_B$ directly into $\operatorname{Int}\mathcal{E}(B)$. However by assumption, it does not cross $\Gamma^W_{A:BD}$ either, and as $\operatorname{Int}\mathcal{E}(B)\cap U_A^{(BD)}=\varnothing$, it cannot reach $\operatorname{Int}\mathcal{E}(B)$ without first exiting $U_A^{(BD)}$ which it can only do along the open arc $\gamma_C$. But by assumption it cannot do this whilst remaining in $U_A^{(CD)}$, and hence not missing $\Gamma^W_{A:CD}$, and so we are done in this case. Assuming instead $\ell$ misses $\Gamma^W_{A:BD}$ we apply similar logic and find a likewise contradiction.

\medskip
\noindent\textbf{Paths from $A$ to $C$.} This is identical to the previous case with appropriate relabeling of $B\leftrightarrow C$.

\medskip
We have shown any simple path $\ell$ with $\ell(0)\in\operatorname{Int}\mathcal{E}(A)$ and $\ell(1)\in\operatorname{Int}\mathcal{E}(X)$, with $X\in\{B,C,D\}$ intersects $\mathcal{L}$. We will now show the first intersection point lies on $\partial r_A$. Let $t_1$ be the first parameter in $[0,1]$ such that $\ell(t_1)\in \mathcal{L}$. Without loss of generality, assume $\ell(t_1)\in\Gamma^W_{A:BD}$ (the case of $\Gamma^W_{A:CD}$ is identical). There are two possibilities:
\begin{itemize}
    \item If $\ell(t_1)\in \Gamma^W_{A:BD}\cap \widehat{U}_A^{(CD)}$, as this is a component of $\partial r_A$, we are done.
    \item If $\ell(t_1)$ lies on the remainder of $\Gamma^W_{A:BD}\setminus(\Gamma^W_{A:BD}\cap \widehat{U}_A^{(CD)})$, then $\ell$ must have exited $U_A^{(CD)}$. But this implies $\ell$ must have already crossed $\Gamma^W_{A:CD}$: if it had not then $\ell(t_1 - \epsilon)\in U_A^{(BD)}$ (with $\epsilon$ a small positive parameter), has not yet crossed $\Gamma^W_{A:BD}$ and so we could deform the path to instead travel to $\operatorname{Int}\mathcal{E}(C)$. However this is impossible without having crossed $\Gamma^W_{A:CD}$ by our previous proof. Thus $\ell$ crosses $\Gamma^W_{A:CD}$ before crossing $\Gamma^W_{A:BD}\setminus(\Gamma^W_{A:BD}\cap \widehat{U}_A^{(CD)})$ which is a contradiction to $t_1$ being the first such time that $\ell$ intersects $\mathcal{L}$.
\end{itemize}
We apply the same argument to assuming $ \ell(t_1)\in\Gamma^W_{A:CD}$ and have thus shown $\ell(t_1)\in\partial r_A$. Consequently no simple path from $\operatorname{Int}\mathcal{E}(A)$ to $\operatorname{Int}\mathcal{E}(X)$ can avoid $\partial r_A$ and so $r_A\sim_\Sigma A$. Hence by RT minimality we must have that
\begin{equation}\label{eq:SA-lower-fixed}
    \frac{\Area(\partial r_A)}{4G_N} \;\ge\; S_A.
\end{equation}

\paragraph{Step 3: Defining a candidate for $\mathbf{E_W(A:D)}$.}
Define the \emph{restricted remainder surface}\footnote{In fact, taking the support in the definition of $\Gamma$ is unnecessary since the resulting surface has multiplicity one everywhere; however, we will not use this fact in the proof. Hence, we will just take the support to ensure multiplicity one.}
\begin{equation}\label{eq:Gamma-def}
\Gamma:=\abs{(\ \mathbf L\ \ominus\ \partial r_A)}\cap\mathcal{E}(AD),
\end{equation}
and the \emph{restricted excised surface}
\begin{equation}\label{eq: restrictied excised surface}
    \mathbf{E} = \partial r_A\cap\mathcal{E}(AD).
\end{equation}

As $\Gamma$ is a restriction of the remainder multisurface we must have that
\begin{equation}\label{eq:area-split-inside-AD}
\frac{\Area(\Gamma)}{4G_N}+\frac{\Area(\partial r_A)}{4G_N}
\;\leq\;E_W(A\!:\!BD)+E_W(A\!:\!CD).
\end{equation}
We now verify that $\Gamma$ is EWCS admissible for $(A:D)$.

\begin{figure}[ht]
    \centering
\scalebox{1}{
\input{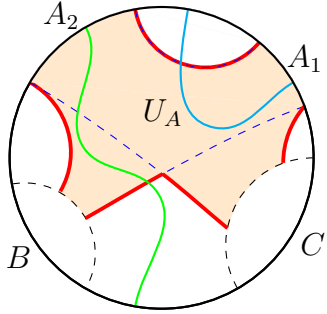}
}
    \caption{Continuation of case in Figure \ref{Monogamy classical unmeasured party} to show $\Gamma$ in red is EWCS admissible for $(A:D)$. Blue dashed curves show excise surfaced $\partial r_A$, and the union of red and blue dashed curves is $\mathcal{L}_{AD}$. Any simple path $\ell$ in $\mathcal{E}(AD)$ starting adjacent to $A$ and ending on $D$ intersects $\Gamma$.}
    \label{fig: monogam unmeasured contt}
\end{figure}

\paragraph{Step 4: Showing $\mathbf{\Gamma}$ is EWCS admissible for $\mathbf{(A:D)}$.}
To be admissible we need to show the following three properties: 
\begin{enumerate}[label=\emph{(\alph*)}, leftmargin=2em]
\item \emph{(Localization)} $\Gamma\subset\mathcal{E}(AD)$,
\item \emph{(Bulk anchoring)} $\partial\Gamma\subset \gamma_{AD}$,
\item \emph{(Separation)} $\mathcal E(AD)\setminus \Gamma\;=\; U_A\,\sqcup\,U_D$ \text{with} $\partial_\infty U_A=A,\ \ \partial_\infty U_D= D$.
\end{enumerate}
Generally, the localization condition is not required for admissibility (in the separation sense). However, a surface that satisfies the two latter conditions is never the global minimal surface: we could always amputate the said surface within $\mathcal{E}(AD)$ to produce a smaller surface still satisfying the other conditions. We do exactly this by taking the intersection in equation \eqref{eq: restrictied excised surface}.

\smallskip\noindent
\emph{(a) (Localization)} Trivial.

\smallskip\noindent
\emph{(b) (Bulk anchoring)} Inside \(\mathcal E(AD)\), each connected component of $\Gamma^W_{A:CD}\cap\;\mathcal{E}(AD)$ and $\Gamma^W_{A:BD}\cap\;\mathcal{E}(AD)$ is a smooth codimension-two surface whose bulk endpoints lie on \(\gamma_{AD}\). They of course cannot have a boundary in $\operatorname{Int}\mathcal{E}(AD)$ as they either anchor on $\gamma_B$ or $\gamma_C$, which both have zero intersect with $\operatorname{Int}\mathcal{E}(AD)$. Moreover, the set \(\mathbf E\;\cap\;\mathcal{E}(AD)\) is a union of sub-arcs of these components; each such sub-arc has endpoints either on \(\gamma_{AD}\) or at an intersection point of $\Gamma^W_{A:CD}$ and $\Gamma^W_{A:BD}$. These are the dashed blue curves in Figure \ref{fig: monogam unmeasured contt}. Excising \(\mathbf E\) from \(\mathbf L\cap\mathcal E(AD)\) thus removes only these sub-arcs, and any remaining endpoints lie on \(\gamma_{AD}\). Thus \(\partial\Gamma\subset \gamma_{AD}\). 

\smallskip\noindent
\emph{(c) (Separation)}
Now set 
\begin{equation}
    \mathcal{L}_{AD}:=\mathcal{L}\cap \mathcal E(AD),
\end{equation}
with $\mathcal{L}$ defined in equation \eqref{eq: mathcal L definition}. This just gives the original $\mathcal{L}$ if $BC$ is unbridged, as in Figure \ref{fig: monogam unmeasured contt}, but will generally be different if they are bridged.

We show that every simple path
\begin{equation}\label{eq: ell defini first hit}
    \ell: [0,1] \to \operatorname{Int}\mathcal E(AD), \qquad \ell(0)\in\operatorname{Int}\mathcal{E}(A),\qquad
\ell(1)\in\operatorname{Int}\mathcal{E}(D), 
\end{equation}
intersects $\mathcal{L}_{AD}$. We then show that any such $\ell$ intersects $\Gamma\subset\mathcal{L}_{AD}$ too. 

\emph{Claim 1.} Any such \(\ell\) intersects \(\mathcal{L}_{AD}\).
Indeed, if \(\ell\) missed \(\Gamma^W_{A:BD}\) it would lie entirely in \(U^{(BD)}_A\supseteq\operatorname{Int}\mathcal{E}(A)\); if it missed \(\Gamma^W_{A:CD}\) it would lie entirely in \(U^{(CD)}_A\).
Missing both forces \(\ell \subset U^{(BD)}_A\cap U^{(CD)}_A \subset r_A\). This cannot reach a point in $\operatorname{Int}\mathcal{E}(D)$ as $r_A\sim A$ by Step 2. Hence \(\ell\cap \mathcal{L}_{AD}\neq\varnothing\).

\emph{Claim 2 (first-hit upgrade).}
Let \(t_1\) be the first time \(\ell\) meets \(\mathcal{L}_{AD}\). If \(\ell(t_1)\notin \mathbf E\), then \(\ell(t_1)\in \mathcal{L}_{AD}\setminus \mathbf E\subset \Gamma\) and we are done. Moreover, if $\mathbf{E}$ and $\mathcal{L}_{AD}$ coincide then it is possible that  \(\ell(t_1)\in \mathbf E\cap\mathcal{L}_{AD}\). Regardless it intersects $\Gamma$ and we are done. This is the case of the cyan curve in Figure \ref{fig: monogam unmeasured contt}. If instead \(\ell(t_1)\in \mathbf E\) let us assume without loss of generality \(\ell(t_1)\in\Gamma^W_{A:BD}\) and  \(\ell(t_1)\notin\Gamma^W_{A:CD}\). This is the statement that the green path in Figure \ref{fig: monogam unmeasured contt} intersects $\mathcal{L}_{AD}$ first on the blue dashed line which is a segment of $\Gamma^W_{A:BD}$. Immediately after \(t_1\), \(\ell\) lies on the \(BD\)-side of \(\Gamma^W_{A:BD}\) but still on the \(A\)-side of \(\Gamma^W_{A:CD}\) (since \(t_1\) was the first contact with $\mathcal{L}_{AD}$). To reach a point adjacent to \(D\), \(\ell\) must cross \(\Gamma^W_{A:CD}\) at some later time \(t_2>t_1\).
At \(t_2\) it cannot be on \(\partial r_A\) (crossing \(\partial r_A\) would return into \(r_A\), contradicting that we remain on the \(BD\)-side of \(\Gamma^W_{A:BD}\)).
Therefore \(\ell(t_2)\in \Gamma^W_{A:CD}\setminus \mathbf E\subset \Gamma\). The same argument applies if instead $\ell(t_1)\in\Gamma^W_{A:CD}$ and  \(\ell(t_1)\notin\Gamma^W_{A:BD}\). If instead $\ell(t_1)\in\Gamma^W_{A:CD}\cap\Gamma^W_{A:BD}\in\partial r_A$, then due to multiplicity, this intersect is also in $\Gamma$ and we are done. Thus every simple path satisfying equation \eqref{eq: ell defini first hit} meets \(\Gamma\) and hence separation follows.

As we have shown \(\Gamma\) is EWCS-admissible for \((A{:}D)\), we can now state
\begin{equation}\label{eq:EW-lower}
    \frac{\Area(\Gamma)}{4G_N} \ \ge\ E_W(A{:}D).
\end{equation}

\paragraph{Step 5: Area bookkeeping.}
From \eqref{eq:area-split-inside-AD}, \eqref{eq:SA-lower-fixed}, and \eqref{eq:EW-lower},
\[
\frac{\Area^\oplus(\mathbf L)}{4G_N}
\ \ge\ \frac{\Area(\Gamma)}{4G_N} \ + \ \frac{\Area(\partial r_A)}{4G_N}
\ \ge\ E_W(A{:}D) \ + \ S_A,
\]
which is precisely our original inequality.
\end{proof}

\section{Proof: Monogamy of classical correlations: Unmeasured party}
\label{monogamy unmeasured party proof}
\jwmonogamousunmeasured*

\begin{proof}[Proof of Theorem \ref{thm:mono-unmeasured}]
Using the holographic expressions for $J_W$ in terms of HEE and EWCS and purifying $ABC$ by $D$,
Theorem \ref{thm:mono-unmeasured} is equivalent to
\begin{equation}\label{eq:target-ineq}
S_{AC}\;+\;E_W(A:CD)\;+\;E_W(C:AD)\ \ge\ S_A\;+\;S_C\;+\;E_W(AC:D).
\end{equation}
Throughout this proof we define $\gamma_{AC}$ as the RT surface for $AC$ and let $\Gamma^W_{A:CD}$ and $\Gamma^W_{C:AD}$ be the EWCS minimizers for $(A:CD)$ and $(C:AD)$.

\paragraph{Step 1: Bridged and unbridged components.} Let us first split $A$ and $C$ into $A$-only and $C$-only components in $AC$, along with their bridged components. Write $A=\mathcal{A}\cup\mathcal{F}_A$, where $\mathcal{A}$ are $A$-only components in $AC$ and $\mathcal{F}_A$ are $A$-bridged components in $AC$. Define $C=\mathcal{C}\cup\mathcal{F}_C$ similarly, and denote the union of the bridged components by
\begin{equation*}
    \mathcal{F}=\mathcal{F}_A\cup\mathcal{F}_C.
\end{equation*}
Then the HEE splits additively
\begin{equation}
    S_{AC} = S_\mathcal{A} + S_\mathcal{C} + S_\mathcal{F}, \qquad S_A = S_\mathcal{A} + S_{\mathcal{F}_A}, \qquad S_C = S_\mathcal{C} + S_{\mathcal{F}_C}.
\end{equation}
Hence we can cancel unbridged contributions leading to
\begin{equation}\label{eq: mono measured alt}
S_{\mathcal{F}} + E_W(A:CD) + E_W(C:AD) \stackrel{?}{\geq} S_{\mathcal{F}_A} + S_{\mathcal{F}_C} + E_W(AC:D).
\end{equation}
Note that we cannot replace $A$ and $C$ by $\mathcal{F}_A$ and $\mathcal{F}_B$ in $E_W$ terms as we know nothing about $A$- and $C$-only components in $ACD$. Now our goal is to prove \eqref{eq: mono measured alt} and hence the theorem.

\paragraph{Step 2: Carving out $\mathbf{\mathcal{F}_A}$ and $\mathbf{\mathcal{F}_C}$  homologous regions.} Let us first introduce the \emph{LHS multisurface}
\begin{equation}
\mathbf L\;:=\;\gamma_{\mathcal{F}}\ \oplus\ \Gamma^W_{A:CD}\ \oplus\ \Gamma^W_{C:AD},
\end{equation}
such that 
\begin{equation}
\frac{\Area^\oplus(\mathbf L)}{4G_N} = S_{\mathcal{F}} + E_W(A:CD) + E_W(C:AD),
\end{equation} 
and consider the restriction 
\begin{equation}
\Gamma_\mathcal{F} := (\Gamma^W_{A:CD}\cup \Gamma^W_{C:AD})\cap\mathcal E(\mathcal F).
\end{equation}
We then have
\begin{equation}
    \mathcal{E(F)} \setminus \Gamma_\mathcal{F} = U_{\mathcal{F}_A}\sqcup U_{\mathcal{F}_C} \sqcup U_{\text{extra}}, \qquad  \partial_\infty U_{\mathcal{F}_A} = \mathcal{F}_A, \quad \text{and} \quad \partial_\infty U_{\mathcal{F}_C} = \mathcal{F}_C,
\end{equation}
as the restriction converts the homology data on $A$ and $C$ to $\mathcal{F}_A$ and $\mathcal{F}_C$. $U_{\text{extra}}$ is the left over component of $\mathcal{E}(\mathcal{F})$ after this operation; this is the region bounded on four sides by the red, green, orange and blue curves in Figure \ref{Monogamy classical measured party}.

We can define their bulk closures as
\begin{equation}\label{eq: gluing to close}
    U^\mathrm{cl}_{\mathcal{F}_A} = U_{\mathcal{F}_A} \cup (\Gamma^W_{A:CD}\cap\mathcal{E(F)}), \quad \text{and} \quad U^\mathrm{cl}_{\mathcal{F}_C} = U_{\mathcal{F}_C} \cup (\Gamma^W_{C:AD}\cap\mathcal{E(F)}).
\end{equation}
This is well defined and produces codimension-one closed bulk subregions as by the barrier theorem $U_{\mathcal{F}_A}$ is only open along $\Gamma^W_{A:CD}\cap\mathcal{E(F)}$ and $U_{\mathcal{F}_C}$ is only open along $\Gamma^W_{C:AD}\cap\mathcal{E(F)}$; in particular $U_{\mathcal{F}_A}$ has no open arc\footnote{If however $\Gamma^W_{A:CD}\cap\Gamma^W_{C:AD}\neq\varnothing$ then the gluing is still well defined and still only uses a subset of $\mathbf{L}$ by multiplicity.} along $\Gamma^W_{C:AD}\cap\mathcal{E(F)}$ and so the gluing in equation \eqref{eq: gluing to close} does generate bounded regions.

Then as $U^\mathrm{cl}_{\mathcal{F}_A}$ is separated from $ABCD\setminus\mathcal{F}$ by $\gamma_\mathcal{F}$ and further separated from $\mathcal{F}_C$ by $\Gamma^W_{A:CD}\cap\mathcal{E}(\mathcal{F})$ then $U^\mathrm{cl}_{\mathcal{F}_A} \sim \mathcal{F_A}$. 
The same argument applies to $U^\mathrm{cl}_{\mathcal{F}_C}$ and so by RT minimality,
\begin{equation}\label{eq:SA_SC_upper}
\frac{\Area(\partial U^\mathrm{cl}_{\mathcal{F}_A})}{4G_N}\ \ge\ S_{\mathcal{F}_A},
\qquad
\frac{\Area(\partial U^\mathrm{cl}_{\mathcal{F}_C})}{4G_N}\ \ge\ S_{\mathcal{F}_C}.
\end{equation}
This is seen in Figure \ref{Monogamy classical measured party} for example.

\paragraph{Step 3: Building an admissible competitor $\mathbf{\Gamma}$ for $\mathbf{E_W(AC{:}D)}$.} 
Set the \emph{excised multisurface}
\[
\mathbf E\;:=\;\partial U^\mathrm{cl}_{\mathcal{F}_A}\ \oplus\ \partial U^\mathrm{cl}_{\mathcal{F}_C},
\]
which is just the portion of $\mathbf{L}$ we have used to define $\mathcal{F}_A$ and $\mathcal{F}_C$ homologous regions. 
Then decompose the remainder of $\mathbf L$ into an ``outer'' and an ``inner'' part:
\begin{align}
\Gamma_{\mathrm{out}}&:=\abs{\ \bigl(\Gamma^W_{A:CD}\ \oplus\ \Gamma^W_{C:AD}\bigr)\ \cap\ \bigl(\mathcal E(ACD)\setminus \mathcal{E(F)}\bigr)},\label{eq:Gamma-out}\\
\Gamma_{\mathrm{in}}&:=\ \abs{(\mathbf L\ominus \mathbf E)\ \cap\ \gamma_{\mathcal{F}}}\ ,\label{eq:Gamma-in}
\end{align}
and define the candidate
\begin{equation}\label{eq:Gamma-def-final}
\Gamma\ :=\ \Gamma_{\mathrm{out}}\ \cup\ \Gamma_{\mathrm{in}}.
\end{equation}
In the central panel of Figure \ref{Monogamy classical measured party} we have $\Gamma_{\text{in}}$ as the thick curves and $\Gamma_{\text{out}}$ the non-thick red and orange curves.

By construction,
\begin{equation}\label{eq:Gamma-support-boundary}
\Gamma\subset \mathcal E(ACD),\qquad
\Gamma_{\mathrm{out}}\cap\Gamma_{\mathrm{in}}=\varnothing,
\qquad
\text{and}\quad \partial_{\Sigma}\Gamma\subset \gamma_{ACD}=\gamma_B.
\end{equation}
The latter holds because the pieces on $\gamma_{\mathcal{F}}$ are glued to the outer legs at their endpoints, so $\gamma_{\mathcal{F}}$-contacts are \emph{interior} contacts of $\Gamma$, not part of $\partial\Gamma$. Such gluing points exist due to the multiplicity: for example if $\Gamma^W_{A:CD}$ intersects $\gamma_{\mathcal{F}}$ at point $x$ then, generally, one copy of $x$ is included for $\partial U^\mathrm{cl}_{\mathcal{F}_A}$ and the other in $\Gamma$, and so smooth gluing is permitted and no extra (albeit zero measure) sets are introduced in the measured surfaces. Hence the only bulk anchoring of $\Gamma$ lies on $\gamma_B$. We also have $\Gamma \oplus \mathbf{E} \;=\; \mathbf{L}$: $\partial U^\mathrm{cl}_{\mathcal{F}_A} \oplus \partial U^\mathrm{cl}_{\mathcal{F}_C}\subset \mathcal{E(F)}$ and so $\Gamma_{\text{out}}$ has zero intersect with $\mathbf{E}$, and $\Gamma_{\text{in}}$ also has zero intersect by definition. Thus we have
\begin{equation}\label{eq:additivity-L}
\frac{\Area^\oplus(\mathbf L)}{4G_N} = \frac{\Area(\Gamma)}{4G_N}+\frac{\Area^\oplus(\mathbf E)}{4G_N}
\;=\; S_{\mathcal{F}}+E_W(A:CD)+E_W(C:AD).
\end{equation}

We claim $\Gamma$ is EWCS admissible for $(AC:D)$. Condition (i) of Definition \ref{def: EWCS admissable} is satisfied by equation \eqref{eq:Gamma-support-boundary} and so it remains to check condition (ii).

\paragraph{Step 4: Showing $\mathbf{\Gamma}$ is an admissible competitor for $\mathbf{(AC{:}D)}$.}  We proceed using the standard path argument: every simple path in $\mathcal E(ACD)$ that joins a point adjacent to $AC$
to a point adjacent to $D$ intersects $\Gamma$. Consequently,
\begin{equation}\label{eq:Gamma-separates}
\mathcal E(ACD)\setminus \Gamma \;=\; U_{AC}\,\sqcup\,U_D,\qquad
\partial_\infty U_{AC}= AC,\ \ \partial_\infty U_D= D.
\end{equation}
Let $\ell$ be such a simple path with
$\ell(0)$ adjacent to $AC$ and $\ell(1)$ adjacent to $D$. We consider three cases according to where $\ell$ starts on the boundary.

\begin{figure}[ht]
    \centering
\scalebox{1}{
\input{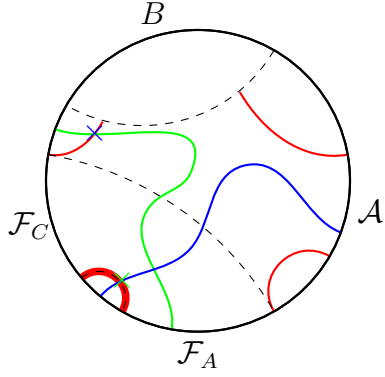}
}
    \caption{Continuation of case in Figure \ref{Monogamy classical measured party}.  \emph{Case 1} shown by green path, and \emph{Case 2} by the blue curve with their intersections with $\Gamma$ marked by $\times$. In either case the path crosses the red curve $\Gamma$ rendering $\Gamma$ EWCS admissible for $(AC:D)$.}
    \label{fig: monogam classical cont}
\end{figure}

\medskip
\noindent\emph{Case 1: $\ell(0)$ is adjacent to $\mathcal F$.} If $\ell$ begins at some point on $\mathcal{F}$ then it must exit the bridged wedge $\mathcal{E(F)}$ through $\gamma_{\mathcal{F}}$ to reach $D$. Let $x$ be this first contact point.

If $x\in \Gamma_{\mathrm{in}}$, then $x\in\Gamma$ and we are done. Otherwise
$x\notin\Gamma_{\mathrm{in}}$, so $x$ lies on an open arc
\[
I\subset\gamma_{\mathcal F}\setminus\Gamma_{\mathrm{in}},
\]
whose endpoints lie on $\Gamma^W_{A:CD}$ or $\Gamma^W_{C:AD}$.
This follows from the definition of $\mathbf E$: we have excised precisely those arcs of
$\gamma_{\mathcal F}$ that bound $U^\mathrm{cl}_{\mathcal F_A}$ and $U^\mathrm{cl}_{\mathcal F_C}$,
so each remaining arc $I$ is bounded by gluing points where outer legs of
$\Gamma^W_{A:CD}$ or $\Gamma^W_{C:AD}$ attach (see the green path Figure~\ref{fig: monogam classical cont}). But since $\ell$ starts on the $\mathcal{F}$-side and must reach the $D$-side, it must cross one of these outer legs.
However these legs are precisely $\Gamma_{\mathrm{out}}\subset\Gamma$, so $\ell$ intersects $\Gamma$ in either case. 

\medskip\noindent\emph{Case 2: $\ell(0)$ is adjacent to $\mathcal A$.}
If instead $\ell$ begins on $\mathcal{A}$, then there are two cases. If the components of $\Gamma^W_{A:CD}$ that generates an EWCS admissible component for $\mathcal{A}$ has zero intersect with $\mathcal{E(F)}$ then necessarily this surface is unaffected by the excision above and so remains in tact in $\Gamma$ (specifically it is contained within $\Gamma_{\mathrm{out}}$): thus $\ell$ must intersect this to reach $D$ by definition of EWCS. If instead this adjacent portion of $\Gamma^W_{A:CD}$ intersects $\gamma_{\mathcal{F}}$ then there is a path from $x$ to a point in $\mathcal{F}$ that nowhere crosses $\Gamma$. But we have proved above all $\ell$ beginning on $\mathcal{F}$ and ending on $D$ must cross $\Gamma$ and so paths beginning on $\mathcal{A}$ must also intersect $\Gamma$ to reach $D$.

\medskip\noindent\emph{Case 3: $\ell(0)$ is adjacent to $\mathcal C$.} Identical to case 2 with $\Gamma^W_{A:CD}\to\Gamma^W_{C:AD}$ and $\mathcal{A}\to\mathcal{C}$.

\noindent In all cases we have shown the intersection of such an $\ell$ proving the lemma.

From \eqref{eq:Gamma-support-boundary} and \eqref{eq:Gamma-separates}, $\Gamma$ satisfies the EWCS
admissibility for $(AC{:}D)$ (bulk anchoring on $\gamma_B$ and separation into regions homologous to $AC$ and $D$).
Therefore
\begin{equation}\label{eq:EW_ACD_upper}
\frac{\Area(\Gamma)}{4G_N}\ \ge\ E_W(AC:D).
\end{equation}

\paragraph{Step 5: Area bookkeeping and conclusion.}
By \eqref{eq:additivity-L}, \eqref{eq:SA_SC_upper}, and \eqref{eq:EW_ACD_upper},
\[
\frac{\Area^{\oplus}(\mathbf E)}{4G_N}
=\frac{\Area(\partial r_A)}{4G_N}+\frac{\Area(\partial r_C)}{4G_N}
\ \ge\ S_{\mathcal{F}_A}+S_{\mathcal{F}_C},
\qquad
\frac{\Area(\Gamma)}{4G_N}\ \ge\ E_W(AC:D),
\]
whence \eqref{eq:target-ineq} and the theorem follows.
\end{proof}

\section{Proof: One-way strong superadditivity of classical correlations}
\label{SSA formal proof}
\jwssa*

This proof is brief in explanation as it is very similar to the proof of Theorem \ref{thm:mono-unmeasured} (see Appendix \ref{monogamy unmeasured party proof}).

\begin{proof}
For a mixed state on $ABCD$, we expand~\eqref{eq: one way classical ssa} in terms of HEE and EWCS to obtain
\[
S_{AB} + E_W(A:BDE) + E_W(B:ACE) \stackrel{?}{\geq} S_A + S_B + E_W(AB:E).
\]
Let $\gamma_{AB}$ be the RT surface for $AB$ and let
$\Gamma^W_{A:BDE}$ and $\Gamma^W_{B:ACE}$ be the EWCS minimizers for $(A:BDE)$ and $(B:ACE)$.

\paragraph{Step 1: Bridged and Unbridged Components.} 
Let $\mathcal{A}$ and $\mathcal{B}$ be the components of $A$ and $B$ that are not bridged in $AB$, and let $\mathcal{F}_A = A\setminus\mathcal{A}$ and $\mathcal{F}_B = B\setminus\mathcal{B}$ be the bridged components with $\mathcal{F} = \mathcal{F}_A\cup\mathcal{F}_B$. Consequently, the entropy splits additively: $S_{AB} = S_{\mathcal{A}}+S_{\mathcal{B}}+S_{\mathcal{F}}$. Moreover this implies we have $S_A = S_\mathcal{A} + S_{\mathcal{F}_A}$ and similarly for $B$ and so we can cancel unbridged entropic terms leading to
\begin{equation}
S_{\mathcal{F}} + E_W(A:BDE) + E_W(B:ACE) \stackrel{?}{\geq} S_{\mathcal{F}_A} + S_{\mathcal{F}_B} + E_W(AB:E),
\end{equation}
whose proof implies our theorem.

\begin{figure}[ht]
    \centering
\scalebox{1}{
\input{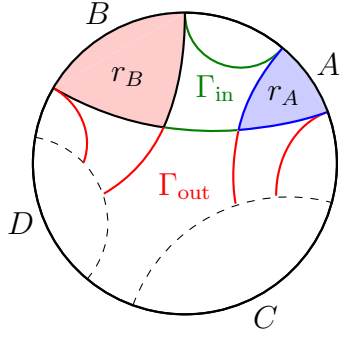}
}
    \caption{Continuation of case in Figure \ref{Strong superadditivtg figure} showing regions and surfaces defined in the proof. We have $\Gamma_{\text{in}}$ as the green curves, and $\Gamma_{\text{out}}$ as the red curves. Homologous regions to $A$ and $B$ shaded in blue and red respectively.}
    \label{fig: ssa labeled}
\end{figure}

\paragraph{Step 2: Carving out $\mathbf{\mathcal{F}_A}$ and $\mathbf{\mathcal{F}_B}$  homologous regions.} Restrict $\Gamma^W_{A:BDE}\cup\Gamma^W_{B:ACE}$ to $\mathcal E(\mathcal F)$. Its complement in $\mathcal{E(F)}$ has (at least two) components whose bulk
boundaries are homologous (within $\mathcal{E(F)}$) to $\mathcal{F}_A$ and to $\mathcal{F}_B$; denote the bulk closure of these regions by adjoining the relevant EWCS by $r_A$ and $r_B$.\footnote{If $\Gamma^W_{A:BDE}$ and $\Gamma^W_{B:ACE}$ overlap tangentially along an arc, use the two copies with their natural multiplicities to contribute one to $\partial r_A$
and the other to $\partial r_B$.}
Then $r_A\sim \mathcal{F}_A$ and $r_B\sim\mathcal{F}_B$ and so by RT minimality,
\begin{equation}\label{eq:SA_SB_upper_SSA}
\frac{\Area(\partial r_A)}{4G_N}\ \ge\ S_{\mathcal{F}_A},
\qquad
\frac{\Area(\partial r_B)}{4G_N}\ \ge\ S_{\mathcal{F}_B}.
\end{equation}

\paragraph{Step 3: Building an admissible competitor $\mathbf{\Gamma}$ for $\mathbf{E_W(AB:D)}$.} 
Introduce the \emph{LHS multisurface}
\[
\mathbf L\;:=\;\gamma_{\mathcal{F}}\ \oplus\ \Gamma^W_{A:BDE}\ \oplus\ \Gamma^W_{B:ACE},
\]
and set the \emph{excised multisurface}
\[
\mathbf E\;:=\;\partial r_A\ \oplus\ \partial r_B,
\]
Decompose the remainder of $\mathbf L$ into an ``outer'' and an ``inner'' part:
\begin{align}
\Gamma_{\mathrm{out}}&:=\ \abs{\bigl(\Gamma^W_{A:BDE}\ \oplus\ \Gamma^W_{B:ACE}\bigr)}\ \cap\ \bigl(\mathcal E(ABE)\setminus \mathcal{E(F)}\bigr),\\
\Gamma_{\mathrm{in}}&:=\ \abs{(\mathbf L\ominus \mathbf E)}\ \cap\ \gamma_{\mathcal{F}}\ ,
\end{align}
and define the candidate
\begin{equation}
\Gamma\ :=\ \Gamma_{\mathrm{out}}\ \cup\ \Gamma_{\mathrm{in}}.
\end{equation}
which by construction satisfies,
\begin{equation}
\Gamma\subset \mathcal E(ABE),\qquad
\Gamma_{\mathrm{out}}\cap\Gamma_{\mathrm{in}}=\varnothing,
\qquad
\text{and}\quad \partial_{\Sigma}\Gamma\subset \gamma_C\cup\gamma_D.
\end{equation}
Identically to Theorem \ref{thm:mono-unmeasured} it follows that $\Gamma \oplus \mathbf{E} = \mathbf{L}$. Moreover, by invoking the same separation lemma in the proof of the said theorem with appropriate replacements of system definitions it follows that $\Gamma$ satisfies the EWCS
admissibility for $(AB:E)$. Therefore
\begin{equation}\label{eq:EW_ACD_upper_SSA}
\frac{\Area(\Gamma)}{4G_N}\ \ge\ E_W(AB:E),
\end{equation}
and the theorem follows.
\end{proof}

\bibliographystyle{JHEP}
\bibliography{ref}

\end{document}